\documentclass[a4paper,12pt]{article}

\usepackage{url}
\usepackage{bm}
\usepackage{amsmath}
\usepackage{graphicx}
\usepackage{wrapfig}
\usepackage{amssymb}
\usepackage{verbatim}
\usepackage{slashbox}
\usepackage{enumitem}
\usepackage{etoolbox}
\usepackage[hidelinks]{hyperref}
\newcommand\T{\rule{0pt}{2.4ex}}
\newcommand\B{\rule[-1.2ex]{0pt}{0pt}}

\newcommand{\BlackBox}{\rule{1.5ex}{1.5ex}}  

\newenvironment{proof}{\par\noindent{\bf Proof\
}}{\hfill\BlackBox\\[2mm]}

\newtheorem{theorem}{Theorem}
\newtheorem{lemma}{Lemma}

\newtheorem{proposition}{Proposition}
\newtheorem{corollary}{Corollary}
\newtheorem{definition}{Definition}

\newtheorem{innercustomthm}{Theorem}
\newenvironment{customthm}[1]
  {\renewcommand\theinnercustomthm{#1}\innercustomthm}
  {\endinnercustomthm}
\newtheorem{innercustomprop}{Proposition}
\newenvironment{customprop}[1]
  {\renewcommand\theinnercustomprop{#1}\innercustomprop}
  {\endinnercustomprop}
\newtheorem{innercustomlem}{Lemma}
\newenvironment{customlem}[1]
  {\renewcommand\theinnercustomlem{#1}\innercustomlem}
  {\endinnercustomlem}

\newtoggle{clr}
\toggletrue{clr}

\begin{document}
\title{Formation of Stable Strategic Networks with Desired Topologies
%\thanks{
%Cite this article as: 
%S. Dhamal and Y. Narahari, ``Formation of Stable Strategic Networks with Desired Topologies,'' Studies in Microeconomics, vol. 3, no. 2, pp. 158-213, 2015.
%%
%The original publication is available at \url{http://journals.sagepub.com/doi/pdf/10.1177/2321022215588873} %{journals.sagepub.com}
%%
%%A very preliminary, concise version of this paper,
%%{\em Forming networks of strategic agents with desired topologies}~\cite{dhamalwine},
%%is available at  
%%%\href{http://link.springer.com/chapter/10.1007/978-3-642-35311-6_39}
%%{link.springer.com}.
%}
}
 \author{Swapnil~Dhamal and Y.~Narahari \\
Department of Computer Science and Automation, \\ Indian Institute of Science, Bangalore, India. \\
\{swapnil.dhamal, hari\}@csa.iisc.ernet.in
}

\date{}

\maketitle

\hrule
~\\
\textbf{Cite this article as:} 
S. Dhamal and Y. Narahari, ``Formation of Stable Strategic Networks with Desired Topologies,'' Studies in Microeconomics, vol. 3, no. 2, pp. 158--213, 2015.
The original publication is available at \url{http://journals.sagepub.com/doi/pdf/10.1177/2321022215588873}
\\
\hrule
~\\

\begin{abstract} %max 200 words
Many real-world networks such as social networks consist of strategic agents. The topology of these networks often plays a crucial role in determining the ease and speed with which certain information driven tasks can be accomplished. Consequently, growing a stable network having a certain desired topology is of interest. Motivated by this, we study the following important problem: given a certain desired topology, under what conditions would best response link alteration  strategies adopted by strategic agents, uniquely lead to formation of a stable network having the given topology. This problem is the inverse of the classical network formation problem where we are concerned with determining stable topologies, given the conditions on the network parameters. We study this interesting inverse problem by proposing (1) a recursive model of network formation and (2) a utility model that captures key determinants of network formation. Building upon these models, we explore relevant topologies such as star graph, complete graph, bipartite Tur\'an graph, and multiple stars with interconnected centers. We derive a set of sufficient conditions under which these topologies uniquely emerge, study their social welfare properties, and investigate the effects of deviating from the derived conditions.
\end{abstract}


\section{Introduction}
\label{sec:intro_nfsc}

A primary reason for networks such as social networks to be formed is that every person (or agent or node)  gets certain benefits from the network. These benefits assume different forms in different types of networks. These benefits, however, do not come for free. Every node in the network has to  incur
a certain cost for maintaining links with its immediate neighbors or direct friends. 
This cost takes the form of time, money, or effort, depending on the type of network. 
Owing to the tension between benefits and costs, self-interested or rational nodes 
think strategically while choosing their immediate neighbors. 
A stable network that forms out
of this process will have a topological structure that is dictated by the individual
utilities and the resulting best response strategies of the nodes.

The underlying social network structure plays a key role in determining the dynamics of several processes 
%over the network 
such as, the spread of epidemics~\cite{ganesh2005effect} and the diffusion of information~\cite{jacksonbook}. 
This, in turn, affects the decision of which nodes should be selected to be vaccinated \cite{abbassi2011toward}, 
or to trigger a campaign so as to either maximize the spread of certain information \cite{kempe2003maximizing} %\cite{narayanam2010shapley} 
or minimize the spread of an already spreading misinformation \cite{budak2011limiting}. %\cite{premm2012influence}.
Often, stakeholders such as a social network owner or planner, who work with the networks so formed, would like the network to have a certain desirable topology to facilitate efficient handling of 
%knowledge management, information retrieval, and information diffusion 
information driven tasks using the network. Typical examples of these tasks include spreading certain information to nodes (information diffusion), extracting certain critical information from nodes (information extraction), enabling optimal communication among nodes for maximum efficiency (knowledge management), etc. If a particular topology is an ideal one for the set of tasks to be handled, it would be useful to orchestrate network formation in a way that only the desired topology emerges as the unique stable topology.

A network in the current context can be naturally represented as a 
graph consisting of strategic agents called {\em nodes} 
and connections 
%or friendships 
among them called {\em links}. 
Bloch and Jackson~\cite{bloch2006definitions} examine a variety of stability and equilibrium notions that have been used to study strategic network formation.
% of networks consisting of strategic agents.
%
Our analysis in this paper is based on the notion of {\em pairwise stability\/} which accounts for bilateral deviations arising from mutual agreement of link creation between two nodes, that Nash equilibrium fails to capture~\cite{jacksonbook}. Deletion is unilateral and a node can delete a link without consent from the other node. 
Consistent with the definition of pairwise stability, we consider that all nodes are homogeneous and they have global knowledge of the network (this is a common, well accepted assumption in the literature on social network formation~\cite{jacksonbook}).
%\footnote{This is a common, well accepted assumption in the literature on social network formation~\cite{jacksonbook}.
%}

Before we proceed further, 
we present two important definitions from the literature~\cite{jacksonbook} for ease of discussion.
Let $u_j(g)$ denote the utility of node $j$ when the network formed is $g$.

\begin{definition} %\citep{jacksonbook}
A network is said to be {\em pairwise stable} if it is a best response for a node not to delete any of its links and there is no incentive for any two unconnected nodes to create a link between them. So $g$  is pairwise stable if \\(a) for each edge $e = (i, j) \in g$, $u_i(g \backslash \{e\}) \leq u_i(g)$ and  $u_j(g \backslash \{e\}) \leq u_j(g)$, and\\
(b) for each edge $e' = (i, j) \notin g$, if $ u_i(g \cup \{e'\})>u_i(g) $, then $u_j(g \cup {e'})<u_j(g)$.
\end{definition}

\begin{comment} %SeptEdit
Sarangi and Gilles~\cite{sarangi2004sps} introduce a variation of pairwise stability called {\em strong pairwise stability}. 
A network is said to be strongly pairwise stable if it is a best response for a node not to delete any {\em subset} of its links and there is no incentive for any two unconnected nodes to create a link between them. 
The authors argue that since link deletion is a unilateral act, it should not be restricted to one link at a time.
However, when a node gets a chance to alter its links, it has to examine the power set of its current links in order to compute a subset of links to be deleted as its best response. 
%Though pairwise stable and strongly pairwise stable networks coincide for the {\em symmetric connections model}~\cite{jackson1996strategic}, this is not true for a general utility model. 
As nodes have limited computational resources, it is reasonable to consider pairwise stability as the equilibrium notion. 
%Furthermore, as computing best response becomes more tractable for the nodes when considering alteration of at most one link at a time, the analysis becomes more tractable for us.
%
\end{comment} %SeptEdit
\begin{comment} %SeptEdit
We recall another important property relevant to networks, namely, efficiency~\cite{jacksonbook}.
\end{comment} %SeptEdit

\begin{definition} %\citep{jacksonbook}
A network is said to be {\em efficient} if the sum of the utilities of the nodes in the network is maximal. So given a set of nodes $N$, $g$ is efficient if it maximizes $\sum_{j\in N} u_j(g)$, that is, for all networks $g'$ on $N$, $\sum_{j\in N} u_j(g) \geq \sum_{j\in N} u_j(g') $.
\end{definition}

Every network has certain parameters that influence its evolution process.
%In this paper, 
We refer to the tuple of values of these parameters as {\em conditions on the network}.
%
%(details in Appendix \ref{app:conditions_network}).
By conditions on a network, we mean a listing of the range of values taken by the various parameters that influence network formation, including the relations between these parameters.
For example, let $b_1$ be the benefit that a node gets from each of its direct neighbors, $b_2$ be the benefit that it gets from each node that is at distance two from it, and $c$ be the cost it pays for maintaining link with each of its direct neighbors. In real-world networks, it is often the case that $0 \leq b_2 \leq b_1$ and $c \geq 0$. The list of relations, say (1) $0 < b_2 < b_1$ and (2) $b_1-b_2 < c < b_1$, are the conditions on the network. Based on these conditions, the utilities of the involved nodes are determined, which in turn affect their (link addition/deletion) strategies, hence influencing the process of formation of that network.
Throughout this paper, we ignore enlisting trivial conditions such as $0 \leq b_2 \leq b_1$ and $c \geq 0$.

In general, the evolution of a real-world social network would depend on several other factors such as the information diffusing through the network~\cite{ehrhardt2006diffusion,zhang2013strategic}.
%In this paper, 
For simplicity, we make a well accepted assumption that the network evolves purely based on the conditions on it and does not
depend on any other factor.
%, conditioning on the conditions on the network.
%\footnote{Ehrhardt et al.~\cite{ehrhardt2006diffusion} examine a model of network co-evolution, where information diffusion takes place along the current network and reciprocally, network formation depends on the information profile.}

\section{Motivation}
\label{sec:motiv_nfsc}
One of the key problems addressed in the literature on social network formation is: 
given a set of self-interested nodes and a model of social network formation, 
which topologies are  stable and which ones are  efficient.
The trade-off between stability and efficiency
is a key topic of  interest and concern in the literature on social network formation~\cite{jackson2005survey,jacksonbook}.

This work focuses on the inverse problem, namely, 
given a certain desired topology, under what conditions would  best response 
(link addition/deletion) strategies played 
by self-interested agents, uniquely lead to the formation of a stable (and perhaps efficient)
network with that topology. 
The problem becomes important because networks, such as an organizational network of a global company, play an important role in a
variety of knowledge management, information extraction, and information diffusion tasks. The topology of these networks is one of the major factors that decides the ease and speed with which the above tasks can be accomplished. In short, a certain topology might serve the 
%business 
interests of the network owner better.
%
%(a note in Appendix \ref{app:multiple_networks}).

%\subsection{Multiple Networks on a Set of Nodes}
%\label{app:multiple_networks}
%
In social networks, in general, it is difficult to figure out what the desired topology is. Moreover, it is possible that the social network is being formed for more than one reason. It can, however, be argued that given a set of individuals, there may not exist a unique social network amongst them. For instance, there may exist several networks like friendship network, collaboration network, organizational network, etc. on the same set of nodes. 
Different networks have different cost and benefit parameters, for example, from a mutual connection, two nodes may gain more in collaboration network than in friendship network, and also pay more cost. 
Furthermore, in real-world networks, a link between two nodes in one network may influence the corresponding link in another network. The influence may be positive (friendship trust leads to business trust) or negative (time spent for maintaining link in one network may adversely affect the corresponding link in another network).
For simplicity, we consider these various networks to be formed independently of each other.
A way to look at the problem under consideration is that, we focus on one such network at a time and derive conditions so that it has the desired topology or structure. 

\begin{figure}[t!]
\begin{tabular}{p{5cm} p{5cm} p{5cm}}
%\hspace{-.5cm}
%\begin{minipage}{.16\textwidth}
\centering
\includegraphics[scale=0.5]{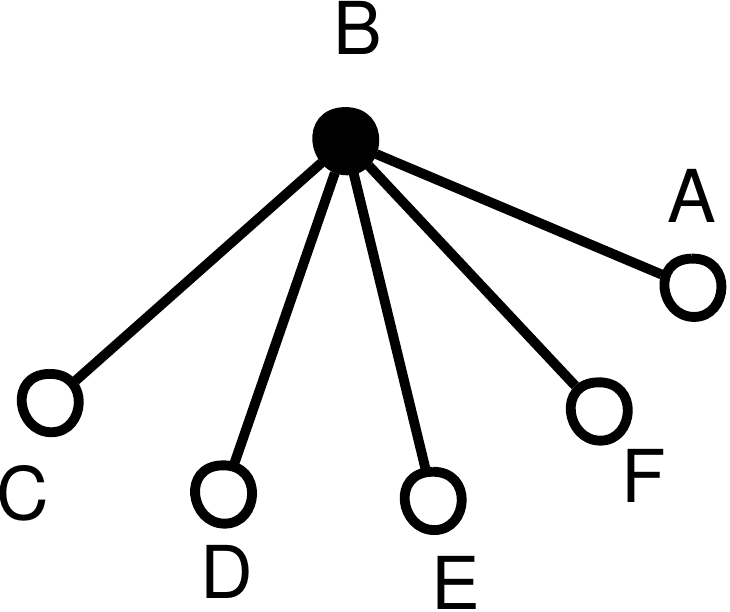}
%\label{fig:motiv_star}
%\end{minipage}
&
%\begin{minipage}{.10\textwidth}
\centering
\includegraphics[scale=0.5]{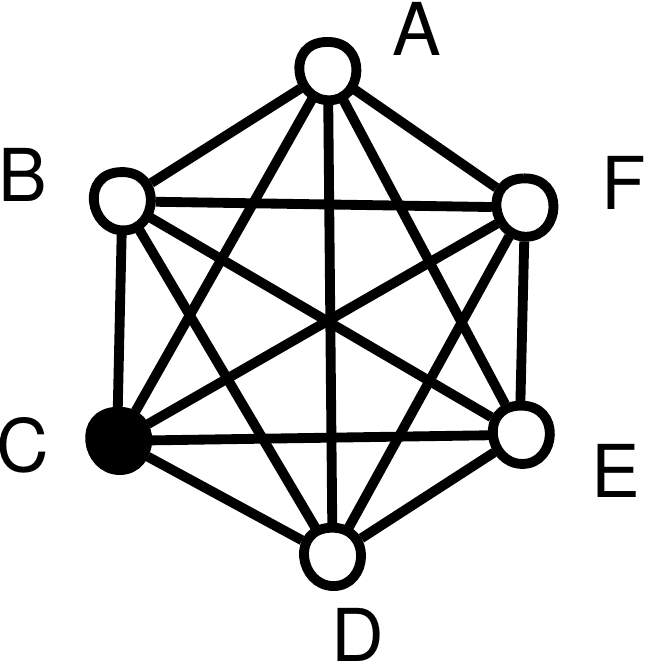}
%\label{fig:motiv_complete}
%\end{minipage}
&
%\begin{minipage}{.10\textwidth}
%\vspace{-.4cm}
%\centering
\hspace{7mm}
\includegraphics[scale=0.5]{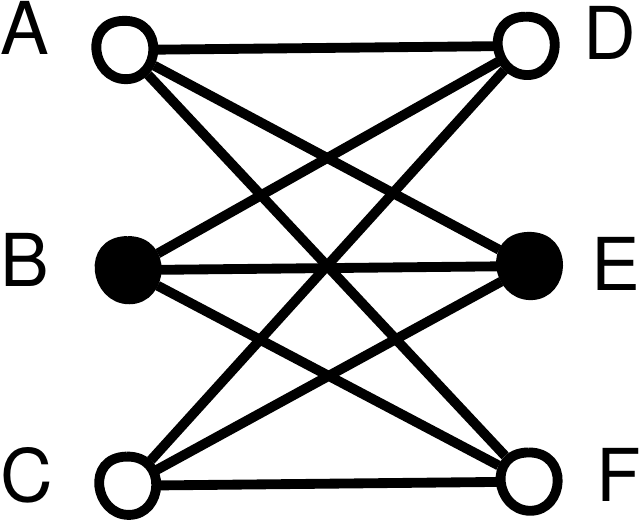}
%\label{fig:motiv_bipartite}
%\end{minipage}
\\
\centering (a) Star & 
\centering (b) Complete & 
\centering (c) Bipartite Tur\'an 
\vspace{5mm}
\end{tabular}

\begin{tabular}{p{7.5cm} p{7.5cm}}
%\begin{minipage}{.24\textwidth}
\centering
%\hspace{-.3cm}
\includegraphics[scale=0.5]{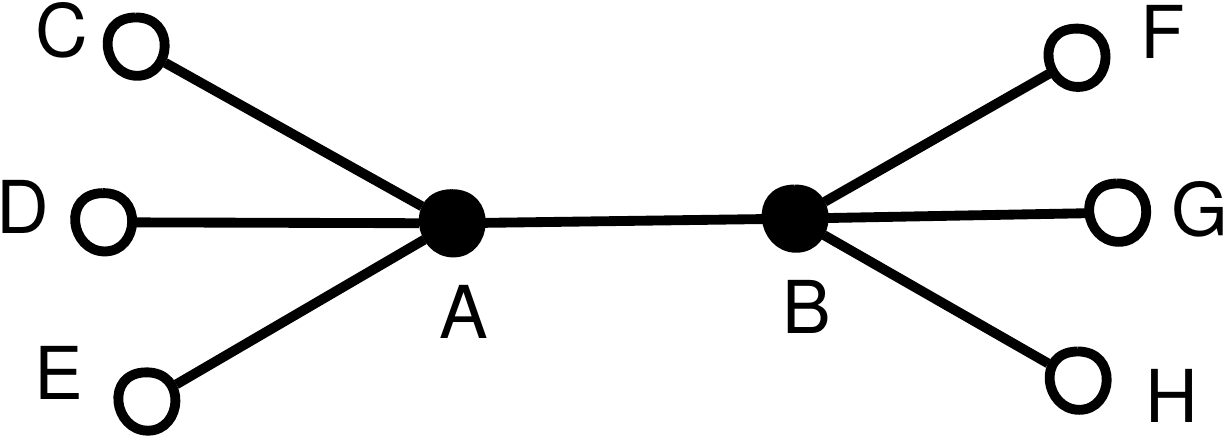}
%\label{fig:motiv_2star}
%\end{minipage}
&
%\begin{minipage}{.2\textwidth}
%\centering
\hspace{8.7mm}
%\vspace{.1cm}
%\hspace{-.1cm}
\includegraphics[scale=0.5]{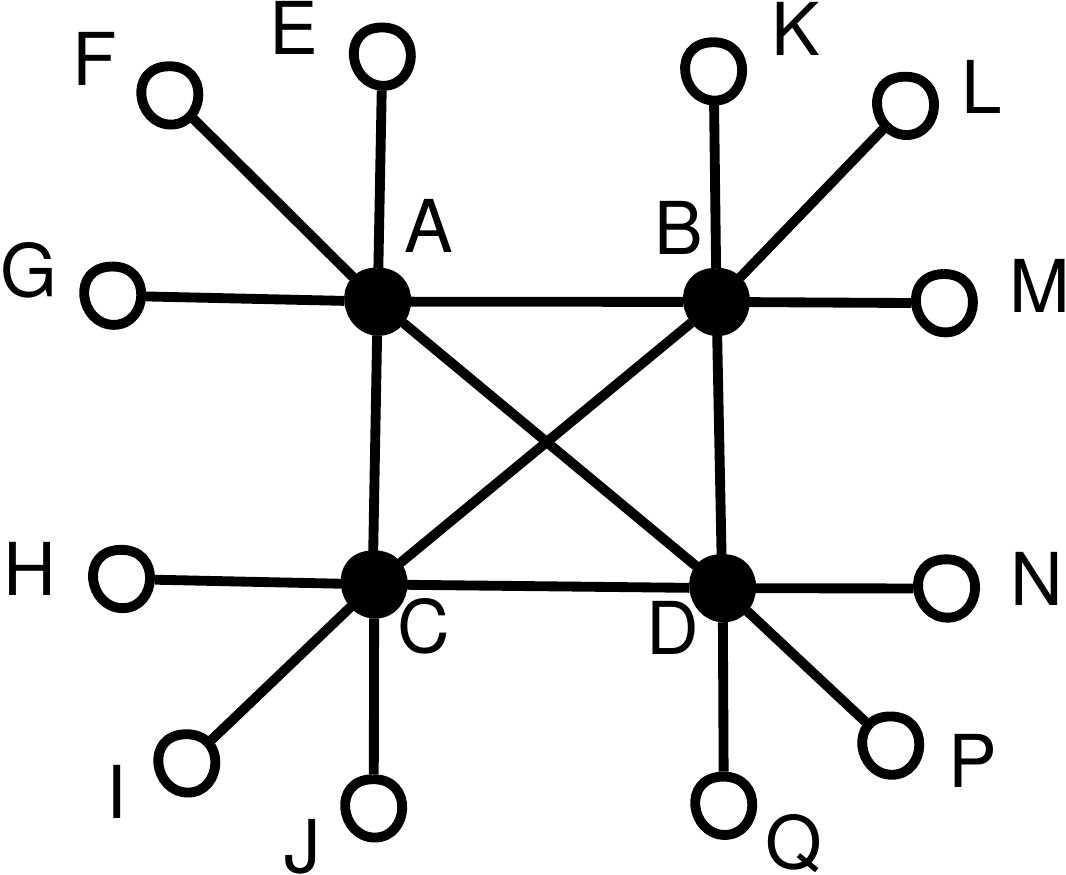}
%\label{fig:motiv_kstar}
%\end{minipage}
%\vspace{.2cm}
\\ 
\centering (d) 2-star & 
\centering (e) $k$-star ($k=4$)
\end{tabular}

\caption{Relevant topologies investigated in this work}
\label{fig:motiv_nfsc}
%\vspace{-.25cm}
\end{figure}

In this paper, for the sake of brevity, we consider only a representative set of commonly encountered topologies for investigation. However, our approach is general and can be used to study other topologies, albeit with more involved analysis.
We motivate our investigation further with the help of several 
%examples of 
relevant topologies shown
in Figure~\ref{fig:motiv_nfsc}.

Consider a network where there is a need to rapidly spread certain critical information, requiring redundant ways of communication to account for any link failures. 
The information may be received by any of the nodes and it is important that all other 
nodes also get the information at the earliest.
%Also owing to the criticality of the information, it is desirable that there are redundant ways of communication, to take care of any link failures. 
In such cases, a complete network (Figure~\ref{fig:motiv_nfsc}(b)) would be ideal. 
In general, if the information received by any node is required to be propagated throughout the network within a certain number of steps $d$, the network's diameter should be bounded by the number $d$.

Consider a different scenario where the time required to spread the information is critical, but
there is also a need for moderation to verify the authenticity of the information before 
spreading it to the other nodes in the network (for example, it could be a rumor).
 Here a star network (Figure~\ref{fig:motiv_nfsc}(a)) would be desirable since the center would act as a moderator and any 
information that originates in any part of the network has to flow through the
moderator before it reaches other nodes in the network. 
Virus inoculation is a related example where a star network would be desirable 
since vaccinating the center may be sufficient to prevent spread of the virus
to other parts of the network, thus reducing the cost of vaccination.

Our next example concerns  two sections
 of a society where some or all members of a section
receive certain information simultaneously. The objective here is to forward
the information to the other section. Moreover, it is desirable to not have intra-section links to save on resources. In this case, it would be desirable to have 
a bipartite network. Moreover, if the information is critical and urgent, 
requiring redundancy, a complete bipartite network would be desirable. 
A bipartite Tur\'an network (Figure~\ref{fig:motiv_nfsc}(c)) is a practical special case where both sections
are required to be nearly of equal sizes.

Consider a generalization of the star network where there are $k$ centers and 
the leaf nodes are evenly distributed among them, that is, the difference between the number of leaf nodes connected to any two centers, is at most one. Such a network would be desirable 
when the number of nodes is expected to be very large and there is a need for 
decentralization for efficiently controlling information in the network. 
We call such a network, $k$-star network (Figures~\ref{fig:motiv_nfsc}(d-e)).

For similar reasons, if fast information extraction is the main criterion, certain topologies may be better than others. Information extraction in social networks can be thought of as the reverse of information diffusion. Also, an information extraction or search algorithm would work better on some topologies than others.

The problem under study also assumes importance in knowledge management. McInerney~\cite{mcinerney2002knowledge} defines knowledge management as an effort to increase useful knowledge within an organization, and highlights that the ways to do this include encouraging communication, offering opportunities to learn, and promoting the sharing of appropriate knowledge artifacts. An organization may want to develop a particular network within, so as to make the most of knowledge management. A complete network would be desirable if the nodes are trustworthy with no possibility of manipulation.
% or bias. 
For practical reasons, an organization may want nodes of different sections to communicate with each other and not within sections so that each node can aggregate knowledge received from nodes belonging to the other section, in its own way. A bipartite Tur\'an network would be desirable in such a case. Such a network may also be more desirable than the complete network in order to prevent inessential investment of time for communication within a section.

Similarly, for a variety of reasons, there may be a need to form networks having certain other structural properties.
%for instance, 
%having diameter less than a certain value, 
%such as nodes having similar degrees and so on.
So depending on the tasks for which the network would be used, a certain topology might
be more desirable than others. This provides the motivation for our work.

\section{Relevant Work}
\label{sec:relevant_nfsc}

Models of network formation in literature can be broadly classified as either simultaneous move models or sequential move models.
Jackson and Wolinsky~\cite{jackson1996strategic} propose a simultaneous move game model where nodes simultaneously propose the set of nodes with whom they want to create a link, and a link is created between any two nodes if they mutually propose a link to each other.
Aumann and Myerson~\cite{myerson20} provide a sequential move game model where nodes are farsighted, whereas
Watts~\cite{watts618} considers a sequential move game model where nodes are myopic. In both of these approaches and in any sequential network formation model in general, the resulting network is based on the ordering in which links are altered and owing to the assumed random ordering, it is not clear which networks would emerge.
 
 The modeling of strategic formation in a general network setting was first studied by
  Jackson and Wolinsky~\cite{jackson1996strategic} by proposing a utility model called {\em symmetric connections model}. This widely cited model, however, does not capture many key determinants involved in strategic network formation.
Since then, several utility models have been proposed in literature in the effort of capturing these determinants. 
 Jackson~\cite{jackson2003stability} reviews several such models in the literature and highlights that pairwise stable networks may not exist in some settings. 
 Hellmann and Staudigl \cite{hellmann2014evolution} provide a survey of random graph models and game theoretic models for analyzing network evolution. %OctEdit

Given a network, Myerson value~\cite{myerson1977graphs} gives an allocation to each of the involved nodes based on certain desirable properties.
 Jackson~\cite{jackson2005allocation} proposes a family of allocation rules that consider alternative network structures when allocating the value generated by the network to the individual nodes.
 Narayanam and Narahari~\cite{ramasuri1} investigate the topologies of networks formed with a generic model based on value functions and analyze resulting networks using 
 %one such value function, 
 Myerson value.
 %~\cite{myerson1977graphs}.
 %
 There have also been studies on stability and efficiency of specific networks such as R\&D networks \cite{konig2012efficiency}. %OctEdit
 Atalay \cite{atalay2013sources} studies sources of variation in social networks by extending the model in \cite{jackson2007meeting} by allowing agents to have varying abilities to attract contacts. % OctEdit
 %
 %Blume et al. \cite{blume2013network} study network formation in the presence of contagious risks such as financial, epidemic, etc. and provide bounds on the welfare of stable and efficient networks. %OctEdit

    Goyal and Joshi~\cite{goyal2006unequal} explore two specific models of network formation and arrive at circumstances under which networks exhibit an unequal distribution of connections across agents.
 Goyal and Vega-Redondo~\cite{goyal2007structural} propose a non-cooperative game model capturing bridging benefits wherein they introduce the concept of {\em essential nodes}, which is a part of our proposed utility model. Their model, however, does not capture the decaying of benefits obtained from remote nodes.
Kleinberg et al.~\cite{kleinberg2008strategic} propose a localized model that considers benefits that a node gets by bridging any pair of its neighbors separated by a path of length 2. Their model does not capture indirect benefits and bridging benefits that nodes can gain by being intermediaries between non-neighbors which are separated by a path of length greater than 2.
%
%\begin{comment} %SeptEdit
Under another localized model where a node's bridging benefits depend on its clustering coefficient, Vallam et al.~\cite{vallam2013topologies} study stable and efficient topologies.
%\end{comment} %SeptEdit

%
Hummon~\cite{hummon2000utility} uses agent-based simulation approaches to explore the dynamics of network evolution based on the symmetric connections model. 
Doreian~\cite{doreian2006actor}, given some conditions on a network, analytically arrives at specific networks that are pairwise stable using the same model. However, the complexity of analysis increases exponentially with the number of nodes and the analysis in the paper is limited to a network with only five nodes.
Some gaps in this analysis are addressed by
Xie and Cui~\cite{xie2008cost,xie2008note}.

Most existing models of social network formation assume that all nodes are present throughout the evolution of a network, thus allowing nodes to form links that may be inconsistent with the desired network.
For instance, if the desired topology is a star, 
%with certain conditions on the network, 
it is desirable to have conditions that ensure a link between two nodes, of which one would play the role of the center. But with the same conditions, links between other pairs would be created with high probability, leading to inconsistencies with the star topology.
Also, with all nodes present in an unorganized network, a random ordering over them in sequential network formation models adds to the complexity of analysis.
However, in most social networks, not all nodes are present from beginning itself. A network starts building up from a few nodes and gradually grows to its full form. Our model captures such a type of network formation.

There have been  a few approaches earlier to design incentives for nodes so that the resulting network is efficient.
Woodard and Parkes~\cite{woodard2003strategyproof} 
use mechanism design to
%design incentives 
ensure that the outcome is an efficient network. 
Mutuswami and Winter~\cite{mutuswami2002subscription} design a mechanism that ensures efficiency, budget balance, and equity. 
Though it is often assumed that the welfare of a network is based only on its efficiency, there are many
situations where this may not be true. A network
may not be efficient in itself, but it may be desirable for reasons
external to the network, as explained in Section~\ref{sec:motiv_nfsc}.

\section{Contributions of this Paper}
\label{sec:gameinbrief}
In this paper, we study the inverse of the classical network formation problem, that is, 
under what conditions would the desired topology uniquely emerge %be obtained as a pairwise stable network
when 
%self-interested 
agents adopt their best response strategies.  
Our specific contributions are summarized below.

\begin{itemize} 
\item We propose a recursive model of  network formation, with which we can guarantee that a network being formed retains a designated topology in each of its stable states. Our model ensures that, for common network topologies, the analysis can be carried out independent of the current number of nodes in the network and also independent of the upper bound on the 
number of nodes in the network.
%\item
 The utility model we propose captures most key aspects relevant to strategic network formation:
(a) benefits from immediate neighbors, (b) costs of maintaining links with immediate neighbors, (c) benefits from indirect neighbors, (d) bridging benefits, (e) intermediation rents, and (f) an entry fee  for entering the network. 
We then present our procedure for deriving sufficient conditions for the formation of a given topology as the unique one. (Section~\ref{sec:model})
\item Using the proposed models, we study common and important networks, namely, star network, complete network, bipartite Tur\'an network, and $k$-star network, 
%and network with a certain bounded diameter, 
and derive sufficient conditions under which these topologies uniquely emerge.
%\item 
We also investigate the efficiency (or social welfare) properties of the above network topologies. (Section~\ref{sec:analysis})
\item We introduce the concept of dynamic conditions on a network and study the effects of deviation 
from the derived sufficient conditions on the resulting network, using the notion of graph edit distance. 
%For this purpose, we use the notion of {\em graph edit distance}.
In this process, we develop a polynomial time algorithm for computing graph edit distance between a given graph and a corresponding $k$-star graph. % with same number of nodes as $g$.
(Section~\ref{sec:deviation})
\end{itemize}

To the best of our knowledge, this is the first detailed effort in investigating the problem of obtaining a desired topology uniquely in social network formation.

\section{The Model}
\label{sec:model}

We consider the
%game is played amongst 
process of formation of a network consisting of 
strategic nodes, where each node aims at maximizing its utility it gets from the network. 

\subsection{A Recursive Model of Network Formation}

The network consists of $n$ nodes at any given time, where $n$ could vary over time. 
The process starts with one node, whose only strategy is to remain in its current state. The strategy of the second node is to either (a) not enter the network or (b) propose a link with the first node. We make a natural assumption that in order to be a part of the network, the second node has to propose a link with the first node and not vice versa. 
%Also, for successful link creation, utility of the first node should not decrease. 
Based on the model under study, the first node may or may not get to decide whether to accept this link. 
If this link is created, the second node successfully enters the network. Following this, the network evolves to reach a stable state after which, 
the third node considers entering the network. 
%If it enters the network by successfully creating a link with one of the first two nodes, 
The third node can enter the network by successfully creating link(s) with one or both of the first two nodes. In this paper, we consider that at most one link is altered at a time, and so the third node can enter the network by successfully creating a link with exactly one of the already present nodes in the network. If it does, 
the network of these three nodes evolves. 
Once the network reaches a stable state, 
the fourth node considers entering the network, and
this process continues.
%, which thus results in the formation of a stable network of, say $n$ nodes. 
Note that in the above process, 
no node in the network of $n-1$ nodes can create a link with the newly entering $n^{th}$ node until the latter proposes and successfully creates a link 
%with one of the existing nodes in 
in order to enter
the network.
After the new node enters the network successfully, 
%nodes who get to make their move are chosen at random at all time and 
the network evolves 
until it reaches a stable state consisting of $n$ nodes. Following this, a new ${(n+1)}^{th}$ node considers entering the network and the process goes on recursively. The assumption that a node considers entering the network only when it is stable may seem unnatural in general networks, but can be justified in networks where entry of nodes can be controlled by a network administrator.
This recursive model is depicted in Figure~\ref{fig:model}.
Note that the model is not based on any utility model, network evolution model, or equilibrium notion.

\begin{figure} [!t]
\centering
\includegraphics[scale=0.42]{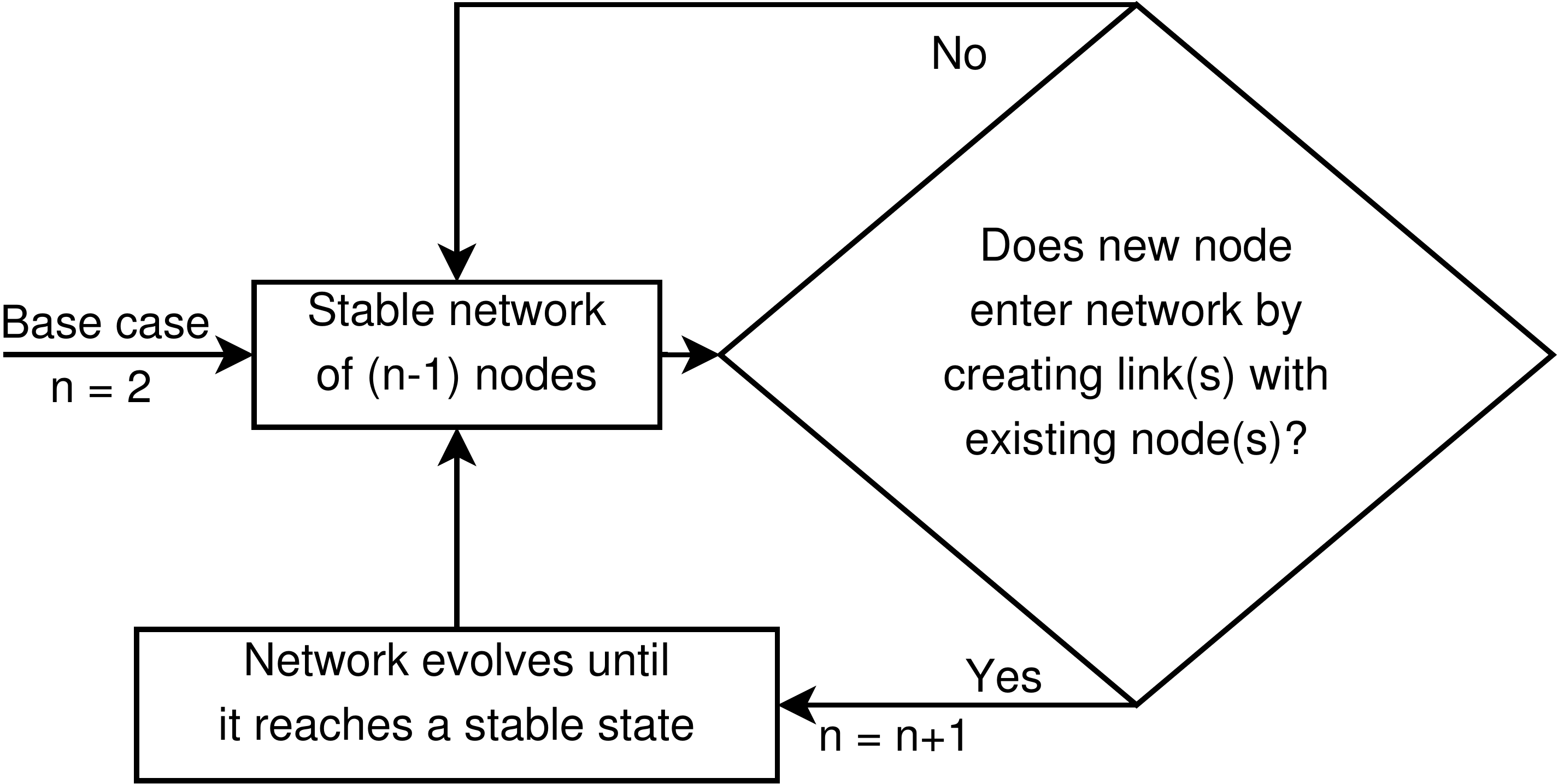}
\caption{
%RELOOK
Proposed recursive model of network formation}
\label{fig:model}
%\vspace{-.2cm}
%\vspace{-.25cm}
\end{figure}

%\subsection{Dynamic Conditions on a Network}
%\label{app:dynamic_conditions}
%
It can be observed at first glance that, if at some point of time, a new node fails to enter the network by failing to create a link with some existing node, the network will cease to grow. In such cases, it may seem that Figure~\ref{fig:model} 
goes into infinite loop for no reason, while it may have just pointed to an exit. The argument holds for the current social network models where the cost and benefit parameters, and hence the conditions on the network, are assumed to remain unchanged throughout the network formation process. 
But in real-world networks, this is often not the case and the conditions may vary over time or evolve owing to some internal or external factors. For instance, if the individual workload on the employees increases, the cost of maintaining link with each other also increases. On the other hand, if the workload is of collaborative nature, then the benefit parameters attain an increased value.
It is possible that no node successfully enters the network for some time, but with changes in the conditions, nodes may resume entering and the network may start to grow again. We explore this concept of {\em dynamic conditions} on a network in Section~\ref{sec:deviation}.

\subsection{Dynamics of Network Evolution}
\label{sec:directing}

The model of network evolution considered in this paper is based on a sequential move game \cite{watts618}.
During the evolution phase, nodes which get to make a move are chosen at random at all time. Each node has a set of strategies at any given time and when it gets a chance to make a move, it chooses its {\em myopic best response} strategy which maximizes its immediate utility. 
A strategy can be of one of the three types, namely (a) creating a link with a node that is not its immediate neighbor, (b) deleting a link with an immediate neighbor, or (c) maintaining status quo. 
Note that a node will compute whether a link it proposes, decreases utility of the other node, because if it does, it is not its myopic best response as the link will not be accepted by the latter. 
Moreover, consistent with the notion of pairwise stability, if a node gets to make a move and altering a link does not strictly increase its utility, then it prefers not to alter it.
%
%As the network evolution model we consider is based on a sequential move game, it can be represented as an extensive form game tree. 
%
The aforementioned sequential move evolution process can be represented as an extensive form game tree. 

\subsubsection{Game Tree}
\label{sec:tree}

The entry of each node in the network results in one game tree, and so the network formation process results in a series of game trees, each tree corresponding to a sequential move game
(see Figure~\ref{fig:star}).
\begin{comment} %SeptEdit
Each node of a game tree (not to be confused with a node of the network) represents a network state, while 
\end{comment} %SeptEdit
Each branch represents a possible transition from a network state, owing to decision made by a node.
So, the root of a game tree represents the network state in which a new node considers entering the network.

A way to find an equilibrium in an extensive form game consisting of farsighted players, is to use backward induction~\cite{osborne}. 
%Our game is a special case of such a game where 
However, in our game, 
the players have bounded rationality, that is, their best response strategies are myopic. So instead of the regular backward induction approach or the bottom-up approach, we take a top-down approach for ease of understanding.
%, which results in derivation of the same set of conditions under which a topology gets formed. 
%
We now recall the definition of an {\em improving path}~\cite{jackson2002evolution}.

\begin{definition}
\label{def:improving} %\citep{jackson2002evolution}
An {\em improving path} is a sequence of networks, where each transition is obtained by either any two nodes choosing to add a mutual link or any node choosing to delete any of its links. 
\end{definition}

Thus, a pairwise stable network is one from which there is no improving path leaving it.
%
%In the literature, some studies observe that agents are myopic \cite{pantz497}, while some others observe otherwise \cite{van2014individual}.
The notion of improving paths is based on the assumption of myopic agents, who make their decisions 
%of altering links 
without considering how their actions affect the decisions of other nodes and hence the evolution of the network. 
\begin{comment} %OctEdit
Though 
this process of improving paths exhibits bounded rationality, it is a natural variation on best response dynamics and has experimental justifications~\cite{pantz497}. 
\end{comment} %OctEdit
%So, deriving using the top-down approach leads to an 
%intuitive 
%understanding of the dynamics of network evolution using the notion of improving paths.

\subsubsection{Notion of Types}
\label{sec:types}
As the order in which nodes take decisions is random, in a general game, the number of branches arising from each state in the game tree depends on the number of nodes, $n$, as well as the number of possible direct connections each node can be involved in (or number of possible direct connections with respect to each node), $n-1$. 
The complexity of analysis can, however, be significantly reduced by the notion of {\em types} using which, several nodes and links can be analyzed at once. This is a widely used technique in analyzing pairwise stability of a network. 
%
%We formalize the same using the following two definitions.
%We explain the notion of types in more detail in the later part of the paper as well as in Appendix \ref{app:notion_types}.
%
We now explain the notion of types in detail.

%\subsection{Notion of Types}
%\label{app:notion_types}
%
\begin{definition}
\label{def:typenodes}
Two nodes $A$ and $C$ of a graph $g$ are of the same type if there exists an automorphism
$f:V(g)\rightarrow V(g)$ such that $f(A)=C$, where $V(g)$ is the vertex set of $g$.
\end{definition}
The implication of nodes being of the same type is that, for any automorphism $f$, if a best response strategy of node $A$ is to alter its link with node $D$, then a best response strategy of $f(A)$ is to alter its link with $f(D)$. So at any point of time, it is sufficient to consider the best response strategies of one node of each type.
\begin{definition}
\label{def:typeconnections}
Two connections with respect to a node $B$, connections $BA$ and $BC$, are of the same type if there exists an automorphism
$f$ such that $f(A)=C$ and $f(B)=B$.
\end{definition}
The implication of connections being of the same type with respect to a node is that, the node is indifferent between the connections, irrespective of the underlying utility model. Different types of connections with respect to a node form different branches in the game tree.

\begin{wrapfigure}{l}{75mm}
%\vspace{-.35cm}
  \centering
    \includegraphics[scale=0.45]{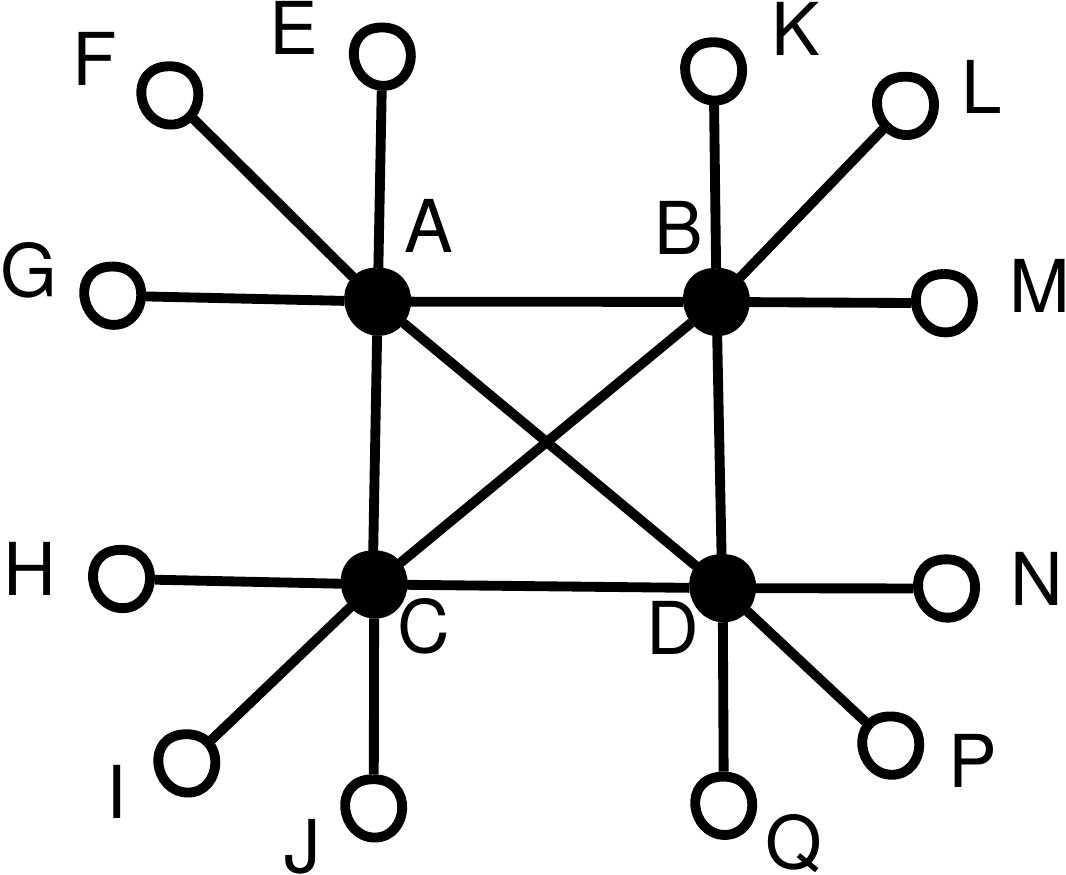} 
%  \vspace{-.15cm}
  \caption{A 4-star graph}
% \vspace{-.35cm}
  \label{fig:kstar}
\end{wrapfigure}

%\begin{figure}[h]
%%\vspace{-.35cm}
%  \centering
%    \includegraphics[scale=0.32]{motiv_kstar.pdf} 
%%  \vspace{-.15cm}
%  \caption{A 4-star graph}
%% \vspace{-.35cm}
%  \label{fig:kstar}
%\end{figure}
%
For example, in Figure~\ref{fig:kstar}, nodes $G$ and $H$ are of the same type. Also, the two possible connections $MG$ and $MH$ with respect to node $M$, are of the same type.
But the possible connections $EG$ and $EH$ with respect to node $E$, are not of the same type. So, these two strategies of node $E$, namely, connecting with node $G$ and connecting with node $H$, form different branches in the game tree, implying that the utilities arising from these two types of connections are not necessarily equal.

\subsubsection{Directing Network Evolution}
\label{directing_dymanics}

Our procedure for deriving sufficient conditions for the formation of a given topology as the unique topology, is modeled on the lines of {\em mathematical induction}. Consider a base case network with very few nodes (two in our analysis).
We derive conditions so that the network formed with these few nodes has the desired topology. Then using induction, we assume that a network with $n-1$ nodes has the desired topology, and derive conditions so that, the network with $n$ nodes, also has that topology. 
Without loss of generality, we explain this procedure with the example of star topology, referring to the game tree in Figure~\ref{fig:star}. 
%Conditions for base case need to be derived (we do this as part of actual derivation).
%deriving sufficient conditions under which the desired topology uniquely emerges). 
%Now 
Assuming that the network formed with $n-1$ nodes is a star, our objective is to derive conditions so that the network of $n$ nodes is also a star.

In Figure~\ref{fig:star}, at the root of the game tree, node $A$ is the newly entering $n^{th}$ node and the network is in state 0, where a star with $n-1$ nodes is already formed. 
%
%It is useful to recollect 
Recall that the complexity of analyzing a network
%in the sense that what all networks it can transit to, 
depends on the number of different types of nodes as well as the number of different types of possible connections with respect to a node in that network.
Note that in state 0, with respect to node $A$, there are two types of possible connections: (a) with the center and (b) with a leaf node. 
In states 1, 3, 4 and 5, there are two types of nodes, and two types of possible connections with respect to a leaf node and one with respect to the center. It will be seen that, the network is directed to not enter state 2, so even though there are four types of nodes in that state, it is not a matter of concern.

\begin{figure}[t]
\begin{tabular}{cc}
\hspace{1cm}
\begin{minipage}{3.7cm}
\vspace{-2.8in}
\centering
\includegraphics[scale=0.37]{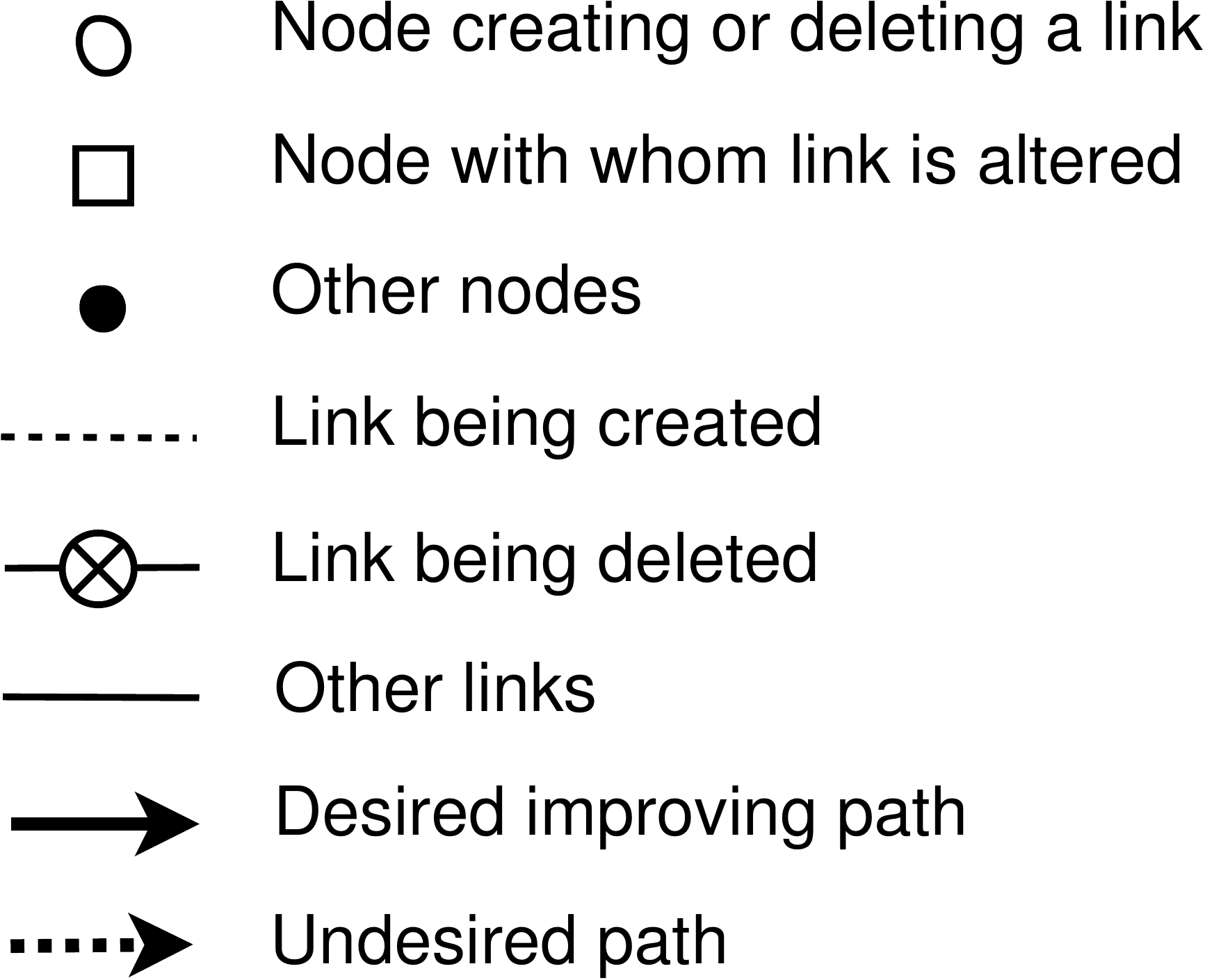}
\end{minipage}
&
\begin{minipage}{7cm}
\centering
\includegraphics[scale=0.37]{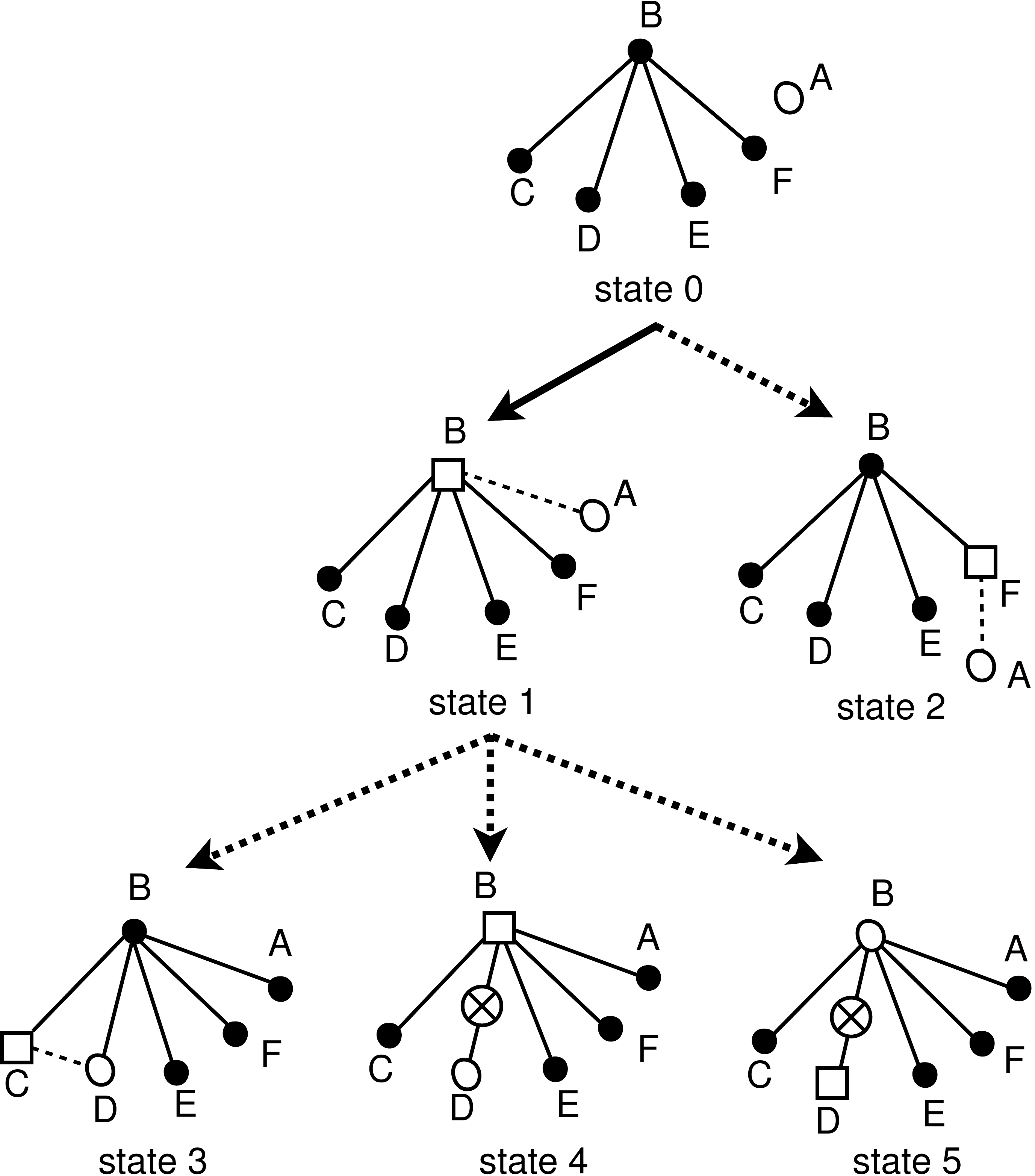}
\end{minipage}
\end{tabular}
\caption{Directing the network evolution for the formation of star topology uniquely}
\label{fig:star}
%\vspace{-.5cm}
\end{figure}

Let $u_j(s)$ be the utility of node $j$ when the network is in state $s$.
In state 0, as the newly entering node $A$ gets to make the first move,
we want it to connect to the center by choosing the improving path that transits from state 0 to state 1. So utility of node $A$ in state 1 should be greater than that in state 0, that is, $u_A(1) > u_A(0)$.
%Equation~(\ref{eqn:utility}) is one such utility function. 
Similarly, for node $B$ to accept the link from node $A$, $B$'s utility should not decrease, that is $u_B(1) \geq u_B(0)$. 
We do not want node $A$ to connect to any of the leaf nodes, that is, we do not want the network to enter state 2. Note that as we are interested in sufficient conditions, we are not concerned if there exists an improving path from state 2 that eventually results in a star (we discard state 2 in order to shorten the analysis). 
One way to ensure that the network does not enter state 2, irrespective of whether it lies on an improving path, is by making it less favorable for node $A$ than the desired state 1, that is, $u_A(2) < u_A(1)$. 
\begin{comment} %SeptEdit
That is, starting from state 0, there should exist an alternative improving path which gives node $A$ better utility than entering state 2. We say that, for node $A$, the strategy of creating a link with any of the leaf nodes is dominated by the strategy of creating link with the center. For state 2 to be dominated by state 1, $u_A(2) < u_A(1)$. 
\end{comment} %SeptEdit
Another way to ensure the same is by a condition for a leaf node such that, accepting a link from node $A$ decreases its utility, and so the leaf node does not accept the link, thus forcing node $A$ to connect to the center. That is, $u_j(2) < u_j(0)$ for any leaf node $j$. Thus the network enters state 1, which is our desired state.

To ensure pairwise stability of our desired state, no improving paths should lead out of it, for which we need to consider two cases. First, when node $B$ gets to make its move, it should not break any of its links (state 5),
%As we do not want the network to enter state 5, node $B$'s utility should be at least as good remaining in state 1 as entering 5, 
that is $u_B(1) \geq u_B(5)$.
Second, when any of the leaf nodes is chosen at random, it should neither create a link with some other leaf node (state 3), nor delete its link with the center (state 4).
%or (c) remain in the same state 1. As we want the network to stay in state 1, the additional 
The corresponding conditions are $u_j(1) \geq u_j(3)$ and $u_j(1) \geq u_j(4)$ for any leaf node $j$.

Thus we direct the network evolution along a desired improving path by imposing a set of conditions, ensuring that the resulting network is in the desired state or has the desired topology uniquely. 
In the evolution process of a network consisting of homogeneous nodes, the number of branches from a state of the game tree depends on the number of different types of nodes and the number of different types of possible connections with respect to a node, at that particular instant.
As we are primarily interested in the formation of special topologies in a recursive manner (nodes are already organized according to the topology and the objective is to extend the topology to that with one more node, so the existing nodes play the same role as before, and most or all of the existing links do not change), the number of different types of nodes as well as the number of different types of possible connections with respect to a node, are small constants at any instant, thus simplifying the analysis.

\subsection{The Utility Model}
\label{sec:utility}

Keeping in view the necessity of solving the problem in a setting that reflects real-world networks in a reasonably general way, we propose a utility model that captures several key determinants of social network formation.
In particular, our model is a considerable generalization of the extensively explored symmetric connections model~\cite{jackson1996strategic} and also builds upon other well known models in literature~\cite{goyal2007structural,kleinberg2008strategic}.
Furthermore, as nodes have global knowledge of existing nodes in the network while making their decisions (for instance, proposing a link with a faraway node), we propose a utility model that captures the global view of indirect and bridging benefits. 

\begin{definition} \cite{goyal2007structural}
\label{def:essential}
A node $j$ is said to be {\em essential} for nodes $y$ and $z$ if $j$ lies on every path joining $y$ and $z$.
% in the network. 
%That is, $j \in E(x,y)$, the set of essential nodes connecting $y$ and $z$.
\end{definition}

Whenever nodes $y$ and $z$ are directly connected, they get the entire benefits arising from the direct link. On the other hand, when they are indirectly connected with the help of other nodes, of which at least one is essential, $y$ and $z$ lose some fraction of the benefits arising from their communication, in the form of intermediation rents paid to the essential nodes without whom the communication is infeasible. 

Let $E(x,y)$ be the set of essential nodes connecting nodes $y$ and $z$.
The model proposed by Goyal and Vega-Redondo~\cite{goyal2007structural} suggests that the benefits produced by $y$ and $z$ be divided in a way that $x$, $y$, and the nodes in $E(x,y)$ get fraction $\frac{1}{|E(x,y)|+2}$ each.
%, where $e(x,y) = |E(x,y)|$. 
However, in practice, if nodes $y$ and $z$ can communicate owing to the essential nodes connecting them, that pair would want to enjoy at least some fraction of the benefits obtained from each other, since that pair is the real producer of these benefits (and possess human characteristics such as ego and prestige). That is, the pair would not agree to give away more than some fraction, say $\gamma$, to the corresponding set of essential nodes. As this fact is known to all nodes, in particular, to the set of essential nodes, they as a whole will charge the pair exactly $\gamma$ fraction as intermediation rents. As each essential node in the set is equally important for making the communication feasible, it is reasonable to assume that the intermediation rents are equally divided among them.

It can be noted that nodes which lie on every shortest path connecting $y$ and $z$, but are not essential for connecting them, also have bargaining power, since without them, the indirect benefits obtained from the communication would be less. And so, they should get some fraction proportional to their differential contribution, that is, the indirect benefits produced through the shortest path minus the indirect benefits produced through the second shortest path.
But, for simplicity of analysis, we ignore this differential contribution and assume that nodes that lie on path(s) connecting $y$ and $z$, but are not essential, do not get any share of the intermediation rents. So, when $y$ and $z$ are indirectly connected with the help of other nodes of which none is essential, they get the entire indirect benefits arising from their communication.

We now describe the determinants of network formation that our model captures, and thus obtain expression for the utility function.
Let $N$ be the set of nodes present in the given network,
$d_j$ be the degree of node $j$,
$l(j,w)$ be the shortest path distance between nodes $j$ and $w$,
$b_i$ be the benefit obtained from a node at distance $i$ in absence of rents (assume $b_\infty=0$),
and
$c$ be the cost for maintaining link with an immediate neighbor.

\begin{comment} %SeptEdit
Table~\ref{tab:notation} enlists the notation we use in the paper.

\begin{table}[t]
\caption{Notation for the Proposed Utility Model}
\label{tab:notation}
\centering
  \begin{tabular}{ p{.7cm}  l }
 \hline  \hline
\T \B 
$N$ & set of nodes present in the given network \\ \hline
\T \B 
$j$ & a typical node in the network \\ \hline
\T \B 
$u_j$ & net utility that node $j$ gets from the given network \\ \hline
\T \B 
$d_j$ & degree of node $j$ \\ \hline
\T \B 
$b_i$ & benefit from a node at distance $i$ in absence of rents \\ \hline
\T \B 
$c$ & cost for maintaining link with an immediate neighbor \\ \hline
\T \B 
$l(j,w)$ & shortest path distance between nodes $j$ and $w$ \\ \hline
\T \B 
$E(j,w)$ & set of nodes essential to connect $j$ and $w$ \\ \hline
%\T \B 
%$e(j,w)$ & cardinality of $E(j,w)$ \\ \hline
\T \B 
$\gamma$ & fraction of indirect benefit paid to corresponding set of \\ & essential nodes \\ \hline
\T \B 
$c_0$ & network entry factor (see 
%Section~\ref{sec:fee}: 
{\em Network Entry Fee}) \\ \hline
\T \B 
$\text{T}(j)$ & node to which node $j$ connects to enter the network \\ \hline
\T \B 
$\textbf{I}_{\{j=\text{NE}\}}$ & 1 when $j$ is a newly entering node about to create its \\ & first  link,  else 0 \\ 
\hline  \hline
  \end{tabular}
\end{table}
\end{comment} %SeptEdit

%\subsubsection{Network Entry Fee}
%\label{sec:fee}
\textbf{\textit{(1) Network Entry Fee:}} 
Since nodes enter a network one by one, we introduce the notion of {\em network entry fee}. 
This fee corresponds to some cost a node has to bear in order to be a part of the network. 
It is clear that, if a newly entering node wants its first connection to be with an existing node which is of high importance or degree, then it has to spend more time or effort. 
So we assume the entry fee that the former pays to be an increasing function of the latter's degree, say $d_{\text{T}}$. 
For simplicity of analysis, we assume the fee to be directly proportional to $d_{\text{T}}$ and call the proportionality constant, {\em network entry factor} $c_0$.

%

%\subsubsection{Direct Benefits}
%\label{sec:direct}
\textbf{\textit{(2) Direct Benefits:}}
These benefits are obtained from immediate neighbors  in a network.
For a node $j$, these benefits equal $b_1$ times $d_j$.

%\subsubsection{Link Costs}
%\label{sec:cost}
\textbf{\textit{(3) Link Costs:}}
These costs are the amount of resources like time, money, and effort a node has to spend in order to maintain links with its immediate neighbors.
For a node $j$, these costs equal $c$ times $d_j$.

%\subsubsection{Indirect Benefits}
%\label{sec:indirect}
\textbf{\textit{(4) Indirect Benefits:}}
These benefits are obtained from indirect neighbors, and these decay with distance $(b_{i+1} < b_i)$.
In the absence of rents, the total indirect benefits that a node $j$ gets is $\sum_{w \in N, \text{ } l(j,w)>1}{b_{l(j,w)}}$.

%\leavevmode

%\subsubsection{Intermediation Rents}
%\label{sec:rents}
\textbf{\textit{(5) Intermediation Rents:}}
Nodes pay a fraction $\gamma$ ($0 \leq \gamma < 1$) of the indirect benefits, in the form of additional favors or monetary transfers to the corresponding set of essential nodes, if any. The loss incurred by a node $j$ due to these rents is $\sum_{w \in N , \text{ } E(j,w)\neq \phi}{\gamma b_{l(j,w)}}$.

%\leavevmode

%\subsubsection{Bridging Benefits}
%\label{sec:bridging}
\textbf{\textit{(6) Bridging Benefits:}}
\begin{comment} %SeptEdit
In our model, a node gets bridging benefits for enabling communication between pairs of nodes which are otherwise disconnected.
Two nodes pay a fraction $\gamma$ of the indirect benefits to the set of essential nodes connecting them, which is assumed to be equally divided among the essential nodes connecting that pair. 
\end{comment} %SeptEdit
Consider a node $j \in E(y,z)$.
%be one of the essential nodes connecting two nodes $y$ and $z$.
Both $y$ and $z$ benefit $b_{l(y,z)}$ each and so this indirect connection produces a total benefit of $2 b_{l(y,z)}$. 
As described earlier, each node from the set $E(y,z)$ gets a fraction $\frac{\gamma}{|E(y,z)|}$, the absolute benefits being $\left( \frac{\gamma}{|E(y,z)|} \right) 2 b_{l(y,z)}$.
So the bridging benefits obtained by a node $j$ from the entire network is $\sum_{j \in E(y,z),\text{ }\{y,z\}\subseteq N}{ \left( \frac{\gamma}{|E(y,z)|} \right) 2 b_{l(y,z)} }$.

%\leavevmode

%\subsubsection{Utility Function}
%\label{sec:utilityfn}
\textbf{\textit{Utility Function:}}
The utility of a node $j$ is a function of the network, that is, $u_j:g \rightarrow \mathbb{R}$. We drop the notation $g$ from the following equation for readability.
Summing up all the aforementioned determinants of network formation that our model captures, we get
%the utility of node $j$ is 
%given by
%
\begin{equation}
%\hspace{-2cm}
\label{eqn:utility}
\begin{split}
u_j =& -c_0d_{\text{T}(j)}\textbf{I}_{\{j=\text{NE}\}} + d_j(b_1-c) +\sum_{\substack{w \in N \\l(j,w)>1}}{b_{l(j,w)}} \\
& - \sum_{\substack{w \in N \\E(j,w)\neq \phi}}{\gamma b_{l(j,w)}}  
    + \sum_{\substack{j \in E(y,z) \\ \{y,z\}\subseteq N}}{ \left( \frac{\gamma}{|E(y,z)|} \right) 2 b_{l(y,z)} } 
    \end{split}
\end{equation}
where 
$\text{T}(j)$ is the node to which node $j$ connects to enter the network, and
$\textbf{I}_{\{j=\text{NE}\}}$ is 1 when $j$ is a newly entering node about to create its first link, else it is 0.
%The way the utility of a node in a network is computed, will be described in detail in Section~\ref{sec:analysis}.

%
%\section{Formation of Relevant Topologies Uniquely}
\section{Analysis of Relevant Topologies}
\label{sec:analysis}

%In this section, we analyze the dynamics of formation of several relevant network topologies, namely star, complete graph, bipartite Tur\'an graph, 2-star, and $k$-star, and derive sufficient conditions under which these and only these topologies are formed.
%In this section, 
Using the proposed model of recursive and sequential network formation and the proposed utility model,
%we analyze some common network topologies using the proposed model of recursive and sequential network formation and the proposed utility model.
we provide sufficient conditions under which several relevant network topologies, namely star, complete graph, bipartite Tur\'an graph, 2-star, and $k$-star, uniquely emerge as pairwise stable networks.
Note that as the conditions derived for any particular topology are sufficient, 
% under the given setting, 
 there may exist alternative conditions that result in the same topology uniquely.

%\subsection{Sufficient Conditions for Formation of a Star Topology Uniquely}
\subsection{Sufficient Conditions for the Formation of Relevant Topologies Uniquely}

We use Equation~(\ref{eqn:utility}) for mathematically deriving the conditions.

\begin{proposition}
\label{thm:star}
For a network, if $b_1-b_2 + \gamma b_2 \leq c < b_1$ and $c_0 < \left( 1-\gamma \right) \left( b_2-b_3 \right)$, 
the unique resulting topology is star.
\end{proposition}
\begin{proof}
Refer to Figure~\ref{fig:star} throughout the proof. For the base case of $n=2$, the requirement for the second node to propose a link 
to the first is that its utility should become strictly positive. Also as the first node has degree $0$, there is no entry fee.
\begin{equation}
\label{B1for2} 
0 < b_1-c  \iff c < b_1
\end{equation}
Now, consider a star consisting of $n-1$ nodes. Let the newly entering $n^{th}$ node get to make a decision of whether to enter the network. For $n\geq3$, if the entering node connects to the center, it gets indirect benefits of $b_2$ each from $n-2$ nodes. But as the center is essential for enabling communication between newly entering node and other leaf nodes, the new node has to pay $\gamma$ fraction of these benefits to the center. Also, it has to pay an entry fee of $(n-2)c_0$ as the degree of center is $n-2$. So in Figure~\ref{fig:star}, 
$u_A(0) < u_A(1)$ gives
\begin{equation}
\nonumber
0 < b_1-c+(n-2) \left( 1-\gamma \right) b_2-(n-2)c_0
\end{equation}
\begin{equation}
\nonumber
\iff  c < b_1+(n-2) \left( \left( 1-\gamma \right) b_2-c_0 \right)
\end{equation}
As it needs to be true for all $n \geq 3$, we set the condition to
\begin{equation}
\nonumber
 c < \min_{n \geq 3} \Big\{ b_1+(n-2) \left( \left( 1-\gamma \right) b_2-c_0 \right) \Big\}
\end{equation}
\begin{equation}
\label{B1}
\Longleftarrow  c <b_1+  \left( 1-\gamma \right)  b_2-c_0
\end{equation}
The last step is obtained so that the condition for link cost is independent of the upper limit on the number of nodes, by enforcing
\begin{equation}
\label{B1forc0}
c_0 \leq  \left( 1-\gamma \right) b_2
\end{equation}
which enables us to substitute $n=3$ and the condition holds for all $n\geq 3$.\\
For the center to accept a link from the newly entering node, we need to have $u_B(0) \leq u_B(1)$.
For $n=2$, the requirement for the first node to accept link from the second node is $0 \leq b_1-c$ which is satisfied by Inequality~(\ref{B1for2}).
For $n=3$, as the center is essential for connecting the other two nodes separated by distance two, it gets $\gamma$ fraction of $b_2$ from both the nodes. So it gets bridging benefits of $2 \gamma b_2$.
\begin{equation}
\nonumber
 b_1-c \leq 2(b_1-c)+2\gamma b_2
\end{equation}
\begin{equation}
\nonumber
\iff c \leq b_1+ 2\gamma b_2
\end{equation}
This condition is satisfied by Inequality~(\ref{B1for2}).
For $n\geq 4$, prior to entry of the new node, the center alone bridged $\dbinom{n-2}{2}$
\normalsize pairs of nodes at distance two from each other, while after connecting with the new node, the center is the sole essential node for \small{$\dbinom{n-1}{2}$} 
\normalsize such pairs.
So the required condition:
%\begin{small}
\begin{equation}
\nonumber
 (n-2)(b_1-c)+\gamma \dbinom{n-2}{2}2 b_2 \leq (n-1)(b_1-c)+\gamma \dbinom{n-1}{2}2 b_2
\end{equation}
%\end{small}
This condition is satisfied by Inequality~(\ref{B1for2}) for all $n \geq 4$.\\
For the newly entering node to prefer the center over a leaf node as its first connection (not applicable for $n=2$ and $3$), we need $u_A(1) > u_A(2)$.
%\begin{small}
\begin{equation}
\nonumber
\begin{split}
 b_1-c +(n-2) \left( 1-\gamma \right) b_2 -(n-2)c_0  > b_1-c+ \left( 1-\gamma \right) b_2 -c_0 +(n-3) \left( 1-\gamma \right) b_3
 \end{split}
\end{equation}
%\end{small}
\begin{equation}
\label{B3a}
\iff c_0 <  \left( 1-\gamma \right) \left( b_2-b_3 \right)
\end{equation}
Alternatively, the newly entering node may want to connect to the leaf node, but the leaf node's utility decreases. In that case, the alternative condition can be $u_j(2)<u_j(0)$ for $j=C,D,E,F$.
Note that this leaf node gets bridging benefits of $2\gamma b_2$ for being essential for indirectly connecting the new node with the center. Also, as it is one of the two essential nodes for indirectly connecting the new node with the other $n-3$ leaf nodes (the other being the center), it gets bridging benefits of $(n-3) (\frac{\gamma}{2})2 b_3 = (n-3) \gamma b_3$.
%\begin{small}
\begin{equation}
\nonumber
\begin{split}
b_1-c +(n-3) \left( 1-\gamma \right) b_2 > 2(b_1-c)+(n-3) \left( 1-\gamma \right) b_2  + 2\gamma b_2  + (n-3) \gamma b_3
 \end{split}
\end{equation}
%\end{small}
which gives $c>b_1+ 2\gamma b_2 + (n-3) \gamma b_3$. But this is inconsistent with the condition in Inequality~(\ref{B1for2}). So in order to ensure that the newly entering node connects to the center and not to any of the leaf nodes, we use Inequality~(\ref{B3a}).

Now that a star of $n$ nodes is formed, we ensure its pairwise stability by deriving conditions for the same. 
Firstly, we ensure that the center does not delete any of its links. So we need $u_B(1) \geq u_B(5)$. Note that from the center's point of view, state $5$ is same as state $0$ and as we have seen earlier that $u_B(0) \leq u_B(1)$, the required condition $u_B(5) \leq u_B(1)$ is already ensured.\\
Next, no two leaf nodes should form a link between them. So we should ensure that, not creating a link between them is at least as good for them as creating, that is $u_j(1) \geq u_j(3)$ for any leaf node $j$. This condition is applicable for $n\geq 3$. 
\begin{equation}
\nonumber
b_1-c+(n-2) \left( 1-\gamma \right) b_2 \geq 2(b_1-c)+(n-3) \left( 1-\gamma \right) b_2
\end{equation}
\begin{equation}
\label{B4}
\iff c \geq b_1-b_2+\gamma b_2
\end{equation}
For a leaf node to not delete its link with the center, we need $u_j(1) \geq u_j(4)$ for any leaf node $j$. For $n \geq 2$, we have
\begin{equation}
\nonumber
b_1-c+(n-2)  \left( 1-\gamma \right) b_2  \geq 0
\end{equation}
\begin{equation}
\nonumber
\iff c \leq b_1+(n-2)  \left( 1-\gamma \right) b_2
\end{equation}
which is a weaker condition than Inequality~(\ref{B1for2}) for $n\geq 2$.

Note that Inequalities~(\ref{B1for2}) and (\ref{B3a}) put together are stronger than Inequalities~(\ref{B1}) and (\ref{B1forc0}) combined. 
We get the required result using Inequalities~(\ref{B1for2}), (\ref{B3a}) and (\ref{B4}).
%\qed 
\end{proof}

%\leavevmode

We provide the proofs of the remaining results of this section in Appendices~\ref{app:smallworld} through \ref{app:kstar}. %\cite{dhamal2012sufficient}.
%Appendices \ref{app:star} through \ref{app:2star}.

%\subsection{Sufficient Conditions for Formation of Other Topologies Uniquely}

\begin{proposition}
\label{thm:smallworld}
For a network, if $c<b_1-b_{d+1}$ 
%($d \geq 1$) 
and $c_0\leq(1-\gamma)b_2$, the resulting diameter is at most $d$.
\end{proposition}
%

%We provide a proof of Proposition~\ref{thm:smallworld} in Appendix~\ref{app:smallworld}.
%When $d=1$, we get the following corollary.
%The following corollary is immediate by substituting $d=1$.
% in the above proposition.
The following corollary results when $d=1$.

\begin{corollary}
\label{thm:complete}
For a network, if $c < b_1-b_2$ and $c_0 \leq \left( 1-\gamma \right) b_2$, the unique resulting topology is complete graph.
\end{corollary}
%
%We provide proofs of Propositions~\ref{thm:bipartite} and \ref{thm:2star} in 
%Appendices 
%\ref{app:bipartite} and \ref{app:2star}, respectively.

%\leavevmode

\begin{proposition}
\label{thm:bipartite}
For a network with $\gamma <   \frac{b_2 - b_3}{3b_2 - b_3} $, if $b_1-b_2+ \gamma \left( 3b_2 - b_3 \right) <  c < b_1 - b_3$ 
and $\left( 1-\gamma \right) \left( b_2-b_3 \right) < c_0 \leq \left( 1-\gamma \right) b_2$, the unique resulting topology is 
bipartite Tur\'an graph.
\end{proposition}

%\leavevmode

%
\begin{proposition}
\label{thm:2star}
Let $\sigma$ be the upper bound on the number of nodes that can enter the network and $\lambda = \lceil \frac{\sigma}{2} -1 \rceil \left( 2b_2-b_3 \right)$.
Then, if $\left( 1-\gamma \right) \left( b_2-b_3 \right) < c_0 < \left( 1-\gamma \right) \left( b_2-b_4 \right) $ and either \\
(i) $\gamma < \min \Big\{  \frac{b_2-b_3}{\lambda-b_3} ,  \frac{b_3}{b_2+b_3} \Big\}$ and $b_1-b_3+\gamma(b_2+b_3) \leq c < b_1$, or\\
(ii) $\frac{b_2-b_3}{\lambda-b_3} \leq \gamma < \min \Big\{ \frac{b_2}{\lambda+b_2} , \frac{b_3}{b_2+b_3}  \Big\}$ and $b_1-b_2+\gamma b_2 + \gamma \lambda \leq c < b_1$,
\\
the unique resulting topology is 
2-star.
\end{proposition}

%\leavevmode

%For arbitrarily large upper bound on the number of nodes that can enter the network, the following corollary is immediate from the above proposition.
%The following corollary transforms the above conditions to be independent of the upper bound on the number of nodes in the network.
%When $\gamma=0$, the conditions are independent of the upper bound on the number of nodes that can enter the network.
The following corollary transforms the above conditions in (i) to be independent of the upper bound on the number of nodes that can enter the network.

\begin{corollary}
\label{cor:2star}
For a network with $\gamma=0$, if $b_1-b_3 \leq  c < b_1$ and $b_2-b_3< c_0 < b_2-b_4$,
the unique resulting topology is 
2-star.
\end{corollary}

%\leavevmode

%\subsection{Sufficient Conditions for Formation of a General $k$-star Topology Uniquely}
%\label{sec:base}

We define {\em base graph} of a network formation process as the graph from which the process starts. The conditions derived for the formation of the above networks are obtained starting from the graph consisting of a single node (corresponding to the base case of formation of a network with $n=2$). 
Now for certain topologies to be well-defined, it is required that the network has a certain minimum number of nodes. For instance, for a network to have a well-defined $k$-star topology, it should consist of at least $2k$ nodes (complete network on $k$ centers with one leaf node connected to each center). So it is reasonable to consider this network of $2k$ nodes as a base graph for forming a $k$-star network.
Moreover, in case of some topologies (under a given utility model), the conditions required for its formation on discretely small number of nodes, may be inconsistent with that required on arbitrarily large number of nodes. We will now see that, under the proposed network formation and utility models, $k$-star ($k\geq 3$) is one such topology; and a way to circumvent this problem is to start the network formation process from the aforementioned base graph.

Note that in a real-world network, the upper bound on the number of nodes 
%that would be a part of the network, 
is unknown to the network owner. So it is essential that, irrespective of the number of nodes, the desired topology is formed and is stable. That is, the conditions on the network must be set such that the entire family of networks having that topology, is stable.
%The following lemma presents such conditions for the family of $k$-stars (given some $k\geq 3$). 
%The following lemma presents the necessary conditions for the entire family of $k$-star topology to be stable, for a given $k$. 
%We provide its proof in \cite{dhamal2012sufficient}.
%Appendix~\ref{app:kstar0}.

%\leavevmode

\begin{lemma}
\label{lem:kstar0}
Under the proposed utility model, for the entire family of $k$-star networks (given some $k\geq 3$) to be pairwise stable, it is necessary that 
$\gamma=0$ and $c=b_1-b_3$.
\end{lemma}

%\leavevmode

%RELOOK
%2 and k star defined for some lower bound on nodes\\

\begin{comment} %SeptEdit
For a given topology to be well-defined for a network, it may be necessary that the network has a certain minimum number of nodes. For example, a bipartite Tur\'an as well as a 2-star topology is well-defined for a network with a minimum of two nodes, while a $k$-star requires a minimum of $k$ nodes.
However, one can argue that a network with size less than this minimum number, should be said to have the topology if it is not inconsistent with the topology structure. 
So certain networks can be trivially classified into multiple topologies, for example, a network consisting of two interconnected nodes can be classified as bipartite Tur\'an, 2-star, as well as $k$-star. The following lemma is concerned with  $k$-star networks ($k\geq3$) which cannot be classified as 2-star.\\
\end{comment} %SeptEdit

%From Lemma~\ref{lem:kstar0}, 
It can be seen that 
the conditions necessary for the family of $k$-star networks to be pairwise stable (Lemma~\ref{lem:kstar0})
%are $\gamma=0$ and $c=b_1-b_3$, which 
are sufficient conditions for the formation of a 2-star network uniquely, 
when $b_2-b_3< c_0 < b_2-b_4$ (Corollary~\ref{cor:2star}).
%This results in the following lemma which is concerned with $k$-star networks ($k\geq3$) that cannot be classified as 2-star.
%
When $c_0 < b_2-b_3$, these conditions $\gamma=0$ and $c=b_1-b_3$, are sufficient for the formation of a star topology uniquely (Proposition~\ref{thm:star}).
When $b_2-b_4 < c_0 < b_2$, these necessary conditions form a cycle among the initially entered nodes, but fails to form a clique among $k$ nodes even as more nodes enter the network, thus making it inconsistent with the $k$-star topology.
%$c_0>b_2$ does not allow the network to grow beyond two nodes. 
It can be similarly seen that for other values of $c_0$ including the boundary cases $c_0 = b_2-b_3$ and $c_0 = b_2-b_4$, the network so formed is not consistent with $k$-star topology for any $k\geq 3$. 
So we have that,
%
%\begin{lemma}
%\label{lem:kstar}
under the proposed network formation and utility models, with the requirement that the entire family be pairwise stable, no $k$-star network (given some $k\geq3$) can be formed starting with a network consisting of a single node.
%\footnote{By no $k$-star network, we mean no $k$-star network which cannot be classified as a 2-star network, for example, networks consisting of two interconnected nodes or line networks consisting of four nodes can be trivially classified as both 2-star as well as $k$-star networks.}
%\end{lemma}
%
%\begin{IEEEproof}
%From Lemma~\ref{lem:kstar0}, the conditions necessary for the family of $k$-star networks to be pairwise stable are $\gamma=0$ and $c=b_1-b_3$, which are sufficient conditions for the formation of a 2-star network uniquely, according to Corollary~\ref{cor:2star}.
%%\qed 
%\end{IEEEproof}

%\leavevmode

A reasonable solution to overcome this problem is to start the network formation process from some other base graph. Such a graph can be obtained by external methods such as providing additional incentives to its nodes.
For 
%instance, for 
initializing the formation of $k$-star, 
as mentioned earlier,
the base graph can be taken to be the complete network on the $k$ centers, with the centers connected to one leaf node each. As the base graph consists of $2k$ nodes, the induction starts with the base case for formation of $k$-star network with $n=2k+1$.
%Proposition~\ref{thm:kstar} gives the sufficient conditions for the formation of $k$-star network uniquely, provided the network starts building itself from this base graph. 
%We provide its proof in
%\cite{dhamal2012sufficient}.
%Appendix~\ref{app:kstar}.
%The following proposition results.

%\leavevmode

\begin{proposition}
\label{thm:kstar}
For a network starting with the base graph for $k$-star (given some $k \geq 3$), and $\gamma =0 $, if $c =b_1-b_3 $ 
and $ b_2-b_3  < c_0 < b_2-b_4$, the unique resulting topology is 
$k$-star.
\end{proposition}

\begin{comment} %SeptEdit
\subsection{A Note on the Derived Sufficient Conditions}
\label{sec:anote}
\end{comment} %SeptEdit

\subsection{Intuition Behind the Sufficient Conditions}
\label{sec:explain}

The network entry fee has an impact on the resulting topology as seen from the above propositions. For instance, in Propositions~\ref{thm:star} and \ref{thm:bipartite}, the intervals spanned by the values of $c$ and $\gamma$ may intersect, but the values of network entry factor $c_0$ span mutually exclusive intervals separated at $(1-\gamma)(b_2-b_3)$. In case of star, $c_0$ is low and so a newly entering node can afford to connect to the center, which in general, has very high degree. In case of bipartite Tur\'an graph, it is important to ensure that the sizes of the two partitions are as equal as possible. As $c_0$ is high, a newly entering node connects to a node with a lower degree (whenever applicable), that is, to a node that belongs to the partition with more number of nodes. Hence the newly entering node potentially becomes a part of the partition with fewer number of nodes, thus maintaining a balance between the sizes of the two partitions. 
In case of $k$-star, as the objective is to ensure that a newly entering node connects to a node with moderate degree, the network entry factor is not so high that a newly entering node prefers connecting to a leaf node and not so low that it prefers connecting to a center with the highest degree. This intuition is clearly reflected in Propositions~\ref{thm:2star} and \ref{thm:kstar} where $c_0$ takes intermediate values.
In general, {\em network entry factor} $c_0$ plays an important role in dictating the degree distribution of the resulting network; 
a higher value of $c_0$ lays the foundation for formation of a more regular graph.
%and plays an important role in dictating the degree distribution of the resulting network. 

\begin{comment} %SeptEdit
Also as discussed earlier, a high value of {\em network entry factor} $c_0$ lays the foundation for formation of a regular graph. In general, it plays an important role in dictating the degree distribution of the resulting network.
\end{comment} %SeptEdit

As $c$ increases, the desirability of a node to form links decreases.
This is clear from Proposition~\ref{thm:smallworld} which says that, as $c$ decreases, nodes would create more links, hence effectively reducing the network diameter.
In particular, a complete network is formed when the costs of maintaining links is extremely low, as reflected in Corollary~\ref{thm:complete}. The remaining topologies are formed in the intermediate ranges of $c$. 

From Propositions~\ref{thm:bipartite}, \ref{thm:2star} and \ref{thm:kstar}, it can be seen that the feasibility of a network being formed depends on the values of $\gamma$ as well, which arises owing to contrasting densities of connections in a network. 
For instance, in a bipartite Tur\'an network, nodes belonging to different partitions are densely connected with each other, while that within the same partition are not connected at all. Similarly, in a $k$-star network, there is an extreme contrast in the densities of connections (dense amongst centers and sparse for leaf nodes).

\subsection{Connection to Efficiency}
\label{sec:efficiency}

We now analyze efficiency of the considered networks. 
As the derived conditions are sufficient, there may exist other sets of conditions that uniquely result in a given topology. We analyze the efficiency 
%based on the assumption 
assuming that the networks are formed using the derived conditions.
%
%We provide the proofs of the results of this section in %\cite{dhamal2012sufficient}.
%Appendix \ref{app:eff_bipartite_kstar}.
%Section 3 of the supplementary material.

From Equation~(\ref{eqn:utility}), the intermediation rents are transferable among the nodes, and so do not affect the efficiency of a network. Furthermore, the network entry fee is paid by any node at most once, and so does not account for efficiency in the long run. So the expression for efficiency of a network is
\begin{equation}
\nonumber
\sum_{j\in N} \Bigg( d_j(b_1-c) + \sum_{\substack{w \in N \\l(j,w)>1}}{b_{l(j,w)}} \Bigg)
 \end{equation}

The following result follows from the analysis by 
Narayanam and Narahari~\cite{ramasuri1}.
\begin{lemma}
\label{lem:efficient}
Let $\mu$ be the number of nodes in  network. \\
(a) If $c < b_1-b_2$, complete graph is uniquely efficient.\\
(b) If $b_1-b_2< c \leq b_1 + \left( \frac{\mu-2}{2} \right) b_2$, star is the unique efficient topology.\\
(c) If $c >b_1+ \left( \frac{\mu-2}{2}\right)b_2 $, null graph is uniquely efficient.
\end{lemma} 

The null network in the proposed model of recursive network formation corresponds to a single node to which no other node prefers to connect, and so the network does not grow. 

%\leavevmode

\begin{proposition}
\label{thm:eff_star}
Based on the derived sufficient conditions, null network, star network, and complete network are efficient.
\end{proposition}
%\begin{comment} %SeptEdit
\begin{proof}
It is easy to see that irrespective of the value of $c_0$, if $c>b_1$, no node, external to the network, connects to the only node in the network and hence, does not enter the network. Such a network is trivially efficient as in the range $c>b_1$, it is a star of one node and also a null network. It is also clear that the star network and the complete network are efficient as the conditions on $c$ from Proposition~\ref{thm:star} and Corollary~\ref{thm:complete}, respectively, form a subset of the range of $c$ in which these topologies are respectively efficient.
%\qed 
\end{proof}
%\end{comment} %SeptEdit

%\leavevmode

%Propositions~\ref{thm:eff_bipartite} and \ref{thm:eff_kstar} give bounds on the efficiency of bipartite Tur\'an and $k$-star networks, respectively.
It can be seen that when the number of nodes in the network is small, the absolute difference between the efficiency of the resulting network and that of the efficient network is also small, and hence the network owner will not be too concerned about the efficiency of the network. 
So for the following propositions, we make a reasonable assumption that the number of nodes in the network is sufficiently large. 
%We provide proofs of Propositions~\ref{thm:eff_bipartite} and \ref{thm:eff_kstar} in Appendix~\ref{app:eff_bipartite_kstar}.

%\leavevmode

\begin{proposition}
\label{thm:eff_bipartite}
Based on the derived sufficient conditions, for sufficiently large number of nodes, the efficiency of a bipartite Tur\'an network is half of that of the efficient network in the worst case and the network is close to being efficient in the best case.
%\\
\end{proposition}
%
%\begin{comment}
\begin{proof} 
As $\mu$ is large, $\mu$ can be assumed to be even without loss of accuracy. The sum of utilities of nodes in a bipartite Tur\'an network with even number of nodes is approximately
%\begin{small}
\begin{equation}
\nonumber
\left( \frac{\mu}{2} \right)^2 2(b_1-c)+2 \dbinom{\frac{\mu}{2}}{2}2b_2
\end{equation}
%\end{small}
From Lemma~\ref{lem:efficient}, star network is efficient in the range of $c$ derived in Proposition~\ref{thm:bipartite}. So, to get the efficiency of the bipartite Tur\'an network relative to the star network, we divide the above expression by the sum of utilities of nodes in a star network, which is
%\begin{small}
\begin{equation}
\label{eqn:star_eff}
2(\mu-1)(b_1-c)+\dbinom{\mu -1}{2}2b_2
\end{equation}
%\end{small}
Using the assumption that $\mu$ is large and the fact from the derived sufficient conditions that $b_2$ is comparable to $b_1-c$, it can be shown that the efficiency relative to the star network, approximately is
$\frac{1}{2}+\frac{b_1-c}{2b_2}$.
%\begin{equation}
%\nonumber
%\frac{1}{2}+\frac{b_1-c}{2b_2}
%\end{equation}
As the range of $c$ in Proposition~\ref{thm:bipartite} depends on the value of $\gamma$, the values of $c$ are bounded by $b_1-b_2$ and $b_1-b_3$. So the efficiency is bounded by 1 on the upper side and $\left( \frac{1}{2}+\frac{b_3}{2b_2} \right)$ on the lower side, of that of the star network; $\left( \frac{1}{2}+\frac{b_3}{2b_2} \right)$ can take a minimum value of $\frac{1}{2}$ when $b_3<<b_2$. 
%\qed 
\end{proof}
%\end{comment}

%\leavevmode

\begin{proposition}
\label{thm:eff_kstar}
Based on the derived sufficient conditions, for sufficiently large number of nodes, the efficiency of a $k$-star network is $\frac{1}{k}$ of that of the efficient network in the worst case and the network is close to being efficient in the best case.
\end{proposition}
%
%\begin{comment}
\begin{proof}
As $\mu$ is large, in particular, $\mu >> k$ (not necessarily $>>k^2$), $\mu$ can be assumed to be divisible by $k$ without loss of accuracy. The sum of utilities of nodes in such a $k$-star network is approximately
%\begin{small}
%\begin{equation}
%\nonumber
%\begin{split}
%2\dbinom{k}{2}(b_1-c)+2(\mu-k)(b_1-c)+ 2k \left( 1-\frac{1}{k} \right) (\mu -k) b_2\\+2k \dbinom{\frac{\mu}{k} -1}{2}b_2 + k \left( \frac{\mu}{k}-1 \right) \left( 1- \frac{1}{k} \right) (n-k) b_3
%\end{split}
%\end{equation}
%\end{small}
%\begin{small}
\begin{equation}
\nonumber
\begin{split}
\left\{\dbinom{k}{2}+(\mu-k) \right\}2(b_1-c)+ \left\{k (k-1)\left( \frac{\mu-k}{k} \right)+k \dbinom{\frac{\mu-k}{k} }{2} \right\} 2b_2  + \dbinom{k}{2} \left( \frac{\mu-k}{k} \right) ^2 2b_3
\end{split}
\end{equation}
%\end{small}
From Lemma~\ref{lem:efficient}, star network is efficient in the range of $c$ derived in Propositions~\ref{thm:2star} and \ref{thm:kstar}. So, to get the efficiency of the $k$-star network relative to the star network, we divide the above expression by Expression~(\ref{eqn:star_eff}).
Using the assumption that $\mu$ is large and the fact from the derived sufficient conditions that $b_2$ and $b_3$ are comparable to $b_1-c$, it can be shown that the efficiency relative to the star network, approximately is
$\frac{1}{k}+ \left( 1- \frac{1}{k} \right) \frac{b_3}{b_2}$.
%\begin{equation}
%\nonumber
%\frac{1}{k}+ \left( 1- \frac{1}{k} \right) \frac{b_3}{b_2}
%\end{equation}
As $b_3$ is bounded by $0$ and $b_2$, the efficiency of $k$-star is bounded by $\frac{1}{k}$ and 1 of that of the star network.
%\qed 
\end{proof} 
%\end{comment}

%\section{Effects of Deviation from the Derived Sufficient Conditions: A Simulation Study}
\section{Deviation from the Derived Sufficient Conditions: A Simulation Study}
\label{sec:deviation}
We have derived sufficient conditions under which various network topologies uniquely emerge. In this section, we investigate the robustness of the derived sufficient conditions by studying the deviation in network topology when there is a slight deviation in these sufficient conditions. This problem is of practical interest since it may be difficult to maintain the conditions on a network throughout its formation process.

\begin{comment} %SeptEdit
We use {\em graph edit distance} (GED) to measure the dissimilarity between two graphs.
A graph can be transformed to another one by a sequence of graph edit operations which are defined in different manners in the literature. The notion of graph edit distance is defined by the least-cost edit operation sequence~\cite{gao2010survey}. In this paper, we use the following definition of graph edit distance.
\end{comment} %SeptEdit
We use the notion of {\em graph edit distance} (GED)~\cite{gao2010survey} to measure 
%the dissimilarity between two graphs. 
the deviation in network topology.

\begin{definition}
Given two graphs $g$ and $h$ having same number of nodes, the {\em graph edit distance} 
between them is the minimum number of link additions and deletions required to transform $g$ into a graph that is isomorphic to $h$.
\end{definition}

\subsection{Computation of Graph Edit Distance}
\label{sec:ged}

The problem of computing GED between two graphs is NP-hard, in general~\cite{zeng2009comparing}.
However, we can exploit structural properties of certain graphs to compute GED between them and other graphs, in polynomial time;
we state three such results. 
%Theorems~\ref{thm:gedstar} and \ref{thm:gedcomplete} are easy to prove. 
%We provide a proof of Theorem~\ref{thm:gedkstar} in 
%Appendix \ref{app:gedkstar}.
%Section 2 of the supplementary material.
%the Appendix.
%\ref{app:gedkstar}.
%\\

\begin{theorem}
\label{thm:gedstar}
%The graph edit distance 
The graph edit distance between a graph $g$ and a star graph with same number of nodes as $g$, is $\mu+\xi-2\Delta-1$, where $\mu$ and $\xi$ are the number of nodes and edges in $g$, respectively, and $\Delta$ is the 
 highest degree 
%degree of a node having the highest degree 
in $g$.%\\
\end{theorem}
%\begin{comment}
\begin{proof}
While transforming $g$ into a corresponding star graph, we need to map one node of $g$ to the center while the others to the leaf nodes. Let $d$ be the degree of the node which is mapped to the center. In order to transform $g$ into a star graph, the node mapped to the center must be connected to $\mu-1$ nodes. So the number of edges to be added is $(\mu-1)-d$. Also all edges connecting any two nodes, that are mapped to the leaf nodes, must be deleted, that is, all edges except the ones incident to the node mapped to the center, must be removed. These account for $\xi - d$ edges. Thus, total number of edges to be added and deleted is $\mu+\xi-2d-1$. This is minimized when $d=\Delta$.
%\qed
\end{proof}
%\end{comment}

\begin{theorem}
\label{thm:gedcomplete}
%The graph edit distance 
The graph edit distance between a graph $g$ and a complete graph with same number of nodes as $g$, is $\frac{\mu(\mu-1)}{2} - \xi$, where $\mu$ and $\xi$ are the number of nodes and edges in $g$, respectively.%\\
\end{theorem}
%\begin{comment}
\begin{proof}
Graph $g$ can be transformed into the corresponding complete graph in minimum number of steps by adding the edges which are absent.
%\qed
\end{proof}
%\end{comment}

%

\begin{theorem}
\label{thm:gedkstar}
There exists an $O(\mu^{k+2})$ polynomial time algorithm to compute the graph edit distance %the graph edit distance
 between a graph $g$ and a $k$-star graph with same number of nodes as $g$, where $\mu$ is the number of nodes in $g$.
\end{theorem}

We provide the proof of Theorem~\ref{thm:gedkstar} in Appendix~\ref{app:gedkstar}.

\subsection{Simulation Setup}
\label{sec:simsetup}

In order to study the robustness of the derived sufficient conditions, we observed the effects of deviation from these conditions, on the resulting networks, using GED as the measure of topology deviation. We first observed the effect when the conditions were made to deviate throughout the network formation process. The results were, however, uninteresting since the deviation from the sufficient conditions for the formation of one topology, lead to the formation of a completely different topology. 
A primary reason for such observations is that, under the deviated conditions, some other networks are pairwise stable and these networks have a very different topology than the desired one.
In some cases, these deviated conditions were sufficient conditions for other topologies, which were, however, not the desired ones.

In fact, it is unreasonable to assume that the conditions remain deviated throughout the entire network formation process.
%, since the network owner will take necessary actions in such cases. 
It is possible that the conditions deviate at some point of time, 
%owing to some internal or external factors. The 
but the network owner will observe the resulting network under such deviations and take necessary actions to rectify this problem.
% We use the concept of dynamic conditions here.
%We refer to this as 
This lets us introduce the concept of {\em dynamic conditions} on the network.

In simulations, we assume that the conditions deviate during the entry of a new node and remain deviated throughout the evolution of the network until it reaches pairwise stability. Once stability is reached, the network owner observes the deviation of the network from the desired one, and takes actions to restore the original conditions. As it is undesirable for the network to remain stagnant, any node which wants to enter the deviated network next, is allowed to do so immediately, and the original conditions take effect during the entry of such a node and evolution thereafter.
% of the network.

We observe how the topology deviates when the conditions deviate, and if, how, and when the topology is restored, once the sufficient conditions are restored. We also observe the values within the sufficient conditions which are more robust than others, that is, when the conditions are restored to these values, the topology is restored at the earliest.
%and is restored within the fewest following node entries.

For simulations, we set the benefit parameters as per the symmetric connections model~\cite{jackson1996strategic}, that is, we set $b_i=\delta^i$, where $\delta \in (0,1)$; we set $\delta=0.8$ in our simulations.
We consider three types of values within the sufficient conditions, namely, \{low($L$), moderate($M$), high($H$)\}
 for each of the parameters $c$, $c_0$ and $\gamma$ (whenever applicable) and observe the combination of their values which are the most robust to deviations.
 %\footnote{
 In our simulation study, low values correspond to value around the lower 10\% of the range in sufficient conditions, moderate to around 50\% mark, and high to around higher 10\%.
 %}
Also, for each combination, we run the network formation process several times in order to account for the effects of randomization in the order in which nodes take decisions.

Owing to sequential entry of nodes, there is an inherent ordering on nodes and they can be numbered from 1 to the current number of nodes in the network, in the order in which they enter. We call the node number at which the sufficient conditions deviate, as the {\em deviation node}. The sufficient conditions are restored during the entry of the node immediately following the deviation node. 
We say that the deviation from sufficient conditions on a parameter is {\em negative} if the deviated value of the parameter is less than its lower bound in the sufficient conditions, and {\em positive} if its deviated value is greater than its upper bound.
%\footnote{
In our simulation study, the amount of deviation for each parameter was 2\% of the length of its range in sufficient conditions. The results observed for 5\% and 10\% deviations were almost same. For parameters whose range in sufficient conditions is a singleton, the results were studied for an absolute deviation of 0.01 on the scale where $b_i = 0.8^i$.
%} % on the scale where $b_i = 0.8^i$.}

\subsection{Simulation Results}
\label{sec:simresults}

We observe the effects of deviation from the derived sufficient conditions for $c$ and $c_0$ on the resulting network. The observations can be primarily classified into the following four cases, in the decreasing order of desirability to network owner:
\begin{enumerate}
\item[(A)] The network does not deviate during the entry and also during the evolution after the entry of deviation node.
\item[(B)] The network deviates after the entry of the deviation node, and perhaps remains deviated during the entry and evolution for the entry of nodes following the deviation node, but after a certain number of such node entries, the network regains its original topology.
\item[(C)] The network deviates after the entry of the deviation node and remains deviated during the entry and evolution for the entry of nodes following the deviation node; the network does not regain its original topology, but the deviation is constant and so a near-desired topology is obtained.
\item[(D)] The network deviates after the entry of the deviation node and the deviation increases monotonically during the entry and evolution for the entry of nodes following the deviation node.
\end{enumerate}

\begin{figure} [t!]
\begin{tabular}{c}
\hspace{-.7cm}
\begin{minipage}{.5\textwidth}
\centering
\iftoggle{clr}{
\includegraphics[scale=0.62]{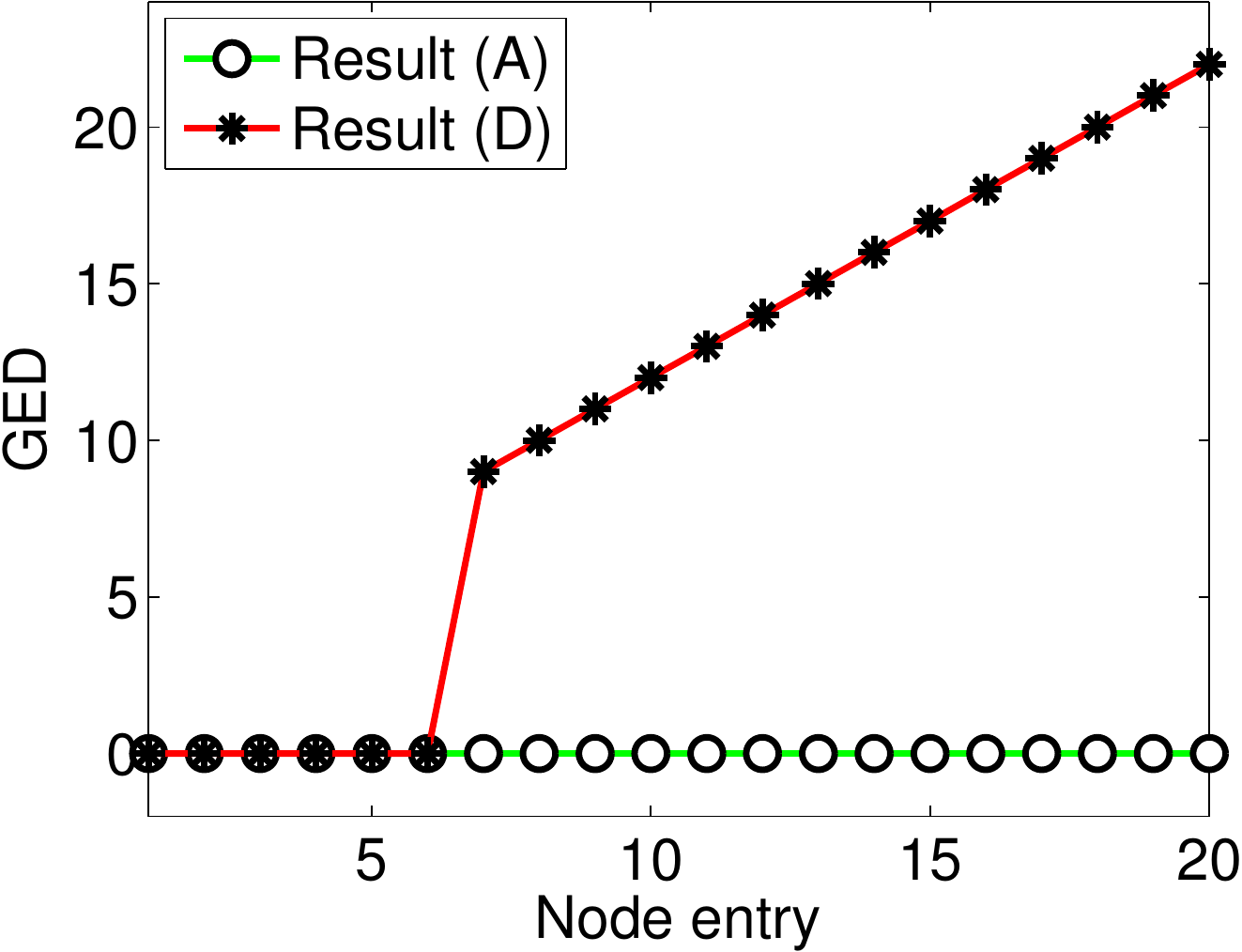}
}{
\includegraphics[scale=0.62]{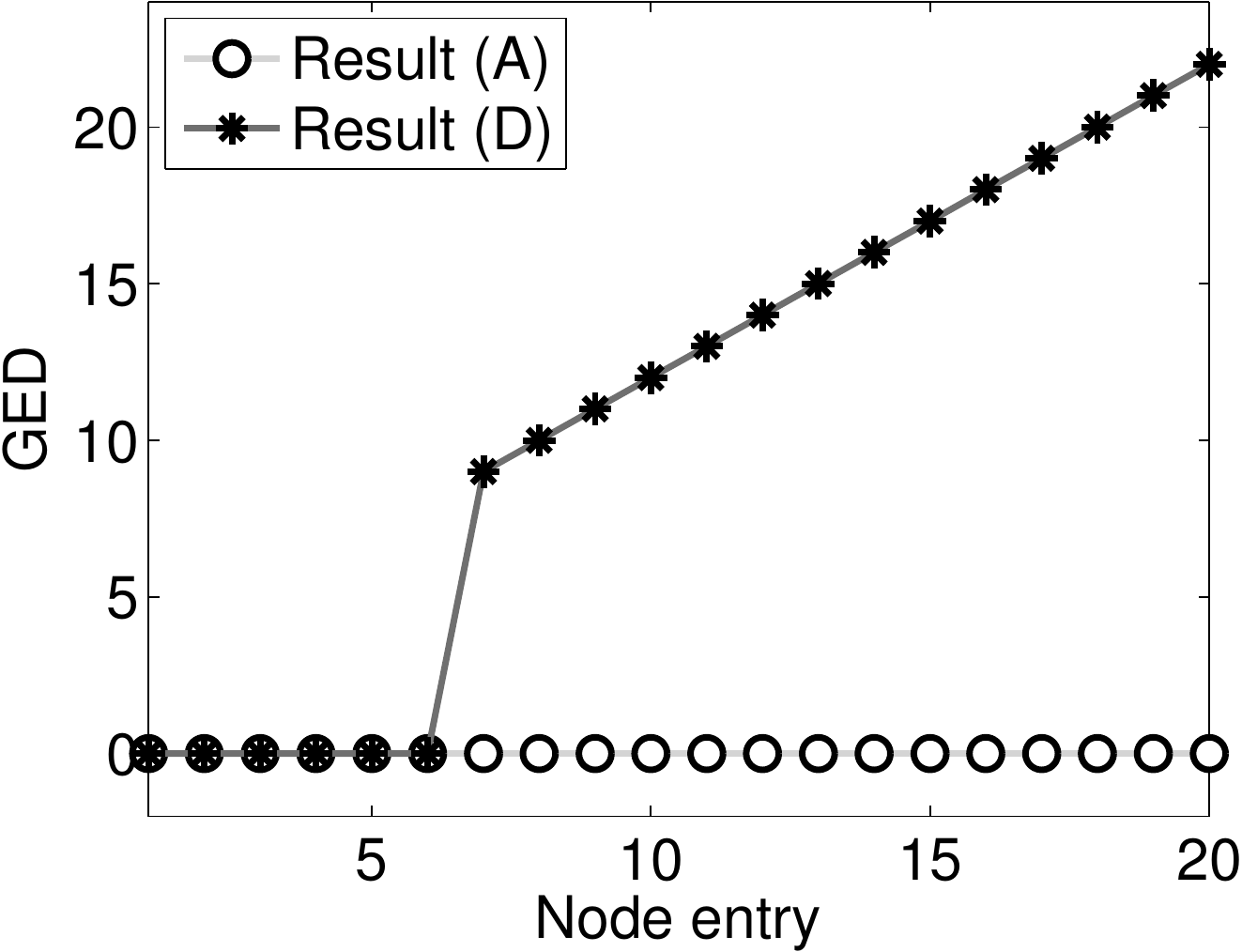}
}
\end{minipage}
\begin{minipage}{.5\textwidth}
\centering
\iftoggle{clr}{
\includegraphics[scale=0.62]{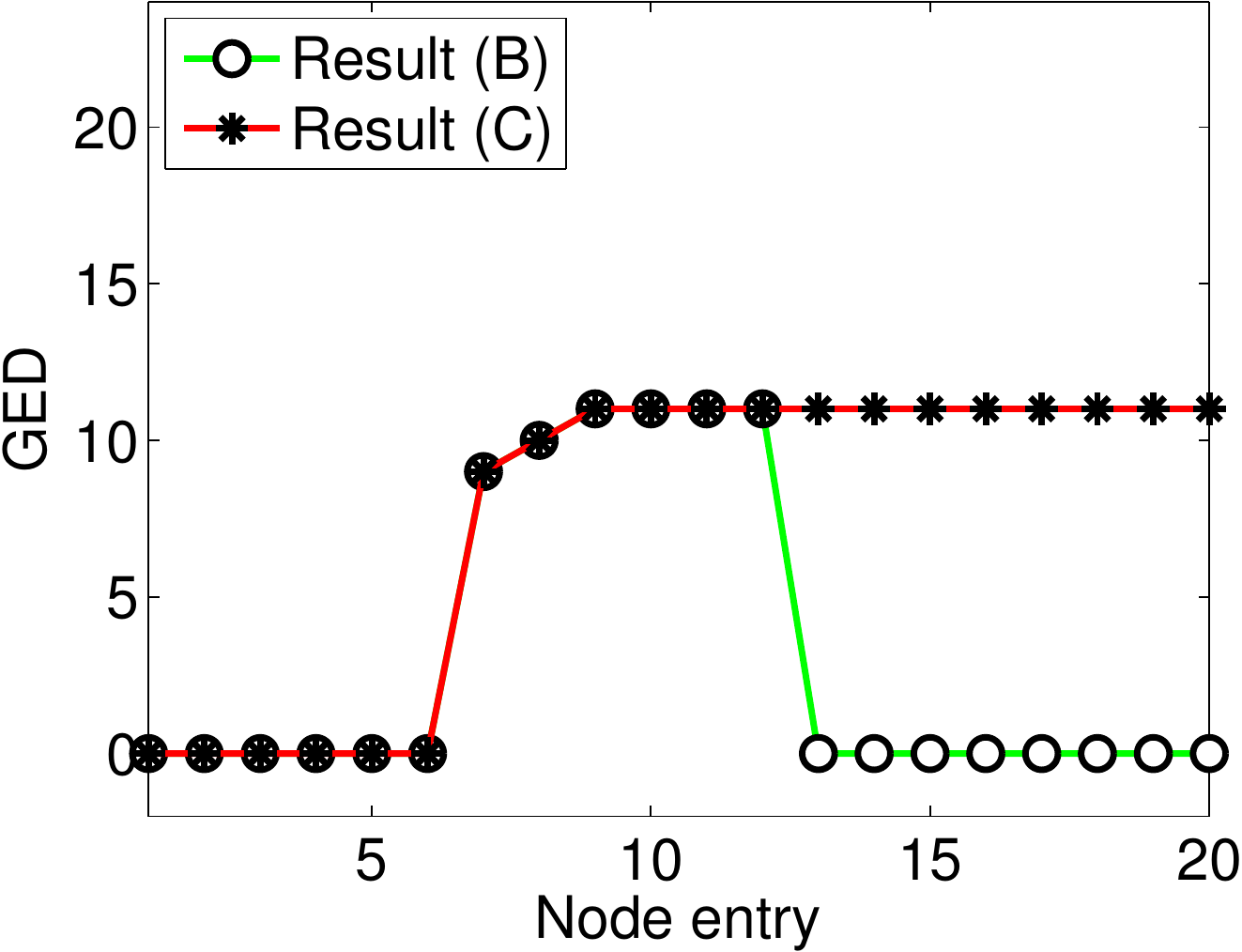}
}{
\includegraphics[scale=0.62]{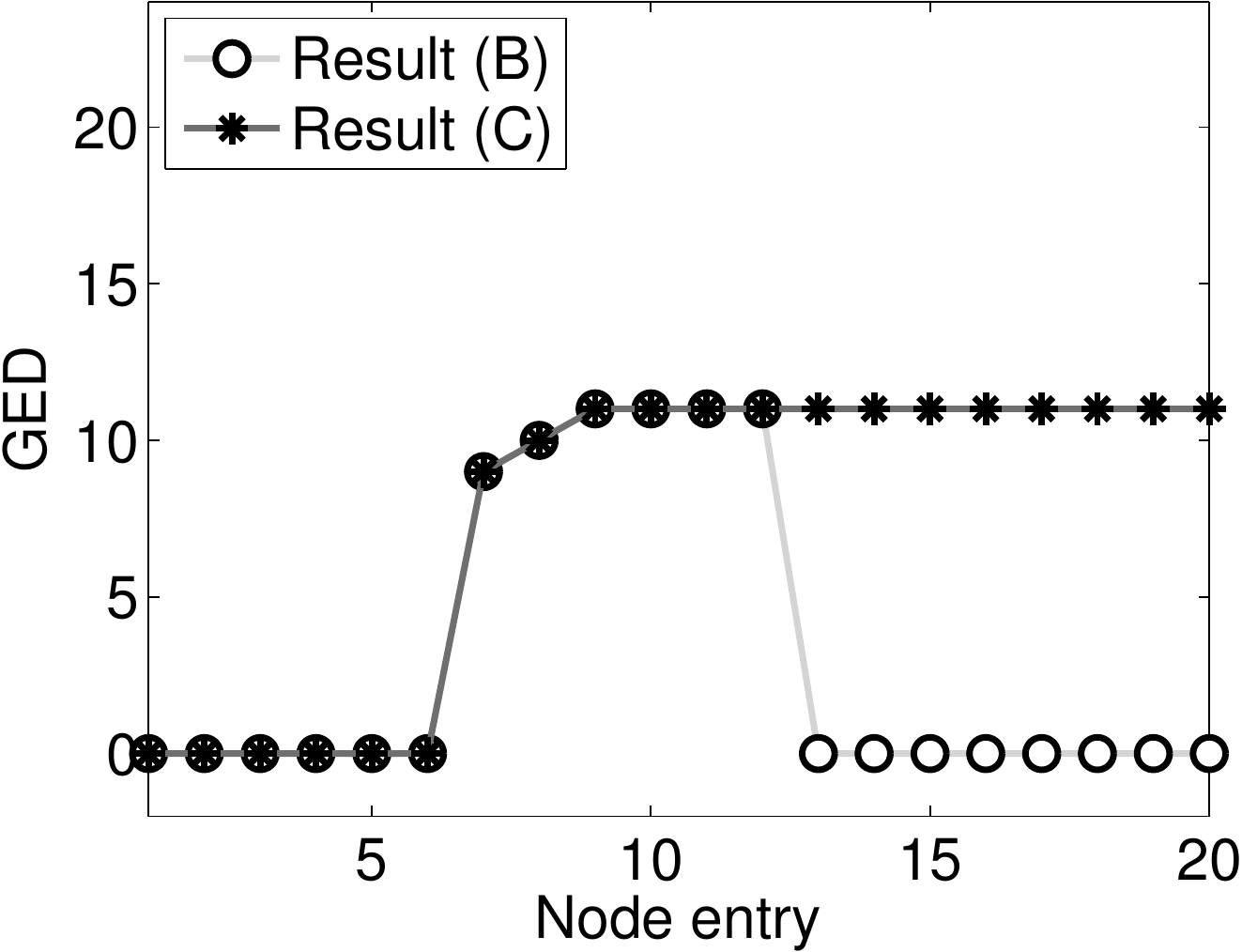}
}
\end{minipage}
\\
%\hspace{-.5cm}
\begin{minipage}{.5\textwidth}
\centering
\small{(a)}
\end{minipage}
\begin{minipage}{.5\textwidth}
\centering
\small{(b)}
\end{minipage}
\end{tabular}
\caption{(a-b) Typical results of deviation from the derived sufficient conditions 
%for the formation of a star network uniquely, 
for deviation node 7
(Y-axis gives the deviation when the network consists of number of nodes given on X-axis)
}
\label{fig:resulttypes}
%\vspace{-.25cm}
\end{figure}

Figures~\ref{fig:resulttypes}(a-b) give typical plots of the above four cases.
% owing to deviation from the derived sufficient conditions.
% for the formation of a star network.
The plots are split into two parts for clarity.
Result (A) is the most desirable but can be obtained only for some particular deviation nodes depending on the topology for which the sufficient conditions are derived. Result (B) is very common and
this is the result the network owner should be looking at. Result (C) is good from a practical viewpoint as the resulting network need not be exactly the desired one, but it may still serve the purpose almost entirely. Result (D) is the one that any network owner should avoid.

Recall that $c$ is the cost incurred by a node in order to maintain a link with each of its immediate neighbors. So as $c$ increases, the desirability of a node to form links decreases.
Also as discussed earlier, a higher value of {\em network entry factor} $c_0$ lays the foundation for formation of a more regular graph. In general, it plays an important role in dictating the degree distribution of the resulting network.
In what follows, we study the effects of all valid deviations from sufficient conditions on cost parameters $c$ and $c_0$, on the resulting network.
In the tables that follow, if there were very few instances in which the network did not deviate, we ignore them since such cases are remote when nodes take decisions in some particular order.
For observing deviations from $k$-star topology ($k \geq 3$), the network is assumed to start with the corresponding base graph consisting of $2k$ nodes as discussed earlier.

Enlisted are the major findings of the simulations:
\begin{itemize} 
\item Certain values of parameters within the derived sufficient conditions may be more robust than others, that is, the value to which the conditions are restored during the entry of the node immediately following the entry of the deviation node, may directly affect the restoration of the topology.
\item Network with certain number of nodes may be bottleneck for the range of sufficient conditions (can be seen from the derivations of these conditions). In such cases, the topology deviates only for discretely few deviation nodes, while it does not for others. So the network owner may relax the conditions for most of the network formation process.
\item The sufficient conditions on $c$ are more sensitive than those on $c_0$, that is, the network deviates more from the desired topology when the value of $c$ deviates than when the value of $c_0$ deviates by similar margins.
\item Results obtained owing to deviation from sufficient conditions during the entry of a deviation node may be very different from that obtained owing to deviation during the entry of some other deviation node.
\item It may be possible to uniquely form some interesting topologies which may not be feasible using any static sufficient conditions.
\item In most scenarios, the order in which nodes take decisions plays an important role in deciding the resulting topology. Deviations from sufficient conditions may cause large deviations from the desired topology due to some ordering, while no deviation at all due to some other.
\end{itemize}

%\begin{figure}[t!]
%\begin{tabular}{cc}
%\begin{minipage}{.5\textwidth}
%\centering
%%\hspace{-9mm}
%\iftoggle{clr}{
%    \includegraphics[scale=0.85]{devnearstar.pdf}
%    }{
%    \includegraphics[scale=0.85]{devnearstar_bw.pdf}
%    }
%\end{minipage}
%\begin{minipage}{.5\textwidth}
%\centering
%%\hspace{-9mm}
%\iftoggle{clr}{
%    \includegraphics[scale=0.85]{devdoublestar.pdf}
%}{
%    \includegraphics[scale=0.85]{devdoublestar_bw.pdf}
%    }
%\end{minipage}
%\\
%\begin{minipage}{.5\textwidth}
%\centering
%    (a)
%\end{minipage}
%\begin{minipage}{.5\textwidth}
%\centering
%    (b)
%\end{minipage}
%\end{tabular}
%  \caption{(a) A near-star network and (b) A $(2,8)$-complete bipartite network
%%  Result of restoring the condition on $c$ to be a low value 
%%after positive deviation of $c_0$ from the sufficient conditions for star network
%}
%  \label{fig:bothdevstars}
%%  \vspace{-.25cm}
%\end{figure}

\begin{figure} [t!]
\begin{tabular}{c}
\hspace{-.7cm}
\begin{minipage}{.5\textwidth}
\centering
\iftoggle{clr}{
\includegraphics[scale=0.62]{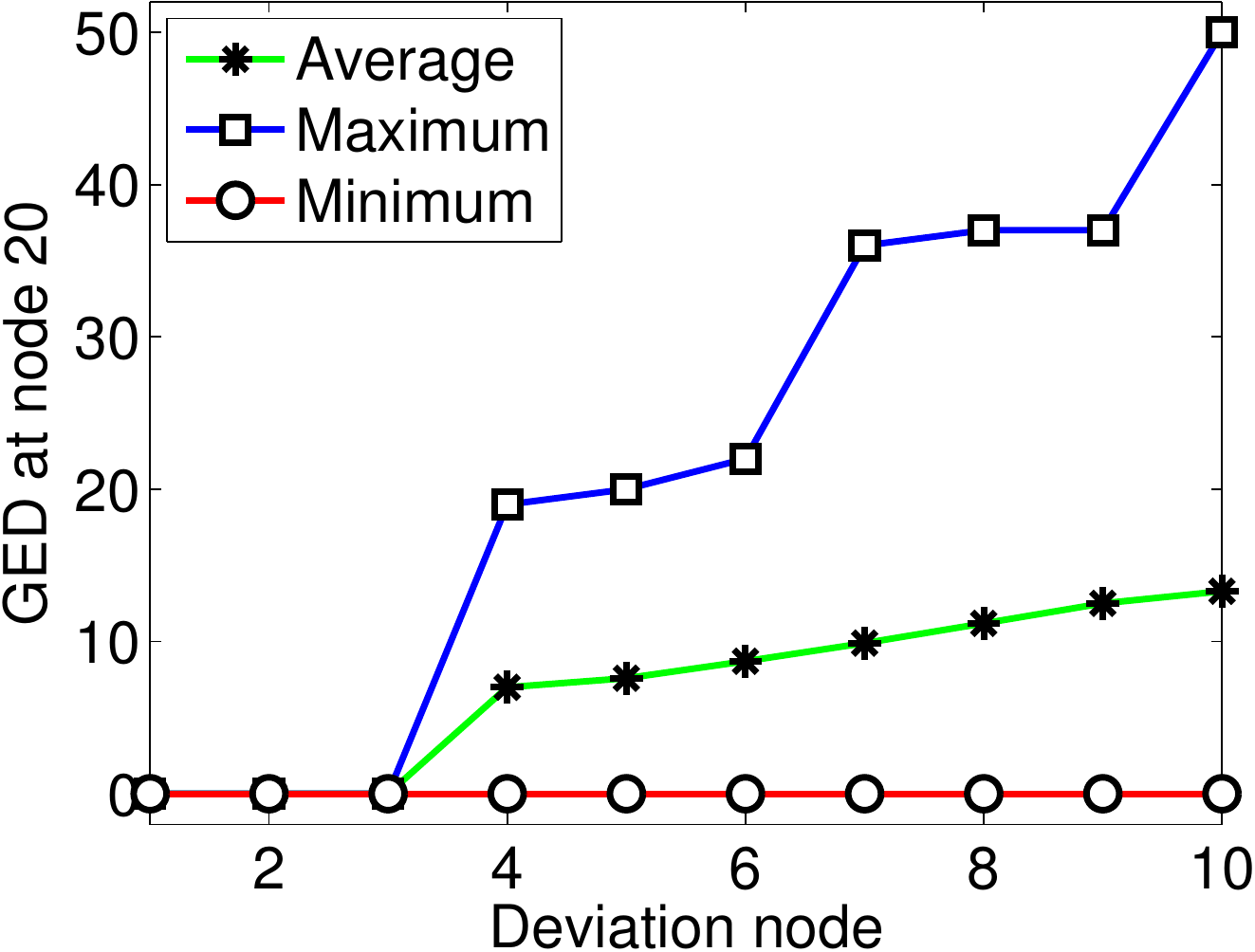}
}{
\includegraphics[scale=0.62]{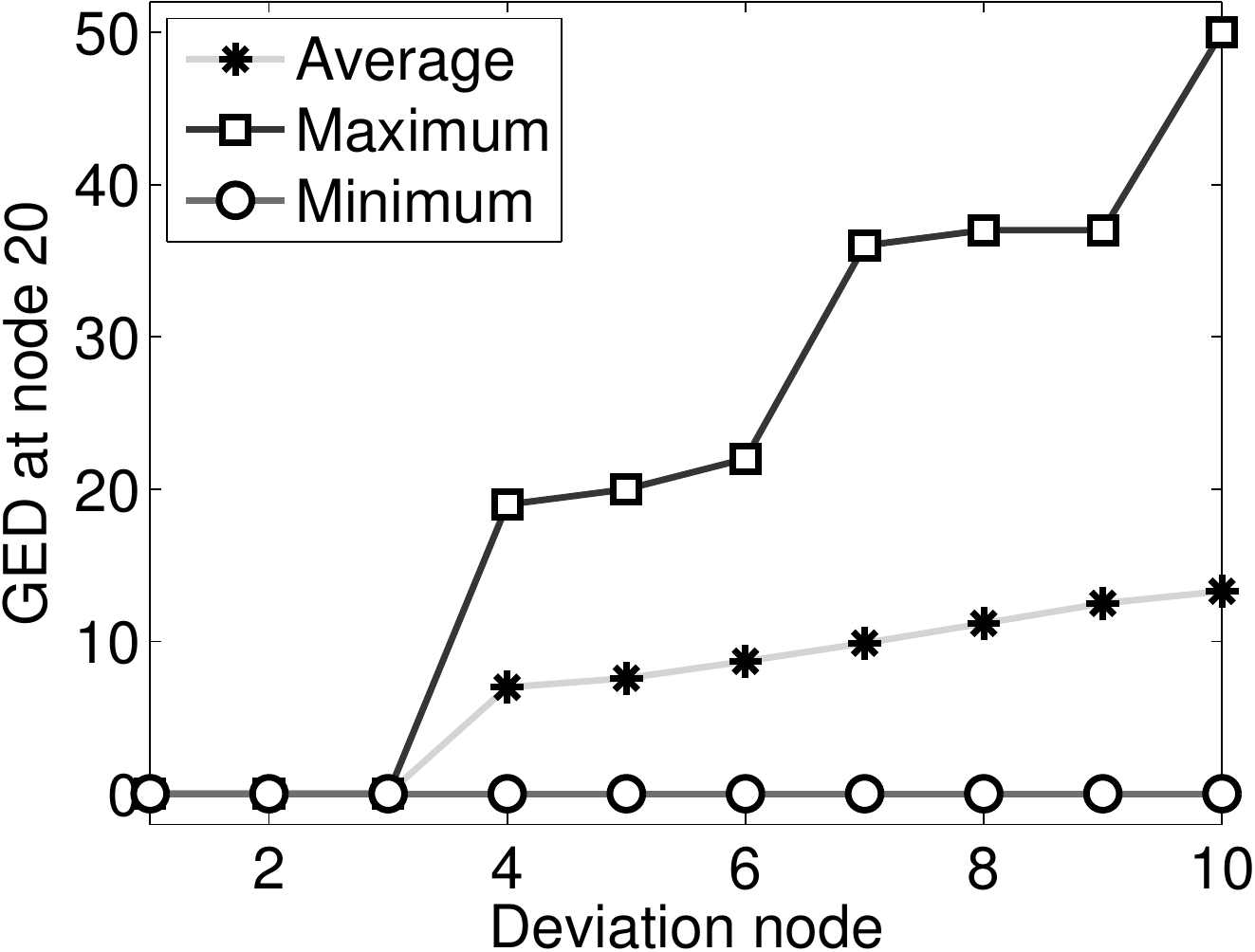}
}
\end{minipage}
\begin{minipage}{.5\textwidth}
\centering
\iftoggle{clr}{
    \includegraphics[scale=0.85]{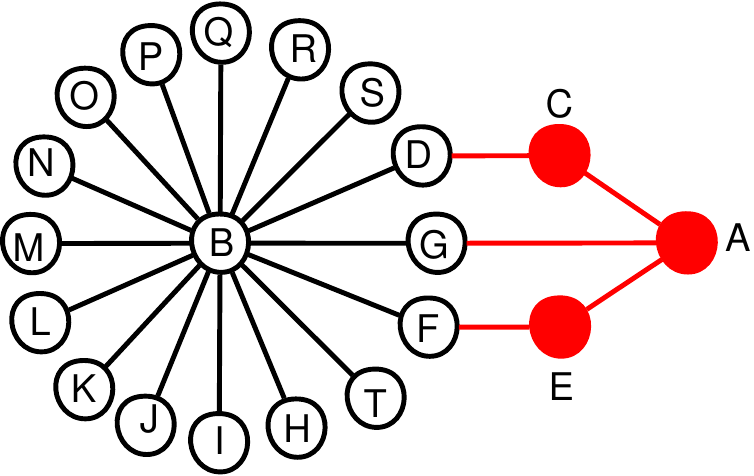}
}{
    \includegraphics[scale=0.85]{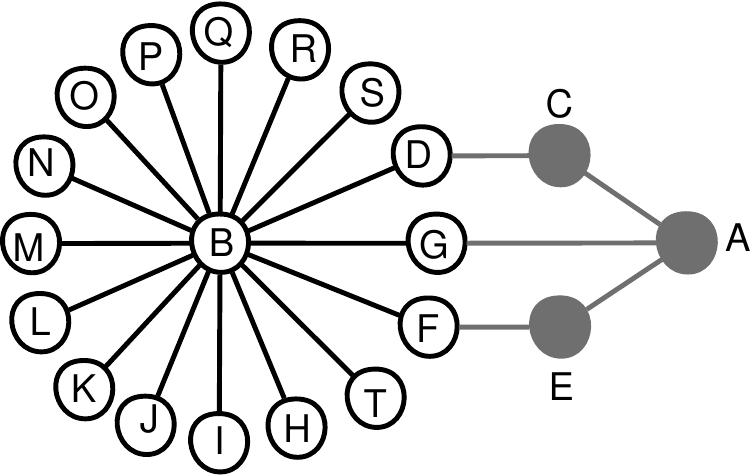}
}
\end{minipage}
\\
%\hspace{-.5cm}
\begin{minipage}{.5\textwidth}
\centering
\small{(a)}
\end{minipage}
\begin{minipage}{.5\textwidth}
\centering
\small{(b)}
\end{minipage}
\end{tabular}
\caption{
(a) Results of negative deviation of $c$ from the sufficient conditions for star topology when the network consists of 20 nodes
and 
(b) A near-star network
}
\label{fig:star_costneg}
%\vspace{-.25cm}
\end{figure}

The reader should note the difference in labels on the X and Y axes of the different plots in this paper.

%\subsection{Specific Results for Deviation with respect to $c$}
\subsection[Results for Deviation with Respect to ${c}$]{Results for Deviation with Respect to $\boldsymbol{c}$}
\label{sec:devcost}

%\noindent
\subsubsection*{{Negative deviation of $\boldsymbol{c}$ from sufficient conditions for star network:}}
These results are shown qualitatively in Table~\ref{tab:devstarcostneg} and quantitatively in Figure~\ref{fig:star_costneg}(a). 
%The graph in 
Figure~\ref{fig:star_costneg}(a) plots the deviation from network as observed for a network with 20 nodes, if the conditions were deviated at a given deviation node.
For deviation nodes 2 and 3, no deviation in network was observed. 
For other deviation nodes, Table~\ref{tab:devstarcostneg} shows the type of result obtained owing to deviation from sufficient conditions on $c$ at a deviation node, following which, the values of $\gamma$, $c_0$ and $c$ are restored to one of \{$L,M,H$\}. 
The results are invariant with respect to the restored value of $c_0$.
The table shows that $\gamma=L$ coupled with $c=H$, and $\gamma=M$ coupled with $c=M \text{ or } H$, give the best results, where the star topology is restored as per result (B). 
$\gamma=L$ coupled with $c=M$, and $\gamma=H$ coupled with $c=M \text{ or } H$, give decent results for practical purposes, where a near-star network (Figure
\ref{fig:star_costneg}(b)) 
%\ref{fig:devstarc0}(a), 
is obtained as per result (C). 
$c=L$ is unacceptable and should be avoided by network owner desiring to form a star network, as these values are not robust to deviations from sufficient conditions.
Typical observations 
%of all these results 
are shown in Figures~\ref{fig:resulttypes}(a-b). %\\

\begin{figure*} [t!]
\begin{tabular}{c}
\hspace{-.7cm}
\begin{minipage}{.5\textwidth}
\centering
\iftoggle{clr}{
\includegraphics[scale=0.62]{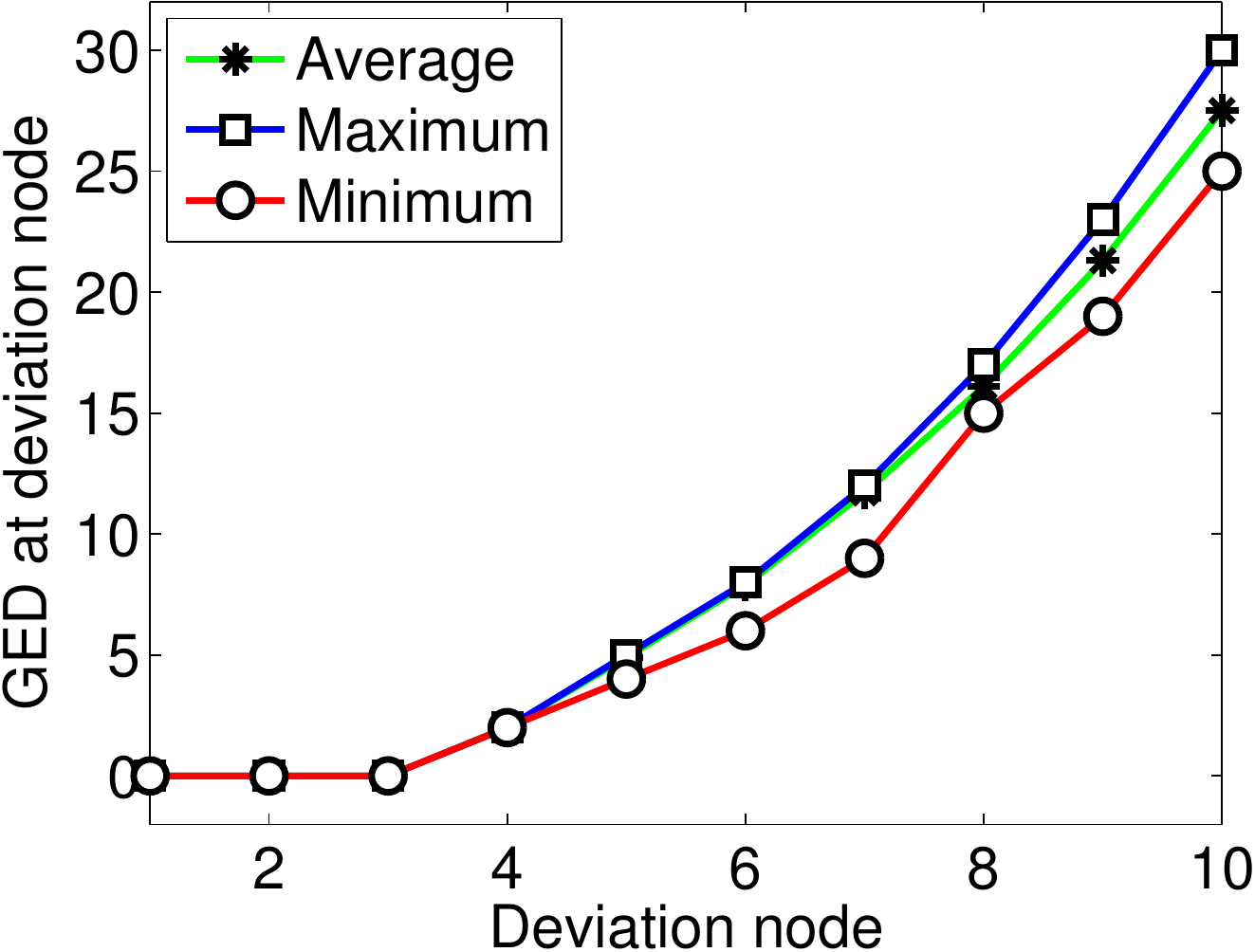}
}{
\includegraphics[scale=0.62]{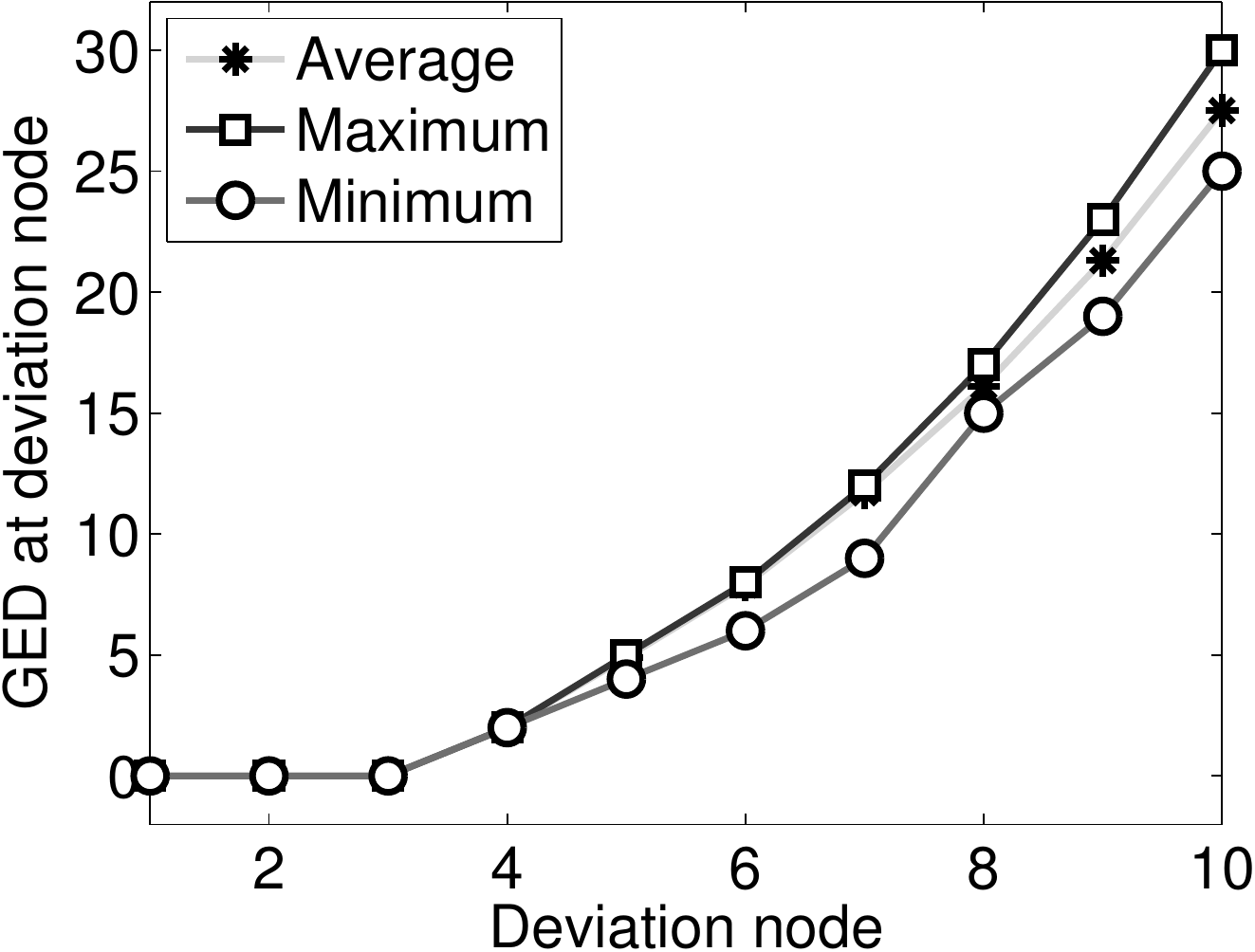}
}
\end{minipage}
\begin{minipage}{.5\textwidth}
\centering
\iftoggle{clr}{
\includegraphics[scale=0.62]{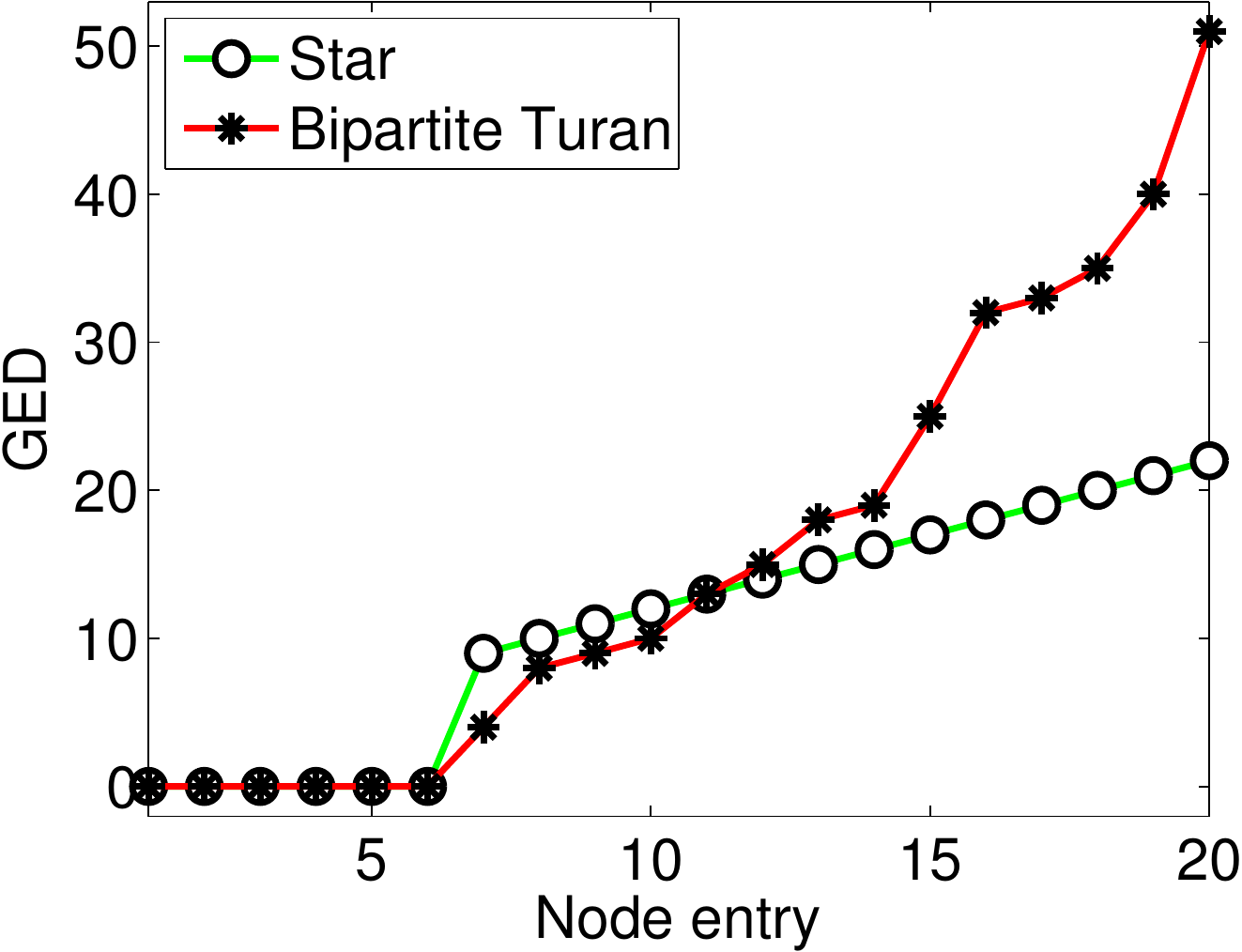}
}{
\includegraphics[scale=0.62]{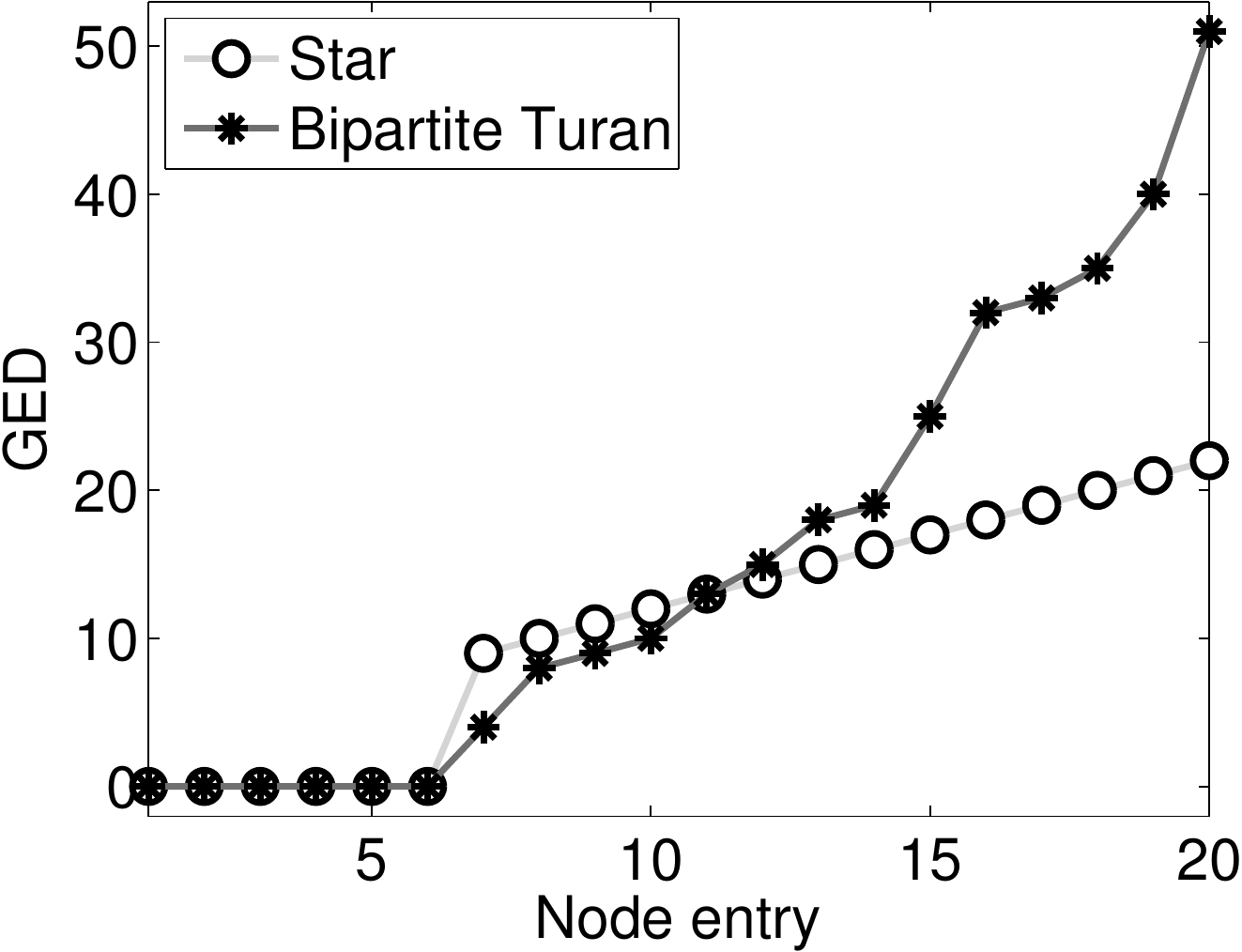}
}
\end{minipage}
\\
%\hspace{-.5cm}
\begin{minipage}{.5\textwidth}
\centering
\small{(a)}
\end{minipage}
\begin{minipage}{.5\textwidth}
\centering
\small{(b)}
\end{minipage}
\end{tabular}
\caption{(a) Results of positive deviation of $c$ from the sufficient conditions for complete network and
(b) Comparison between result (D) for star network and bipartite Tur\'an network for deviation node 7
(Y-axis gives the deviation when the network consists of number of nodes given on X-axis)
}
\label{fig:devc}
%\vspace{-.25cm}
\end{figure*}

\begin{table}[t]
\centering
%  \begin{tabular}{| c || c | c | c || c | c | c || c | c | c |}
  %\begin{tabular}{ p{2.2cm}  p{1cm}  p{1cm}  p{1.5cm}  p{1cm}  p{1cm}  p{1.5cm} p{1cm}  p{1cm}  p{1.5cm} }
  \begin{tabular}{ p{2.2cm}  p{.25cm}  p{.25cm}  p{.75cm}  p{.25cm}  p{.25cm}  p{.75cm} p{.25cm}  p{.25cm} p{.64cm} }
  \hline  \hline
 \T \B 
 & \multicolumn{3}{c}{\hspace{-.3cm}$\gamma=L$}  & \multicolumn{3}{c}{\hspace{-.3cm}$\gamma=M$}  & \multicolumn{3}{c}{\hspace{-.3cm}$\gamma=H$} \\ \hline
  \backslashbox{$c_0$}{$c$} & $L$ & $M$ & $H$  & $L$ & $M$ & $H$  & $L$ & $M$ & $H$ 
  \\ \hline %\hline
\T \B $L/M/H$ &
D & C & B &
D & B & B &
D & C & C
%\\ %\hline
%\T \B $M$ &
%D & C & B &
%D & B & B &
%D & C & C
%\\ %\hline
%\T \B $H$ &
%D & C & B &
%D & B & B &
%D & C & C
\\ \hline  \hline
  \end{tabular}
  %\vspace{-.25cm}
  \caption{Results of negative deviation of $c$ for star network
  }
  \label{tab:devstarcostneg}
\end{table}

%\noindent
\subsubsection*{{Positive deviation of $\boldsymbol{c}$ from sufficient conditions for star network:}}
No node enters the network at deviation node 2, while for all other deviation nodes, the network does not deviate at all and so result (A) is obtained. The same is clear from the derivation of sufficient conditions for star network, that entry of node 2 is the bottleneck on the upper bound for $c$ ($c<b_1$). So node 2 stays out of the network until the sufficient conditions are restored so that they are favorable for it to enter the network, and hence the network builds up as desired. 
%All these observations belong to result (A) and 
These results are desirable if the network owner is not too concerned about the delay of node 2's entry into the network. %\\

%\noindent
\subsubsection*{{Positive deviation of $\boldsymbol{c}$ from sufficient conditions for complete network:}}
No deviation in network was observed for deviation nodes 2 and 3. 
For other deviation nodes, deviations in network were observed only during the entry of the deviation node until the stabilization of the network henceforth (Figure~\ref{fig:devc}(a)). Following this, the sufficient conditions were restored and the network regained the desired topology, after the entry of the node following the deviation node and the stabilization henceforth (result (B)), since the condition $c<b_1-b_2$ ensures that the network so formed has diameter at most 1 (Proposition~\ref{thm:smallworld}), and this is irrespective of the preceding network states. %\\
%that is, it is a complete network.
%The graph edit distance from the complete network at a given deviation node, owing to positive deviation of $c$, is plotted in Figure~\ref{fig:devc}(a).\\

%\noindent
\subsubsection*{{Negative deviation of $\boldsymbol{c}$ from sufficient conditions for bipartite Tur\'an network:}}
The desired network was obtained for all deviation nodes except 4, as clear from the derivation of  sufficient conditions (the 4-node network is the bottleneck for the lower bound on $c$).
%($c>b_1-b_2+\gamma(3b_2-b_3)$). 
For deviation node 4, GED between the resulting network of 4 nodes and the corresponding bipartite Tur\'an network was 3. The topology was restored from the entry of the following node onwards in most instances, while it took up to 9 node entries for some.
%to settle back to a bipartite Tur\'an network. %\\

%\noindent
\subsubsection*{{Positive deviation of $\boldsymbol{c}$ from sufficient conditions for bipartite Tur\'an network:}}
%GED between the resulting network and the corresponding bipartite Tur\'an network was 0 
No deviation in network was observed 
for deviation nodes 2 to 5. However, deviation node 6 onwards, result (D) was observed regularly for all combinations of values \{$L,M,H$\} assigned to $\gamma$, $c_0$ and $c$, apart from when nodes take decisions in a particular order (in which case, no deviation was observed).
For each deviation node 6 onwards, the average GED when the network reached the size of 20 nodes was around 50 and was increasing rapidly as shown in Figure~\ref{fig:devc}(b).
%\footnote{
%Note that 
This GED is expected to be more than that in the case of star network, owing to its relatively high edge density.
%}
Such deviations from the desired network were observed even for extremely minor deviations of $c$ from the derived sufficient conditions.
So restoring the sufficient conditions is not a viable solution for this case.
The network owner should ensure that the values of $c$ are on the lower side so as to stay away from the upper bound. %\\

%\noindent
\subsubsection*{{Negative deviation of $\boldsymbol{c}$ from sufficient conditions for $k$-star network:}}
GED for all deviation nodes were strictly positive and monotonically increasing, qualitatively looking like result (D) in Figure~\ref{fig:resulttypes}(a). %\\

%\noindent
\subsubsection*{{Positive deviation of $\boldsymbol{c}$ from sufficient conditions for $k$-star network:}}
Result (A) was observed for all deviation nodes except $2k$ through $3k-1$. 
The reason for the deviation in network for these deviation nodes is that, in the $k$-star network consisting of number of nodes between $2k$ and $3k-1$, both inclusive, there exists at least one center with only one leaf node linked to it.
When there is a positive deviation of $c$ from the sufficient conditions for $k$-star network, it is beneficial for any other center to delete link with a center that is linked to only one leaf node, and this link deletion leads to other link alterations among other nodes, thus deviating the network from the desired topology.
For deviation nodes $2k$ through $3k-1$, result (D) was observed consistently, which qualitatively looked like the one in Figure~\ref{fig:resulttypes}(a).

%

%\subsection{Specific Results for Deviation with respect to $c_0$}
\subsection[Results for Deviation with Respect to ${c_0}$]{Results for Deviation with Respect to $\boldsymbol{c_0}$}
\label{sec:devc0}

%\noindent
\subsubsection*{{Positive deviation of $\boldsymbol{c_0}$ from sufficient conditions for star network:}}
These results are shown qualitatively in Table~\ref{tab:devstarc0pos} and quantitatively in Figure~\ref{fig:star_c0pos}(a). 
The graph in Figure~\ref{fig:star_c0pos}(a) plots the deviation from network as observed when the network reached the size of 20 nodes, if the conditions were deviated at a given deviation node.
For deviation nodes 2 and 3, no deviation in network was observed. 
For other deviation nodes, Table~\ref{tab:devstarc0pos} shows the type of result obtained owing to deviation from sufficient conditions on $c_0$ at a deviation node, following which, the values of $\gamma$, $c_0$ and $c$ are restored to one of \{$L,M,H$\}. 
%
%
%
\begin{comment} %SeptEdit
The results are invariant with respect to the restored value of $c_0$.
%
It is clear from the table that low or moderate values of $\gamma$ coupled with moderate or high values of $c$, give the best results, where the star topology is resumed as per result (B). 
High values of $\gamma$ coupled with moderate or high values of $c$, give decent results for practical purposes, where a near-star network is obtained as per result (C). 
%
Low values of $c$ are unacceptable and should be avoided by any network owner desiring to form a star network, because these values are not robust to deviations from sufficient conditions.
\end{comment} %SeptEdit
When the sufficient conditions are restored to low values of $c$ after deviating from the sufficient conditions, the resulting network is a $(2,n-2)$-complete bipartite network (result (D)) similar to that in Figure~\ref{fig:star_c0pos}(b), where node $Y$ was the original center and the conditions were deviated during entry of node $X$. %\\

\begin{figure} [t!]
\begin{tabular}{c}
\hspace{-.7cm}
\begin{minipage}{.5\textwidth}
\centering
\iftoggle{clr}{
\includegraphics[scale=0.62]{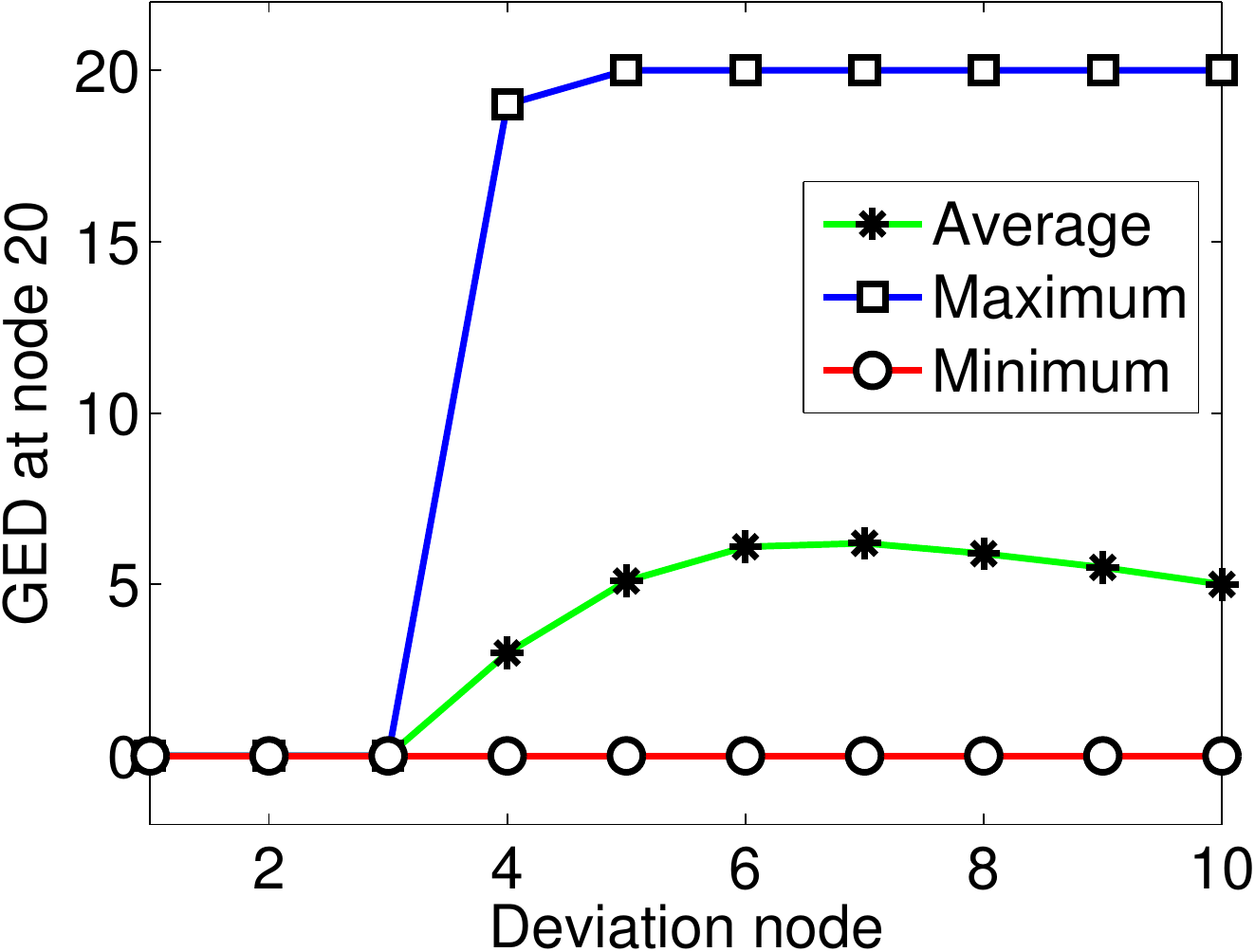}
}{
\includegraphics[scale=0.62]{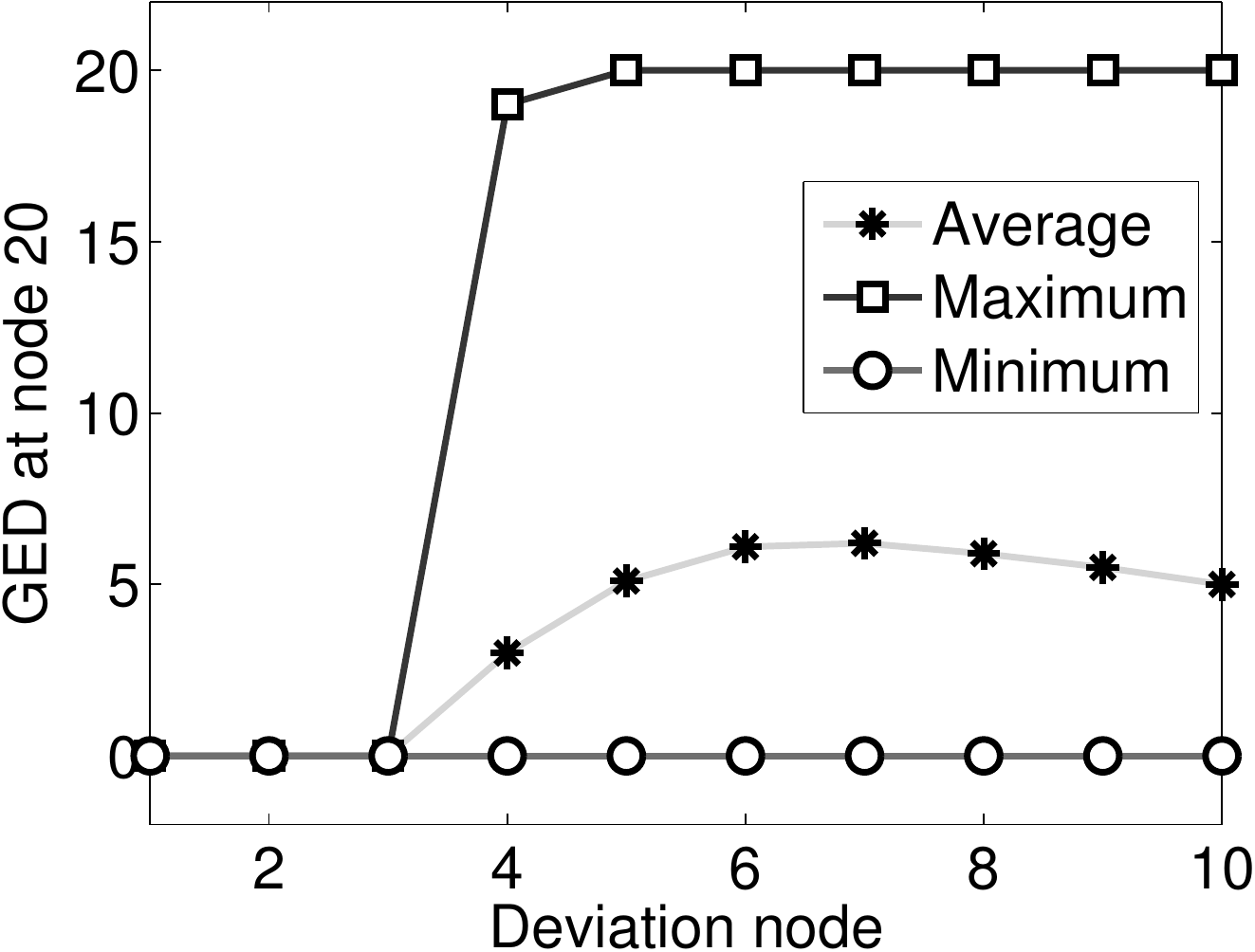}
}
\end{minipage}
\begin{minipage}{.5\textwidth}
\centering
\iftoggle{clr}{
    \includegraphics[scale=0.85]{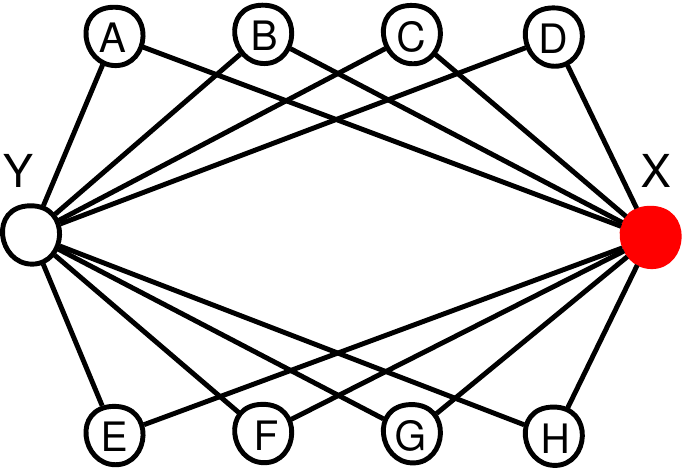}
}{
    \includegraphics[scale=0.85]{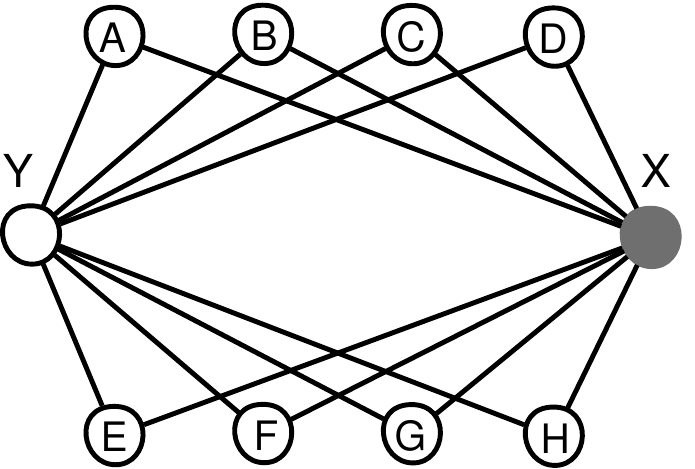}
}
\end{minipage}
\\
%\hspace{-.5cm}
\begin{minipage}{.5\textwidth}
\centering
\small{(a)}
\end{minipage}
\begin{minipage}{.5\textwidth}
\centering
\small{(b)}
\end{minipage}
\end{tabular}
\caption{
(a) Results of positive deviation of $c_0$ from the sufficient conditions for star network when the network consists of 20 nodes
and 
(b) A $(2,8)$-complete bipartite network
}
\label{fig:star_c0pos}
%\vspace{-.25cm}
\end{figure}

\begin{table}[b]
\centering
  %\begin{tabular}{| c || c | c | c || c | c | c || c | c | c |}
  %\begin{tabular}{ p{2.2cm}  p{1cm}  p{1cm}  p{1.5cm}  p{1cm}  p{1cm}  p{1.5cm} p{1cm}  p{1cm}  p{1.5cm} }
  \begin{tabular}{ p{2.2cm}  p{.25cm}  p{.25cm}  p{.75cm}  p{.25cm}  p{.25cm}  p{.75cm} p{.25cm}  p{.25cm} p{.64cm} }
    \hline  \hline
   \T \B 
   & \multicolumn{3}{c}{\hspace{-.3cm}$\gamma=L$}  & \multicolumn{3}{c}{\hspace{-.3cm}$\gamma=M$}  & \multicolumn{3}{c}{\hspace{-.3cm}$\gamma=H$} \\ \hline
  \backslashbox{$c_0$}{$c$} & $L$ & $M$ & $H$  & $L$ & $M$ & $H$  & $L$ & $M$ & $H$ 
  \\ \hline %\hline
\T \B $L/M/H$ &
D & B & B &
D & B & B &
D & C & C
%\\ %\hline
%\T \B $M$ &
%D & B & B &
%D & B & B &
%D & C & C
%\\ %\hline
%\T \B $H$ &
%D & B & B &
%D & B & B &
%D & C & C
\\ \hline \hline
  \end{tabular}
  %\vspace{-.25cm}
  \caption{Results of positive deviation of $c_0$ for star network
  }
  \label{tab:devstarc0pos}
\end{table}

%\noindent
\subsubsection*{{Positive deviation of $\boldsymbol{c_0}$ from sufficient conditions for complete and bipartite Tur\'an networks:}}
No deviation was observed for early deviation nodes, that is, if the conditions were deviated when the network consisted of less number of nodes.
Let $d_T$ be the degree of the node to which a new node desires to connect in order to enter the network. For both complete and bipartite Tur\'an networks, beyond a certain limit on the number of nodes, the minimum value of $d_T$ is very high.
So during positive deviation of $c_0$, the term $d_T((1-\gamma)b_2 - c_0)$ becomes extremely negative, overpowering other benefits, thus making it undesirable for a new node to enter the network. A new node enters once the sufficient conditions are restored.
%These observations belong to result (A) and are the most 
These results are desirable if the network owner is not concerned about the delay of node entry.
% into the network. %\\

%\noindent
\subsubsection*{{Negative deviation of $\boldsymbol{c_0}$ from sufficient conditions for bipartite Tur\'an network:}}
The desired network was obtained for all odd numbered deviation nodes and deviation node 2. For deviation node 4, GED between the resulting network of 4 nodes and the corresponding bipartite Tur\'an network was 3. For most instances, the topology was restored from the entry of the following node onwards; but some instances took up to 9 node entries to settle back to a bipartite Tur\'an network (very similar to the case of negative deviation of $c$). 
For every even-numbered deviation node $n \geq 4$, deviations in network were observed only during the entry of the deviation node until the stabilization of the network henceforth, with GED $=n-1$. Following this, the sufficient conditions were restored and the network regained the desired topology, after the entry of the node following the deviation node and the stabilization henceforth. 
Figure~\ref{fig:devc0}(a) shows the result when node $X$ tries to enter the bipartite Tur\'an network consisting of nodes $A,B,C,D,E$, as the $6^{th}$ node, during negative deviation of $c_0$. It creates links with nodes $B,D$ instead of $A,C,E$, thus giving graph edit distance of 5. Following this, the sufficient conditions are restored and the following node $X+1$ 
%tries to enter, which 
forms links with low degree nodes, forming a bipartite Tur\'an network of 7 nodes, thus restoring the topology. %\\

%\noindent
\subsubsection*{{Negative deviation of $\boldsymbol{c_0}$ from sufficient conditions for $k$-star network:}}
For deviation node $n$ such that $(n \mod k)=1$, the network did not deviate and so result (A) was observed. 
For all other deviation nodes, result (B) was observed.
In general, for deviation node $n$, GED was observed to be 2, and it took $\left[ (k+1-z)\mod k \right]$ node entries for the topology to be restored once the sufficient conditions were restored, where $z = (n \mod k)$.
Figure~\ref{fig:devc0}(b) shows the result when node $X$ tries to enter the 3-star network consisting of nodes $A$ through $J$, as the $11^{th}$ node, during negative deviation of $c_0$.
It creates a link with node $A$ instead of either $B$ or $C$, thus giving GED of 2. Following this, the sufficient conditions are restored and so the following node $X+1$ 
%tries to enter, which 
forms links with a lowest degree center, say $C$; but GED remains 2. 
Then the next node $X+2$ tries to enter, which forms a link with the only lowest degree center $B$, forming a 3-star network of 13 nodes, thus restoring the topology.
In this example, $k=3$ and $n=11$ and so it takes 2 node entries for the topology to be restored. %\\

%\noindent
\subsubsection*{{Positive deviation of $\boldsymbol{c_0}$ from sufficient conditions for $k$-star network:}}
%RELOOK
Let $C$ be a center with the lowest degree and $m_j$ be the number of leaf nodes already connected to center $j$.
It can be shown that result (A) will be obtained if the positive deviation of $c_0$ is less that the threshold:
\begin{equation}
\nonumber
(b_3-b_4)\left( \frac{\sum_{j \neq C}m_j + \textbf{I}_{n \neq pk+1}}{k+m_C-2}-1 \right)
\end{equation}
where $n$ is the deviation node, and $\textbf{I}_{n \neq pk+1}$ is 1 if $n \neq pk+1$ for any integer $p$, else it is 0.
If the deviation crossed this threshold in simulations, result (D) was observed consistently, which qualitatively looked like the one in Figure~\ref{fig:resulttypes}(a).
The result is owing to the fact that a high value of $c_0$ would force a new node to prefer connecting to a leaf node which is linked to a center with the highest degree, rather than any center directly; this leads to other link alterations among other nodes, thus deviating the network from the desired topology.

\begin{figure} [t]
\begin{tabular}{c}
%\hspace{-.2cm}
\begin{minipage}{.4\textwidth}
\centering
\iftoggle{clr}{
\includegraphics[scale=0.87]{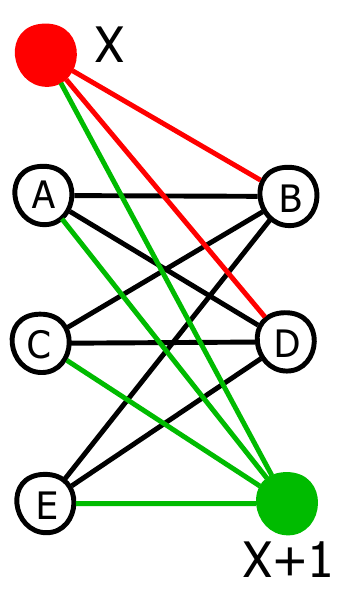}
}{
\includegraphics[scale=0.87]{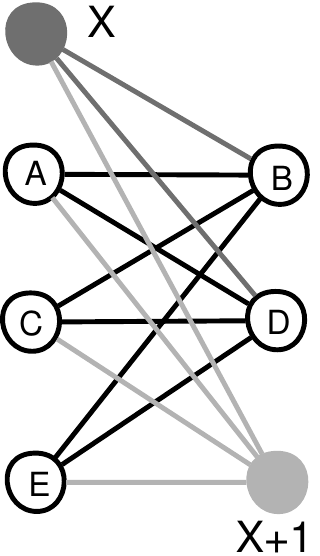}
}
\end{minipage}
\begin{minipage}{.6\textwidth}
\centering
\iftoggle{clr}{
\includegraphics[scale=0.82]{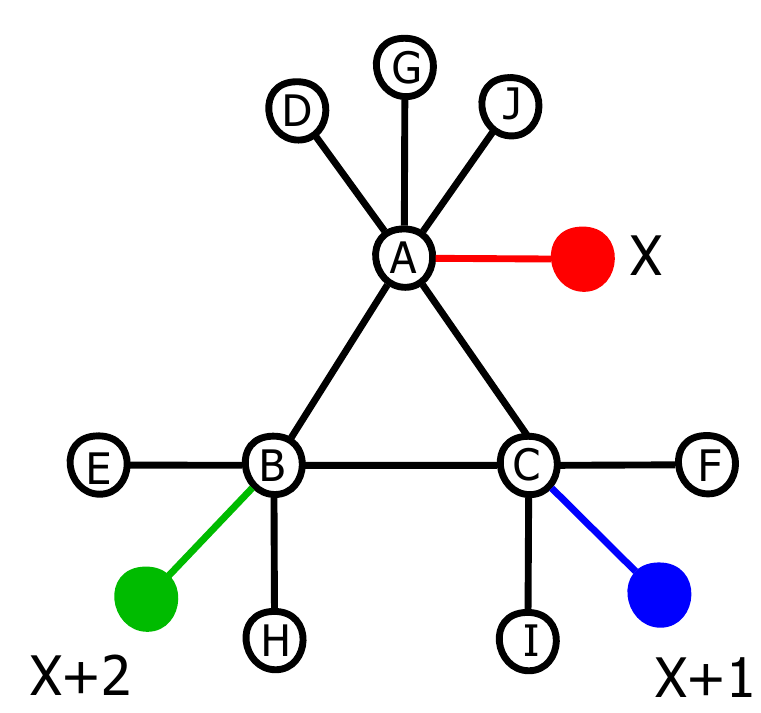}
}{
\includegraphics[scale=0.82]{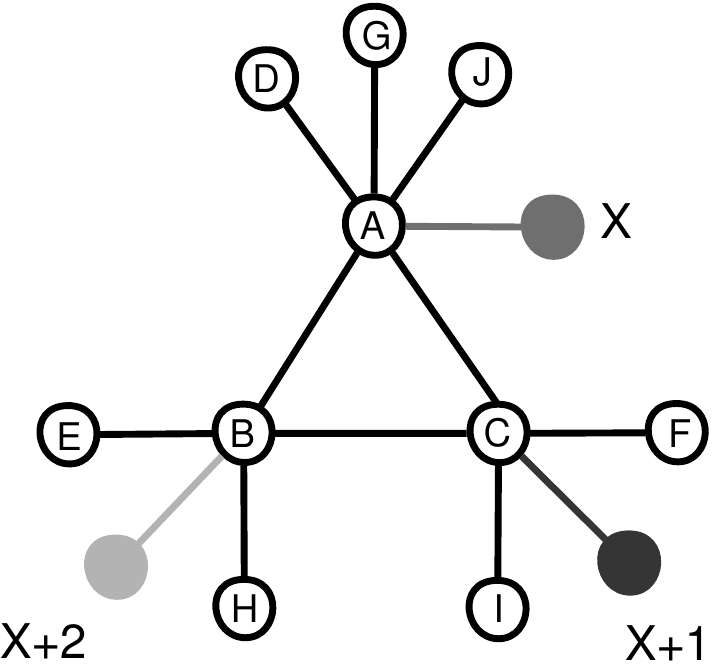}
}
\end{minipage}
\\
%\hspace{-.2cm}
\begin{minipage}{.4\textwidth}
\centering
(a)
\end{minipage}
\begin{minipage}{.6\textwidth}
\centering
(b)
\end{minipage}
\end{tabular}
\caption{Restorations of (a) bipartite Tur\'an and (b) 3-star network topologies
%Results of (a) negative deviation of $c_0$ from the sufficient conditions for bipartite Tur\'an network and
%(b) negative deviation of $c_0$ from the sufficient conditions for 3-star
}
\label{fig:devc0}
%\vspace{-.25cm}
\end{figure}

\section{Conclusion}
\label{sec:conclusion_nfsc}

We proposed a model of recursive network formation where nodes enter a network sequentially, thus triggering evolution 
%of the network 
each time a new node enters.
We considered a sequential move game model with myopic nodes under a very general utility model, and pairwise stability as the equilibrium notion; however the proposed model (Figure~\ref{fig:model}) is independent of the network evolution model, the equilibrium notion, as well as the utility model.
The recursive nature of our model enabled us to analyze the network formation process using an elegant induction-based technique.
%Given a network topology, 
For each of the relevant topologies, by directing network evolution as desired, we derived sufficient conditions 
%using the proposed utility model
under which that topology uniquely emerges. 
%for the formation of that topology.
% along a desired improving path.
% in the 
%sequential move
 %game tree.
  The derived conditions suggest that conditions on network entry impact degree distribution, while
   conditions on link costs impact density;
   also there arise constraints on intermediary rents owing to contrasting densities of connections in the desired topology.
We then analyzed the social welfare properties of the considered topologies,
% of relevant topologies, 
and studied the effects of deviating from the derived conditions.
% using simulations. 

\begin{comment} %SeptEdit
A question to investigate is what happens if the network gets disconnected into multiple components if the values of the cost parameters rise significantly over time. A network owner is in the best position to answer this question since he may want to decide the outcome based on his objectives. He may want to try to restore the sufficient conditions at the earliest, he may want to completely remove smaller components from the network, or he may want to go ahead with the network formation process by treating the different components as different networks.
\end{comment} %SeptEdit

\vspace{-.2cm}

%\vspace{-.05cm}

\section*{Acknowledgments}
%\noindent
The original publication appears in Studies in Microeconomics, volume 3, number 2, pages 158-213, 2015, and is available at \href{http://journals.sagepub.com/doi/pdf/10.1177/2321022215588873}{journals.sagepub.com}.
%%In case of acceptance of this paper, the source code will be made available on \texttt{lcm.csa.iisc.ernet.in} for community contribution.
A previous preliminary, concise version of this paper is published in Proceedings of The 8th International Conference on Internet \& Network Economics, 2012 \cite{dhamalwine} and 
%
%A previous preliminary, concise version of this paper,
%{\em Forming networks of strategic agents with desired topologies}~\cite{dhamalwine},
is available at  \href{http://link.springer.com/chapter/10.1007/978-3-642-35311-6_39}{link.springer.com}.
The authors thank Rohith D. Vallam and Prabuchandran K.J. for useful discussions.

%\balance
\bibliographystyle{plain}
\bibliography{NFSC_SiM_arXiv_references}

\newpage

% SUPPLEMENRTARY MATERIAL

\appendix
\numberwithin{equation}{section}
\section*{APPENDIX}

\section{Proof of Proposition~\ref{thm:smallworld}
}
\label{app:smallworld}
%\noindent
\begin{customprop}{\ref{thm:smallworld}}
%\label{thm:smallworld}
For a network, if $c<b_1-b_{d+1}$ ($d \geq 1$) and $c_0\leq(1-\gamma)b_2$, the resulting diameter is at most $d$.
\end{customprop}
\begin{proof}
The conditions $c<b_1$ and $c_0\leq(1-\gamma)b_2$ ensure that any new node successfully enters the network, that is, it gets a positive utility by doing so, and the node to which it connects to in order to enter the network, also gets a higher utility.

Now consider a network where $c<b_1-b_{d+1}$ and there exist two nodes, say $A$ and $B$, which are at a distance $x>d$ from each other. The indirect benefit they get from each other is  $b_x \leq b_{d+1}$. In the case where there exist essential nodes connecting these nodes, each has to pay an additional rent of $\gamma b_{x}$. By establishing a connection between them, each node gets an additional direct benefit of $b_1$ and incurs an additional cost $c$. Also this connection may decrease the distances between either of these nodes and other nodes, for instance, direct neighbors of node $B$ which were at distance $b_{x-1}$, $b_x$ or $b_{x+1}$ from node $A$, are now at distance  $\min\{b_2,b_{x-1}\}$, resulting in increase in indirect benefits for node $A$.

It can be easily seen that if either (or both) of these nodes acted as an essential node for some pair of nodes, it remains to do so even after the connection is established. Furthermore, it is possible that the established connection shortens the path between this pair, resulting in higher bridging benefits for the node under consideration.

Summing up, by establishing a mutual connection between nodes which are at distance $x>d$ from each other, the overall increase in utility for either node is at least $b_1-c$ and the overall decrease is at most $b_{d+1}$. So the condition sufficient for link creation is $b_1-c>b_{d+1}$. As this is true for any such pair, without loss of generality, the network will evolve until distance between any pair is at most $d$.
%\qed 
\end{proof}

%The following corollary results when $d=1$.
%
%
%\begin{customcor}{\ref{thm:complete}}
%%\label{thm:complete}
%For a network, if $c < b_1-b_2$ and $c_0 \leq \left( 1-\gamma \right) b_2$, the unique resulting topology is a complete graph.
%\end{customcor}

\section{Proof of Proposition~\ref{thm:bipartite}
}
\label{app:bipartite}
%\noindent
\begin{customprop}{\ref{thm:bipartite}}
%\label{thm:bipartite}
For a network with $\gamma <   \frac{b_2 - b_3}{3b_2 - b_3} $, if $b_1-b_2+ \gamma \left( 3b_2 - b_3 \right) <  c < b_1 - b_3$ 
and $\left( 1-\gamma \right) \left( b_2-b_3 \right) < c_0 \leq \left( 1-\gamma \right) b_2$, the unique resulting topology is a 
bipartite Tur\'an graph.
\end{customprop}
\begin{proof}
We first derive conditions for ensuring pairwise stability of a bipartite Tur\'an network, that is, assuming that such a network is formed, what conditions are required so that there are no incentives for any two unconnected nodes to create a link between them and for any node to delete any of its links. Note that these conditions can be integrated in the later part of the proof within different scenarios that we consider.\\
In what follows, $p_1$ is the size of the partition constituting the node taking its decision, $p_2$ is the size of the other partition and $n=p_1+p_2$ is the number of nodes in the network.
We need to consider cases for some discretely small number of nodes owing to the nature of essential nodes, after which, the analysis holds for arbitrarily large number of nodes. For brevity, we present the analysis for the base case and a generic case in each scenario, omitting presentation of discrete cases. \\

\noindent
\textbf{No two nodes belonging to the same partition should create a link between them:} Their utility should not increase by doing so. This is not applicable for $n=2$. \\
For $n=3$, 
\begin{equation}
\nonumber
2(b_1-c) \leq b_1-c+(1-\gamma)b_2
%\vspace{-5mm}
\end{equation}
\begin{equation}
\label{E16}
\iff c \geq b_1-b_2+\gamma b_2
\end{equation}
For $n\geq 4$,
\begin{equation}
\nonumber
(p_2+1)(b_1-c)+(p_1-2)b_2 \leq p_2(b_1-c)+(p_1-1)b_2
%\vspace{-5mm}
\end{equation}
\begin{equation}
\nonumber
\iff c \geq b_1-b_2
\end{equation}
which is a weaker condition that Inequality~(\ref{E16}).\\

\noindent
\textbf{No node should delete its link with any node belonging to the other partition:} That is, their utility should not increase by doing so. \\
For $n=2$,
\begin{equation}
\nonumber
0\leq b_1-c
%\vspace{-5mm}
\end{equation}
\begin{equation}
\label{E17for2}
\iff c \leq b_1
\end{equation}
For $n\geq 6$,
\begin{equation}
\nonumber
(p_2-1)(b_1-c)+(p_1-1)b_2+b_3 \leq p_2(b_1-c)+(p_1-1)b_2
%\vspace{-5mm}
\end{equation}
\begin{equation}
\label{E17}
\iff c \leq b_1-b_3
\end{equation}
It can be shown that conditions for the discrete cases $n=3,4,5$  are satisfied by Inequality~(\ref{E17}).\\  \\
In the process of formation of a bipartite Tur\'an network, at most four different types of nodes exist at any point in time.
  \begin{center}
  %\begin{small}
\begin{tabular}{l l }
    \hline \hline
\T \B  
I & newly entered node \\ \hline
\T \B  
II & nodes connected to the newly entered node\\ \hline
\T \B  
III & nodes in the same partition as Type II nodes, but not connected to newly entered node\\ \hline
\T \B  
IV & rest of the nodes\\ \hline \hline
  \end{tabular}
  %\end{small}
  \end{center}
The notation we use while deriving the sufficient conditions are as follows:
\begin{center}
%\begin{small}
\begin{tabular}{l l}
    \hline  \hline
\T \B 
$k$ & number of nodes of Type II \\ \hline
\T \B 
$n$ & number of nodes in network, including new node\\ \hline
\T \B 
$m_1$ & number of nodes of Types II and III put together\\ \hline
\T \B 
$m_2$ & number of nodes of Type IV\\ \hline \hline
  \end{tabular}
  %\end{small}
\end{center}

\noindent
\textbf{For the newly entering node to enter the network:} Its utility should be positive after doing so. Also, in case of even $n$, for the new node to be a part of the smaller partition, its first connection should be a node belonging to the larger partition. So for $k=0$, we have \\
For $n\geq 2$,
%\vspace{-.3cm} 
\begin{equation}
\nonumber
b_1-c+ \lceil \frac{n}{2}-1 \rceil \left( (1-\gamma)b_2 - c_0 \right) + \lfloor \frac{n}{2}-1 \rfloor (1-\gamma)b_3 >0
%%\vspace-.1cm} 
\end{equation}
It can be seen that the condition is the strongest when $n=2$ whenever
%%\vspace-.2cm} 
\begin{equation}
\label{E1b}
c_0 \leq (1-\gamma)b_2
%%\vspace-.1cm} 
\end{equation}
The condition thus becomes
%%\vspace-.2cm} 
\begin{equation}
\nonumber
c< b_1
%%\vspace-.1cm} 
\end{equation}
which is satisfied by Inequality~(\ref{E17}).\\

\noindent 
\textbf{The utility of a node in the larger partition, whenever applicable, should not decrease after accepting link from the new node:}\\
For $n=2$,
%%\vspace-.3cm} 
\begin{equation}
\nonumber
b_1-c \geq 0
%\vspace-5mm}
\end{equation}
\begin{equation}
\nonumber
\iff c \leq b_1
%%\vspace-.1cm} 
\end{equation}
For $n \geq 5$,
%%\vspace-.1cm} 
\begin{align}
\nonumber
\begin{split}
&
\lceil \frac{n}{2} \rceil (b_1 -c) + \lfloor \frac{n}{2}-1 \rfloor b_2 + \gamma \lceil \frac{n}{2}-1 \rceil 2 b_2 + \gamma \lfloor \frac{n}{2}-1 \rfloor 2 b_3 \\
&\geq \lceil \frac{n}{2}-1 \rceil (b_1 -c) + \lfloor \frac{n}{2}-1 \rfloor b_2
\end{split}
%%\vspace-.1cm} 
\end{align}
\begin{equation}
\nonumber
\iff c \leq b_1+ 2\gamma \lceil \frac{n}{2}-1 \rceil b_2 + \lfloor \frac{n}{2}-1 \rfloor b_3
\end{equation}
The conditions for these as well as the discrete cases $n=3,4$ are satisfied by Inequality~(\ref{E17}).\\

\noindent
\textbf{The new node should connect to a node in the larger partition, whenever applicable:} One way to see this is by ensuring that this strategy strictly dominates connecting to a node in the smaller partition.
This scenario arises for even values of $n\geq 4$.
%%\vspace-.1cm} 
\begin{align}
\begin{split}
\nonumber
&
b_1-c+ \left( \frac{n}{2}-1 \right) \left( (1-\gamma)b_2 - c_0 \right) + \left( \frac{n}{2}-1 \right) (1-\gamma)b_3 \\
&> b_1-c+ \left( \frac{n}{2} \right) \left( (1-\gamma)b_2 - c_0 \right) + \left( \frac{n}{2}-2 \right) (1-\gamma)b_3 
\end{split}
\end{align}
\begin{equation}
\label{E2a}
\iff c_0 > (1-\gamma)(b_2-b_3)
%%\vspace-.1cm} 
\end{equation}
An alternative condition would be such that the utility of a node in the smaller partition decreases if it accepts the link from the new node, thus forcing the latter to connect to a node in the other partition. But it can be seen that this condition is inconsistent with Inequality~(\ref{E17}) and so we use Inequality~(\ref{E2a}) to meet our purpose.\\

\noindent
\textbf{Type I node should prefer connecting to a Type III node, if any, than remaining in its current state:}
For $k\geq 2$, this scenario does not arise for $n<6$. 
For $n\geq 6$,
%%\vspace-.3cm} 
%\begin{small}
\begin{equation}
\nonumber
\begin{split}
(k+1)(b_1-c)+m_2b_2+(m_1-k-1)b_3 > k(b_1-c)+m_2b_2
+(m_1-k)b_3
%%\vspace-.1cm} 
\end{split}
\end{equation}
%\end{small}
\begin{equation}
\label{E6}
\iff c<b_1-b_3
%%\vspace-.1cm} 
\end{equation}
Now for $k=1$, this scenario does not arise for $n=2,3$.\\
For $n \geq 4$,
%%\vspace-.3cm} 
\begin{align}
\nonumber
\begin{split}
2(b_1-c)+m_2b_2+(m_1-2)b_3 
> b_1-c+(1-\gamma)m_2b_2
+(1-\gamma)(m_1-1)b_3
\end{split}
\end{align}
%\vspace-5mm}
\begin{equation}
\nonumber
\iff c<b_1-b_3+\gamma(m_2b_2+(m_1-1)b_3)
\end{equation}
Note that as $n \geq 4$, we have $m_1 \geq 2$ and $m_2 \geq 1$ and so the above condition is weaker that Inequality~(\ref{E6}).\\
It is also necessary that utility of Type III node does not decrease on accepting link from Type I node. In fact, when the former gets a chance to move, we derive conditions so that it also volunteers to create a link with the later.\\

\noindent
\textbf{The utility of Type III node should increase if it successfully creates a link with Type I node:}
When $k=1$, the case does not arise for $n=2,3$.\\
For $n \geq 6$,
%%\vspace-.3cm} 
%\begin{small}
\begin{equation}
\nonumber
(m_2+1)(b_1-c)+(m_1-1)b_2>m_2(b_1-c)+(m_1-1)b_2+(1-\gamma)b_3
%\vspace-5mm}
\end{equation}
%\end{small}
\begin{equation}
\nonumber
\iff c<b_1-b_3+\gamma b_3
%%\vspace-.1cm} 
\end{equation}
The conditions obtained from discrete cases $n=4,5$ are weaker than this one.\\
For $k\geq 2$, this case does not arise for $n <6$. \\
For $n\geq 6$,
%%\vspace-.3cm} 
\begin{equation}
\nonumber
(m_2+1)(b_1-c)+(m_1-1)b_2>m_2(b_1-c)+(m_1-1)b_2+b_3
%\vspace-5mm}
\end{equation}
\begin{equation}
\nonumber
\iff c<b_1-b_3
%%\vspace-.1cm} 
\end{equation}
The conditions for all cases are satisfied by Inequality~(\ref{E6}).\\

\noindent
\textbf{Type III node should not delete its link with Type IV node:} This can be assured if this strategy is dominated by its strategy of forming a link with Type I node.
This scenario does not arise for $n=2,3$. 
The conditions for the discrete cases $n=4,5,6$ are weaker than that for $n\geq 7$.\\
For $n\geq 7$,
%%\vspace-.3cm} 
%\begin{small}
\begin{equation}
\nonumber
\begin{split}
(m_2+1)(b_1-c)+(m_1-1)b_2>(m_2-1)(b_1-c)+b_3
+(m_1-1)b_2+(1-\gamma)b_3
%%\vspace-.1cm} 
\end{split}
\end{equation}
%\end{small}
%\vspace-5mm}
\begin{equation}
\nonumber
\iff c<b_1-b_3+\frac{\gamma}{2}b_3
%%\vspace-.1cm} 
\end{equation}
For $k\geq 2$, the cases applicable are $n\geq 6$.
The condition for discrete case $n=6$ is weaker than the following condition.\\
For $n\geq 7$,
%%\vspace-.3cm} 
%\begin{small}
\begin{equation}
\nonumber
\begin{split}
(m_2+1)(b_1-c)+(m_1-1)b_2>(m_2-1)(b_1-c)+b_3+b_3+(m_1-1)b_2
\end{split}
%%\vspace-.1cm} 
\end{equation}
%\end{small}
%\vspace-5mm}
\begin{equation}
\nonumber
\iff c<b_1-b_3
%%\vspace-.1cm} 
\end{equation}
Hence, all conditions for this scenario are satisfied by Inequality~(\ref{E6}).\\

\noindent
\textbf{Type III node should prefer connecting to Type I node than to another Type III node:} This does not arise for $n<6$.
When $k=1$, \\
For $n\geq 6$,
%%\vspace-.3cm} 
%\begin{small}
\begin{align}
\nonumber
\begin{split}
(m_2+1)(b_1-c)+(m_1-1)b_2
>(m_2+1)(b_1-c)
+(m_1-2)b_2+(1-\gamma)b_3
%%\vspace-.1cm} 
\end{split}
\end{align}
%\end{small}
%\vspace-5mm}
\begin{equation}
\nonumber
\iff b_2>(1-\gamma)b_3
%%\vspace-.1cm} 
\end{equation} 
which is always true. For $k\geq 2$,\\
For $n \geq 6$,
%%\vspace-.3cm} 
%\begin{small}
\begin{equation}
\nonumber
(m_2+1)(b_1-c)+(m_1-1)b_2>(m_2+1)(b_1-c)+(m_1-2)b_2+b_3
%%\vspace-.1cm} 
\end{equation}
%\end{small}
%\vspace-5mm}
\begin{equation}
\nonumber
\iff b_2>b_3
%%\vspace-.1cm} 
\end{equation} 
which is always true.\\

\noindent
\textbf{Type IV node should not delete its link with Type III node:} That is, its utility should not increase by doing so.
This does not arise for $n<4$.\\
For $n\geq 7$,
%%\vspace-.3cm} 
\begin{equation}
\nonumber
\begin{split}
(m_1-1)(b_1-c)+(m_2-1)b_2+(1-\gamma)b_2+b_3 \\ \leq m_1(b_1-c)+(m_2-1)b_2+(1-\gamma)b_2
%%\vspace-.1cm} 
\end{split}
\end{equation}
\begin{equation}
\nonumber
\iff  c\leq b_1-b_3
%%\vspace-.1cm} 
\end{equation} 
The conditions for discrete cases $n=4,5,6$ are weaker than the above condition.
For $k\geq 2$, the new cases are $n\geq 6$, where the discrete case $n=6$ result in conditions weaker than the following one.\\
For $n \geq 7$,
%%\vspace-.3cm} 
%\begin{small}
\begin{equation}
\nonumber
(m_1-1)(b_1-c)+(m_2-1)b_2+b_2+b_3 \leq m_1(b_1-c)+(m_2-1)b_2+b_2
%%\vspace-.1cm} 
\end{equation}
%\end{small}
%\vspace-5mm}
\begin{equation}
\nonumber
\iff c\leq b_1-b_3
%%\vspace-.1cm} 
\end{equation} 
It can be seen that all conditions of this scenario are satisfied by Inequality~(\ref{E6}).\\

\noindent
\textbf{Type IV node should also not break its link with Type II node:} That is, its utility should not increase by doing so.
For $k=1$,\\
For $n\geq 6$,
%%\vspace-.3cm} 
\begin{align}
\nonumber
\begin{split}
&
(m_1-1)(b_1-c)+(m_2-1)b_2+(1-\gamma)b_4+(1-\gamma)b_3 \\ 
&\leq m_1(b_1-c)+(m_2-1)b_2+(1-\gamma)b_2 
%%\vspace-.1cm} 
\end{split}
\end{align}
\begin{equation}
\nonumber
\iff c\leq b_1-b_3+(1-\gamma)(b_2-b_4)+\gamma b_3
%%\vspace-.1cm} 
\end{equation} 
The discrete cases $n=3,4,5$ result in weaker conditions than this.
For $k\geq 2$,\\
For $n\geq 6$,
%%\vspace-.3cm} 
%\begin{small}
\begin{equation}
\nonumber
(m_1-1)(b_1-c)+(m_2-1)b_2+b_2+b_3  \leq m_1(b_1-c)+(m_2-1)b_2+b_2 
%%\vspace-.1cm} 
\end{equation}
%\end{small}
%\vspace-5mm}
\begin{equation}
\nonumber
\iff c\leq b_1-b_3
%%\vspace-.1cm} 
\end{equation} 
The conditions are satisfied by Inequality~(\ref{E6}).\\

\noindent
\textbf{Type I node should not propose a link to a Type IV node:} One way is to ensure that this strategy of Type I node is dominated by its strategy to propose a link to a Type III node.
It can be seen that for $k \geq 2$ and $n \geq 6$, this translates to
%%\vspace-.3cm} 
\begin{align}
\nonumber
\begin{split}
&
(k+1)(b_1-c)+m_2b_2+(m_1-k-1)b_3\\ 
&> (k+1)(b_1-c)+(m_2-1)b_2+(m_1-k)b_2
%%\vspace-.1cm} 
\end{split}
\end{align}
\begin{equation}
\nonumber
\iff b_2-b_3>(m_1-k)(b_2-b_3)
%%\vspace-.1cm} 
\end{equation}
which is not true for $m_1>k$.\\
So we look at the alternative condition that the utility of Type IV node decreases if it accepts the link from Type I node, and as Type I node computes this decrease in utility, it will not propose a link to Type IV node.
First, we consider $k=1$. The discrete case $n=4$ gives the following condition.\\
%%\vspace-.2cm} 
\begin{equation}
\nonumber
3(b_1-c)+2\gamma b_2+2\gamma b_2 < 2(b_1-c)+(1-\gamma)b_2+2\gamma b_2+\gamma b_3 
%\vspace-5mm}
\end{equation}
\begin{equation}
\label{E7b}
\iff c>b_1-b_2+\gamma(3b_2-b_3)
%%\vspace-.1cm} 
\end{equation}
The other discrete cases $n=3,5$ result in weaker conditions than the above.\\
For $n\geq 6$,
%%\vspace-.3cm} 
%\begin{small}
\begin{equation}
\nonumber
(m_1+1)(b_1-c)+(m_2-1)b_2<m_1(b_1-c)+(m_2-1)b_2+(1-\gamma)b_2
%\vspace-5mm}
\end{equation}
%\end{small}
\begin{equation}
\nonumber
\iff c>b_1-b_2+\gamma b_2 
%%\vspace-.1cm} 
\end{equation}
which is a weaker condition than Inequality~(\ref{E7b}).
Now for $k \geq 2$, $n=4,5$ correspond to pairwise stability conditions and cases $n<4$ are not applicable.\\
For $n\geq 6$,
%%\vspace-.3cm} 
\begin{equation}
\nonumber
(m_1+1)(b_1-c)+(m_2-1)b_2 < m_1(b_1-c)+(m_2-1)b_2+b_2
%\vspace-5mm}
\end{equation}
\begin{equation}
\nonumber
\iff c>b_1-b_2
%%\vspace-.1cm} 
\end{equation}
which is satisfied by Inequality~(\ref{E7b}).\\

\noindent
\textbf{Type IV node should not propose a link to Type I node:} This scenario is essentially equivalent to the previous one scenario of utility of Type IV node decreasing due to link with Type I node, with the equalities permitted. So these result in weaker and hence no additional conditions.\\

\noindent
\textbf{Type III node should not propose a link to Type II node:} One way is to ensure that for Type III node, connecting to Type II node is strictly dominated by connecting to Type I node.
It can be seen that for $k\geq 2$ and $n\geq 6$, this translates to
%%\vspace-.3cm} 
%\begin{small}
\begin{equation}
\nonumber
(m_2+1)(b_1-c)+(m_1-2)b_2+b_2<(m_2+1)(b_1-c)+(m_1-1)b_2
%%\vspace-.1cm} 
\end{equation}
%\end{small}
which gives $0>0$. 
So we need to use the alternative condition that the utility of Type II node decreases on accepting link from Type III node.
For $k=1$,\\
For $n=4$,
%%\vspace-.7cm} 
\begin{equation}
\nonumber
3(b_1-c)+4\gamma b_2 < 2(b_1-c)+(1-\gamma)b_2+2\gamma b_2+\gamma b_3
%\vspace-5mm}
\end{equation}
\begin{equation}
\nonumber
\iff c > b_1-b_2+\gamma(3b_2 - b_3)
%%\vspace-.1cm} 
\end{equation}
which is same as Inequality~(\ref{E7b}). \\
For $n\geq 5$,
%%\vspace-.3cm} 
%\begin{footnotesize}
\begin{align}
\nonumber
\begin{split}
&
(m_2+2)(b_1-c)+(m_1-2)b_2+2\gamma(m_2+1)b_2+2\gamma(m_1-2)b_3 \\
&< (m_2+1)(b_1-c)+(m_1-1)b_2+2\gamma m_2b_2+ 2\gamma(m_1-1)b_3
\end{split}
\end{align}
%\end{footnotesize}
\begin{equation}
\nonumber
\iff c > b_1-b_2+2\gamma(b_2-b_3)
%%\vspace-.1cm} 
\end{equation}
which is a weaker condition than Inequality~(\ref{E7b}). 
Now for $k \geq 2$, the only new case is the following.\\
For $n\geq 6$,
%%\vspace-.3cm} 
\begin{equation}
\nonumber
(m_2+2)(b_1-c)+(m_1-2)b_2<(m_2+1)(b_1-c)+(m_1-1)b_2
%\vspace-5mm}
\end{equation}
\begin{equation}
\nonumber
\iff c>b_1-b_2
%%\vspace-.1cm} 
\end{equation}
which is satisfied by Inequality~(\ref{E7b}).\\

\noindent
\textbf{Type II node should not propose a link with Type III node:} This is essentially equivalent to the above scenario of utility of Type II node decreasing due to link with Type III node, with the equalities permitted. So these result in weaker and hence no additional conditions.\\

\noindent
\textbf{No Type II node should delete link with Type IV node:} First, we consider $k=1$.\\
For $n \geq 7$, 
%%\vspace-.3cm} 
%\begin{footnotesize}
\begin{align}
\nonumber
\begin{split}
&
m_2(b_1-c)+(b_1-c) + (m_1-1)b_2+2\gamma m_2 b_2 + 2\gamma(m_1-1)b_3 \\
&\geq (m_2-1)(b_1-c) + (b_1-c) + (m_1-1)b_2+b_3
+2\gamma (m_2-1)b_2+2\gamma (m_1-1)b_3 +2\gamma b_4
\end{split}
\end{align}
%\end{footnotesize}
\begin{equation}
\nonumber
\iff c \leq b_1 - b_3 + 2 \gamma (b_2 - b_4)
%%\vspace-.1cm} 
\end{equation}
This as well as all discrete cases $n<7$ are satisfied by Inequality~(\ref{E6}).\\
For $k \geq 2$,  the cases of $n=4,5$ correspond to pairwise stability condition that we have already considered, while cases $n<4$ are not applicable.\\
For $n \geq 6$,
%%\vspace-.3cm} 
\begin{equation}
\nonumber
m_2(b_1-c)+(m_1-1)b_2+b_3 \leq (m_2+1)(b_1-c)+(m_1-1)b_2
%\vspace-5mm}
\end{equation}
\begin{equation}
\nonumber
\iff c \leq b_1-b_3
%%\vspace-.1cm} 
\end{equation}
which is satisfied by Inequality~(\ref{E6}).\\

\noindent
\textbf{Two Type IV nodes should not create a mutual link:} That is their utilities should not increase by doing so.
When $k=1$, it is not applicable for $n<5$. Also, the discrete case $n=5$ results in the same condition as below.\\
For $n\geq 6$,
%%\vspace-.3cm} 
\begin{align}
\nonumber
\begin{split}
&
(m_1+1)(b_1-c)+(m_2-2)b_2+(1-\gamma)b_2 \\ 
&\leq m_1(b_1-c)+(m_2-1)b_2+(1-\gamma)b_2
%%\vspace-.1cm} 
\end{split}
\end{align}
\begin{equation}
\nonumber
\iff c \geq b_1-b_2
%%\vspace-.1cm} 
\end{equation} 
For $k\geq 2$, $n=5$ corresponds to pairwise stability condition. \\
For $n\geq 6$,
%%\vspace-.3cm} 
%\begin{small}
\begin{equation}
\nonumber
(m_1+1)(b_1-c)+(m_2-2)b_2+b_2 \leq m_1(b_1-c)+(m_2-1)b_2+b_2
%\vspace-5mm}
\end{equation}
%\end{small}
\begin{equation}
\nonumber
\iff c \geq b_1-b_2
%%\vspace-.1cm} 
\end{equation} 
These are weaker conditions than Inequality~(\ref{E7b}).\\

\noindent
\textbf{No two Type II nodes should create a link between themselves:} This only applies to $k\geq 2$.
Also $n=4,5$ result in pairwise stability condition.\\
For $n\geq 6$,
%%\vspace-.3cm} 
\begin{equation}
\nonumber
(m_2+2)(b_1-c)+(m_1-2)b_2 \leq (m_2+1)(b_1-c)+(m_1-1)b_2
%\vspace-5mm}
\end{equation}
\begin{equation}
\nonumber
\iff c\geq b_1-b_2
%%\vspace-.1cm} 
\end{equation} 
which is a weaker condition than Inequality~(\ref{E7b}).\\

\noindent
\textbf{Link between Type I node and Type II node  should not be deleted:} It is clear that it will not be deleted as such a link is just formed with no other changes in the network.\\

%%\vspace-.3cm} 
Inequalities~(\ref{E6}) and (\ref{E7b}) are stronger conditions than Inequalities~(\ref{E16}), (\ref{E17for2}) and (\ref{E17}).
Furthermore, for non-zero range of $c$, from Inequalities~(\ref{E6}) and (\ref{E7b}), we have
%%\vspace-.1cm} 
\begin{equation}
\label{E0}
\gamma < \frac{b_2-b_3}{3b_2-b_3}
\end{equation} 
The required sufficient conditions are obtained by combining Inequalities~(\ref{E1b}), (\ref{E2a}), (\ref{E6}), (\ref{E7b}) and (\ref{E0}).
%\qed 
\end{proof}

\section{Proof of Proposition~\ref{thm:2star}
}
\label{app:2star}
%\noindent
\begin{customprop}{\ref{thm:2star}}
%\label{thm:2star}
Let $\sigma$ be the upper bound on the number of nodes that can enter the network and $\lambda = \lceil \frac{\sigma}{2} -1 \rceil \left( 2b_2-b_3 \right)$.
Then, if $\left( 1-\gamma \right) \left( b_2-b_3 \right) < c_0 < \left( 1-\gamma \right) \left( b_2-b_4 \right) $ and either \\
(i) $\gamma < \min \Big\{  \frac{b_2-b_3}{\lambda-b_3} ,  \frac{b_3}{b_2+b_3} \Big\}$ and $b_1-b_3+\gamma(b_2+b_3) \leq c < b_1$, or\\
(ii) $\frac{b_2-b_3}{\lambda-b_3} \leq \gamma < \min \Big\{ \frac{b_2}{\lambda+b_2} , \frac{b_3}{b_2+b_3}  \Big\}$ and $b_1-b_2+\gamma b_2 + \gamma \lambda \leq c < b_1$,\\
the unique resulting topology is a 
2-star.
\end{customprop}
\begin{proof}
We derive sufficient conditions for the formation of a 2-star network by forming its skeleton of four nodes first, that is, a network with two interconnected centers, connected to one leaf node each. Once this is formed, we ensure that a newly entering node connects to the center with fewer number of leaf nodes, whenever applicable, so as to maintain the load balance between the two centers.\\

\noindent
\textbf{Forming the skeleton of the 2-star network:}
With one node in the network, the second node should successfully create a link with the former. The condition for ensuring this is
\begin{equation}
\label{F0.1}
c<b_1
\end{equation}
For the third node to enter, it should propose a link to any of the two existing nodes in the network, that is, it should get a positive utility by doing so. This gives
\begin{equation}
\nonumber
c<b_1+(1-\gamma)b_2-c_0
\end{equation}
This is ensured by Inequality~(\ref{F0.1}) and 
\begin{equation}
\label{F0.2}
c_0\leq (1-\gamma)b_2
\end{equation}
Also the existing node to which the link is proposed, should accept it, that is, its utility should not decrease by doing so.
\begin{equation}
\nonumber
2(b_1-c) + 2\gamma b_2 \geq b_1-c
%\vspace-5mm}
\end{equation}
\begin{equation}
\nonumber
\iff c \leq b_1+2\gamma b_2
\end{equation}
which is a weaker condition than Inequality~(\ref{F0.1}).
We have to also ensure that this V-shaped network of three nodes is pairwise stable. It is clear that no node will delete any of its links since such a link is just formed. However, we have to ensure that the two leaf nodes of this V-shaped network do not create a mutual link. This can be ensured by
\begin{equation}
\nonumber
2(b_1-c) \leq b_1-c + (1-\gamma)b_2
%\vspace-5mm}
\end{equation}
\begin{equation}
\label{F5for3}
\iff c \geq b_1-b_2+\gamma b_2
\end{equation}
Following this, the fourth node should propose a link to one of the two leaf node in the V-shaped network. For ensuring that its utility increases by doing so,
\begin{equation}
\nonumber
c<b_1+(1-\gamma)b_2-c_0 +(1-\gamma)b_3
\end{equation}
which is satisfied by Inequalities~(\ref{F0.1}) and (\ref{F0.2}).
Also, it should prefer connecting to a leaf node than the center of the V-shaped network, that is,
\begin{equation}
\nonumber
b_1-c +(1-\gamma)b_2 - c_0 +(1-\gamma)b_3 > b_1-c +2(1-\gamma)b_2 -2c_0
%\vspace-5mm}
\end{equation}
\begin{equation}
\label{F0.3}
\iff c_0 > (1-\gamma)(b_2-b_3)
\end{equation}
The leaf node to which the link is proposed, should accept the link.
\begin{equation}
\nonumber
2(b_1-c)+(1-\gamma)b_2+2\gamma b_2 +\gamma b_3 \geq b_1-c +(1-\gamma)b_2
%\vspace-5mm}
\end{equation}
\begin{equation}
\nonumber
\iff c \leq b_1 + \gamma(2b_2+b_3)
\end{equation}
which is satisfied by Inequality~(\ref{F0.1}).\\
We have to also ensure that this network is pairwise stable. We derive sufficient conditions for pairwise stability of a general 2-star network with number of nodes $n \geq 4$, which includes the sufficient conditions for pairwise stability of the skeleton thus formed.\\

Let the centers of the 2-star be labeled $C_1$ and $C_2$. Also, let the number of leaf nodes connected to $C_1$ be $m_1$ and that connected to $C_2$ be $m_2$. \\

\noindent
\textbf{Leaf nodes that are connected to different centers, should not create a mutual link:} 
This scenario is valid for $n\geq 4$. Without loss of generality, for a leaf node connected to $C_1$,
\begin{align}
\nonumber
\begin{split}
&
2(b_1-c)+(m_1-1)(1-\gamma)b_2 +b_2 +(m_2-1)(1-\gamma)b_3\\
&\leq b_1-c +m_1(1-\gamma)b_2 +m_2(1-\gamma)b_3
\end{split}
\end{align}
\begin{equation}
\label{F5}
\iff c \geq b_1-b_3+\gamma(b_2+b_3)
\end{equation}

\noindent
\textbf{Link between one center and a leaf node of the other center should not be created:}
One option to ensure this is to see that the utility of center $C_1$ decreases owing to its link with a leaf node of $C_2$. This is valid for $n\geq 4$.
%\begin{small}
\begin{align}
\nonumber
\begin{split}
&
(m_1+2)(b_1-c) +(m_2-1)(1-\gamma)b_2 + \gamma(2)(m_1)2b_2 
+\frac{\gamma}{2}(m_1)(m_2-1)2b_3\\
&< (m_1+1)(b_1-c) +m_2(1-\gamma)b_2 + \gamma(1)(m_1)2b_2 
+ \frac{\gamma}{2}(m_1)(m_2)2b_3
\end{split}
\end{align}
%\end{small}
\begin{equation}
\nonumber
\iff c > b_1-b_2+\gamma b_2 +\gamma m_1 (2b_2-b_3)
\end{equation}
As it needs to be true for all $n \geq 4$, we set the condition to
\begin{equation}
\nonumber
 c > \max_{n \geq 4} \Big\{ b_1-b_2+\gamma b_2 +\gamma m_1 (2b_2-b_3) \Big\}
\end{equation}
Since $\max\{m_1\} = \lceil \frac{\sigma}{2}-1 \rceil$, where $\sigma$ is the upper bound on the number of nodes that can enter the network,
\begin{equation}
\label{F6a}
c > b_1-b_2+\gamma b_2 +\gamma \lceil \frac{\sigma}{2}-1 \rceil (2b_2-b_3)
\end{equation}
An alternative option to the above condition is to ensure that the utility of leaf node connected to $C_2$ decreases when it establishes a link with $C_1$.
\begin{align}
\nonumber
\begin{split}
&
2(b_1-c) +(m_2-1)(1-\gamma)b_2 +m_1(1-\gamma)b_2 \\
&< b_1-c +m_2(1-\gamma)b_2  + m_1(1-\gamma)b_3
\end{split}
\end{align}
\begin{equation}
\nonumber
\iff c > b_1-(1-\gamma)b_2+m_1(1-\gamma)(b_2-b_3)
\end{equation}
As it needs to be true for all $n \geq 4$ and $\max\{m_1\} = \lceil \frac{\sigma}{2}-1 \rceil$, we set the condition to
\begin{equation}
\label{F6b}
 c > b_1-(1-\gamma)b_2+(1-\gamma) \lceil \frac{\sigma}{2}-1 \rceil(b_2-b_3)
\end{equation}

\noindent
\textbf{No link is broken in the 2-star network:}
It is easy to check that, as 2-star is a tree graph, the condition $c<b_1$ in Inequality~(\ref{F0.1}) is sufficient to ensure this.\\

\noindent
\textbf{Two leaf nodes of a center should not create a mutual link:}
This case arises for $n\geq 5$. It can be easily checked that the condition $c \geq b_1-b_2+\gamma b_2$ in Inequality~(\ref{F5for3}) is sufficient to ensure this.\\

This completes the sufficient conditions for pairwise stability of a 2-star network.
In what follows, we ensure that any new node successfully enters an existing 2-star network such that the topology is maintained.\\

\noindent
\textbf{A newly entering node should prefer connecting to the center with less number of leaf nodes, whenever applicable:}
This case arises when $n$ is even and $n\geq 6$, that is, when a new node tries to enter a 2-star network with odd number of nodes. Without loss of generality, let $m_1=m_2+1$. So the new node should prefer connecting to $C_2$ over $C_1$.
\begin{align}
\nonumber
\begin{split}
&
b_1-c +(m_2+1)((1-\gamma)b_2 - c_0) +m_1(1-\gamma)b_3\\
&> b_1-c +(m_1+1)((1-\gamma)b_2 -c_0) +m_2(1-\gamma)b_3
\end{split}
\end{align}
\begin{equation}
\nonumber
\iff c_0 > (1-\gamma)(b_2-b_3)
\end{equation}
which is same as Inequality~(\ref{F0.3}).\\

\noindent
\textbf{The new node should not stay out of the network:}
Its utility should be positive when it enters the network by connecting to the center with less number of leaf nodes, whenever applicable.
\begin{equation}
\nonumber
b_1-c +(m_2+1)((1-\gamma)b_2 - c_0) +m_1(1-\gamma)b_3 > 0
\end{equation}
It can be easily seen that, as $m_1,m_2 \geq 1$, the above is always true when Inequalities~(\ref{F0.1}) and (\ref{F0.3}) are satisfied.\\

\noindent
\textbf{The center with less number of leaf nodes, whenever applicable, should accept the link from the newly entering node:}
The condition $c<b_1$ in Inequality~(\ref{F0.1}) is sufficient to ensure this.\\

\noindent
\textbf{The newly entering node should prefer connecting to the center with less number of leaf nodes, whenever applicable, over connecting to any leaf node:}
It is easy to see that, as $b_3>b_4$, whenever the number of leaf nodes connected to the centers are different, a newly entering node prefers connecting to a leaf node connected to $C_1$ over that connected to $C_2$ (assuming $m_1=m_2+1$). 
Hence we have to ensure that connecting to the center with less number of leaf nodes, whenever applicable, is more beneficial to a newly entering node than connecting to a leaf node that is connected to $C_1$. Without loss of generality, we want the new node to prefer connecting to $C_2$ (irrespective of whether $m_1=m_2$ or $m_1=m_2+1$).
\begin{align}
\nonumber
\begin{split}
&
b_1-c +(m_2+1)((1-\gamma)b_2 - c_0) +m_1(1-\gamma)b_3 \\
&> b_1-c +(1-\gamma)b_2-c_0 +m_1(1-\gamma)b_3 +m_2(1-\gamma)b_4
\end{split}
\end{align}
\begin{equation}
\nonumber
\iff m_2(1-\gamma)b_2 -m_2 c_0 > m_2(1-\gamma)b_4
\end{equation}
As $m_2 \geq 1$,
\begin{equation}
\label{F4}
c_0 < (1-\gamma)(b_2-b_4)
\end{equation}

The conditions on $c$ can be obtained from Inequalities~(\ref{F0.1}), (\ref{F5for3}), (\ref{F5}), and either (\ref{F6a}) or (\ref{F6b}). 
Suppose we choose Inequality~(\ref{F6b}) over Inequality~(\ref{F6a}). So, for $c$ to have a non-empty range of values, from Inequalities~(\ref{F0.1}) and (\ref{F6b}), we must have
\begin{equation}
\nonumber
b_1-(1-\gamma)b_2+(1-\gamma) \lceil \frac{\sigma}{2}-1 \rceil(b_2-b_3) < b_1
\end{equation}
As $\gamma<1$, the above is equivalent to
\begin{equation}
\nonumber
b_2>\lceil \frac{\sigma}{2}-1 \rceil(b_2-b_3) 
\end{equation}
which is not true for arbitrarily large values of $\sigma$. So we cannot use Inequality~(\ref{F6b}).
Suppose we choose Inequality~(\ref{F6a}). So, for $c$ to have a non-empty range of values, from Inequalities~(\ref{F0.1}) and (\ref{F6a}), we must have
\begin{equation}
\nonumber
b_1-b_2+\gamma b_2 +\gamma \lceil \frac{\sigma}{2}-1 \rceil (2b_2-b_3) < b_1
\end{equation}
Let $\lambda = \lceil \frac{\sigma}{2}-1 \rceil (2b_2-b_3)$. So the above is equivalent to
\begin{equation}
\label{2stargamma1}
\gamma < \frac{b_2}{\lambda+b_2}
\end{equation}
which is a valid range of $\gamma$ as $\gamma \in {[0,1)}$.
So we use Inequality~(\ref{F6a}) instead of Inequality~(\ref{F6b}). Also, Inequality~(\ref{F5for3}) is weaker than Inequality~(\ref{F6a}).
For $c$ to have a non-empty range of values, it is also necessary, from Inequalities~(\ref{F0.1}) and (\ref{F5}), that
\begin{equation}
\nonumber
b_1-b_3+\gamma(b_2+b_3) < b_1
%\vspace-5mm}
\end{equation}
\begin{equation}
\label{2stargamma2}
\iff \gamma < \frac{b_3}{b_2+b_3}
\end{equation}
Both Inequalities~(\ref{F5}) and (\ref{F6a}) lower bound $c$. So we need to determine the stronger condition of the two. It can be seen that Inequality~(\ref{F5}) is at least as strong as Inequality~(\ref{F6a}) if and only if
\begin{equation}
\nonumber
b_1-b_3+\gamma(b_2+b_3) \geq b_1-b_2+\gamma b_2 + \gamma \lambda
\end{equation}
\begin{equation}
\label{2stargamma3}
\iff \gamma \leq \frac{b_2-b_3}{\lambda-b_3}
\end{equation}
We consider the cases when either is a stronger condition.

\noindent
\textbf{Case (i)} If Inequality~(\ref{2stargamma3}) is true:\\
Inequalities~(\ref{F0.1}) and (\ref{F5}) are the strongest conditions. So the sufficient condition on $c$ is
\begin{equation}
\label{2starfinalc1}
b_1-b_3+\gamma(b_2+b_3) \leq c < b_1
\end{equation}
and Inequalities~(\ref{2stargamma1}), (\ref{2stargamma2}) and (\ref{2stargamma3}) give
\begin{equation}
\nonumber
\gamma < \min \Big\{ \frac{b_2-b_3}{\lambda-b_3} , \frac{b_2}{\lambda+b_2} ,   \frac{b_3}{b_2+b_3}  \Big\}
\end{equation}
It can also be shown that for $\lambda b_3 \geq b_2^2$, 
$\min \Big\{ \frac{b_2-b_3}{\lambda-b_3} , \frac{b_2}{\lambda+b_2} ,   \frac{b_3}{b_2+b_3}  \Big\} = \frac{b_2-b_3}{\lambda-b_3}$ \\
 and for $\lambda b_3 \leq b_2^2$, 
$\min \Big\{ \frac{b_2-b_3}{\lambda-b_3} , \frac{b_2}{\lambda+b_2} ,   \frac{b_3}{b_2+b_3}  \Big\} = \frac{b_3}{b_2+b_3}$. So the above reduces to
\begin{equation}
\label{2starfinalgamma1}
\gamma < \min \Big\{ \frac{b_2-b_3}{\lambda-b_3} , \frac{b_3}{b_2+b_3}  \Big\}
\end{equation}

\noindent
\textbf{Case (ii)} If Inequality~(\ref{2stargamma3}) is not true:\\
Inequalities~(\ref{F0.1}) and (\ref{F6a}) are the strongest conditions. So the sufficient condition on $c$ is
\begin{equation}
\label{2starfinalc2}
b_1-b_2+\gamma b_2 + \gamma \lambda \leq c < b_1
\end{equation}
and Inequalities~(\ref{2stargamma1}), (\ref{2stargamma2}) and the reverse of (\ref{2stargamma3}) give
\begin{equation}
\label{2starfinalgamma2}
\frac{b_2-b_3}{\lambda-b_3} \leq \gamma < \min \Big\{ \frac{b_2}{\lambda+b_2}  ,  \frac{b_3}{b_2+b_3}  \Big\}
\end{equation}
Furthermore, Inequalities~(\ref{F0.2}), (\ref{F0.3}) and (\ref{F4}) give the sufficient conditions on $c_0$.
\begin{equation}
\label{2starfinalc0}
(1-\gamma)(b_2-b_3) < c_0 < (1-\gamma)(b_2-b_4)
\end{equation}
Inequalities~(\ref{2starfinalc1}), (\ref{2starfinalgamma1}) and (\ref{2starfinalc0}) give the sufficient conditions $(i)$ in the proposition, while Inequalities~(\ref{2starfinalc2}), (\ref{2starfinalgamma2}) and (\ref{2starfinalc0}) give the sufficient conditions $(ii)$.
%\qed 
\end{proof}

\section{Proof of Lemma~\ref{lem:kstar0}
}
\label{app:kstar0}
%\noindent
\begin{customlem}{\ref{lem:kstar0}}
%\label{lem:kstar0}
Under the proposed utility model, for the entire family of $k$-star networks (given some $k\geq 3$) to be pairwise stable, it is necessary that 
$\gamma=0$ and $c=b_1-b_3$.
\end{customlem}
\begin{proof}
We consider two scenarios sufficient to prove this.\\

%G6 %
\noindent
\textbf{I) No center should delete its link with any other center:} Here, only one case is enough to be considered, that is, when each center has just one leaf node,
since in all other cases, the benefits obtained by each center from the connection with other centers is at least as much. For $k=3$, 
%\begin{small}
\begin{align}
\nonumber
\begin{split}
&
3(b_1-c) + 2(1-\gamma)b_2 + \gamma(1)(2)2b_2 + \frac{\gamma}{2}(1)(2)b_3 \\
&\geq
2(b_1-c) + 2(1-\gamma)b_2 + (1-\gamma)b_3 + \gamma(1)(1)2b_2 
 + 2\left( \frac{\gamma}{2}(1)(1)2b_3 \right) +  \frac{\gamma}{3}(1)(1)2b_4
\end{split}
\end{align}
%\end{small}
\begin{equation}
\label{eq:kstarineq1}
\iff c \leq b_1-b_3+\gamma(2b_2+b_3)-\frac{2\gamma}{3}b_4
\end{equation}
For $k\geq 4$,
%\begin{small}
\begin{align}
\nonumber
\begin{split}
&
(k-1+1)(b_1-c) + (k-1)(1-\gamma)b_2 + \gamma(1)(k-1)2b_2  
+ \frac{\gamma}{2}(1)(k-1)2b_3\\
 &\geq (k-2+1)(b_1-c) + (k-2)(1-\gamma)b_2 + b_2 + (1-\gamma)b_3  \\
 &\;\;\;\;\; +  \gamma(1)(k-2)2b_2+ \frac{\gamma}{2}(1)(k-2)2b_3 + \gamma(1)(1)2b_3 + \frac{\gamma}{2}(1)(1)2b_4
\end{split}
\end{align}
%\end{small}
\begin{equation}
\label{eq:kstarineq2}
\iff c \leq b_1-b_3+\gamma(b_2-b_4)
\end{equation}

%G9 %
\noindent
\textbf{II) Leaf nodes of different centers should not form a link with each other:} Consider a leaf node. Let $m_i$ be the number of leaf nodes connected to the center to which the leaf node under consideration, is connected. For $k\geq 3$,
%\begin{small}
\begin{align}
\nonumber
\begin{split}
&
2(b_1-c) + (m_i-1)(1-\gamma)b_2 + (1-\gamma)b_3 (\sum_{j \neq i}m_i - 1) + (k-1)b_2 \\
&\leq b_1-c + (m_i-1)(1-\gamma)b_2 + (1-\gamma)b_3 \sum_{j \neq i}m_i + (k-1)(1-\gamma)b_2
\end{split}
\end{align}
%\end{small}
\begin{equation}
\label{eq:kstarineq3}
\iff c \geq b_1-b_3+\gamma((k-1)b_2+b_3)
\end{equation}
The only way to satisfy Inequalities~(\ref{eq:kstarineq1}), (\ref{eq:kstarineq2}) and (\ref{eq:kstarineq3}) simultaneously is by setting 
\begin{equation}
\label{eq:kstargamma}
\gamma=0
\end{equation}
and
\begin{equation}
\label{eq:kstarcost}
c=b_1-b_3
\end{equation}
thus proving the lemma.
%\qed 
\end{proof}

\section{Proof of Proposition~\ref{thm:kstar}
}
\label{app:kstar}
%\noindent
\begin{customprop}{\ref{thm:kstar}}
%\label{thm:kstar}
For a network starting with the base graph for $k$-star (given some $k \geq 3$), and $\gamma =0 $, if $c =b_1-b_3 $ 
and $ b_2-b_3  < c_0 < b_2-b_4$, the unique resulting topology is a 
$k$-star.
\end{customprop}
\begin{proof}
It is clear from Lemma~\ref{lem:kstar0} that under the proposed utility model, for the family of $k$-star networks ($k\geq 3$) to be pairwise stable, it is necessary that $\gamma=0$ and $c=b_1-b_3$ in order to stabilize all possible $k$-star networks for a given $k$, and hence forms the necessary part of sufficient conditions for the formation of a $k$-star network. Hence, for the rest of this proof, we will assume that
\begin{equation}
\label{eq:appkstargamma}
\gamma=0
\end{equation}
and
\begin{equation}
\label{eq:appkstarcost}
c=b_1-b_3
\end{equation}

 Without loss of generality, assume some indexing over the $k$ centers from $1$ to $k$.
Let $C_i$ be the center with index $i$ and $m_i$ be the number of leaf nodes it is linked to. 
Also we start with a base graph in which every center is linked to one leaf node and the number of leaf nodes linked to each center increases as the process goes on. 
So we have, $m_i \geq 1\text{ for } 1 \leq i \leq k$.\\

%G1%
\noindent
\textbf{For the newly entering node to propose entering the network:}
Our objective is to ensure that the newly entering node connects to a center with the least number of leaf nodes, in order to maintain balance over the number of leaf nodes linked to the centers. Without loss of generality, assume that we want the newly entering node to connect to $C_1$. The utility of the newly entering node should be positive after doing so.
\begin{equation}
\nonumber
b_1 - c + (m_1+k-1)\left(b_2-c_0\right) +b_3\sum_{i=2}^{k} m_i > 0
\end{equation}
Since the minimum value of $m_i$ is $1$ for any $i$, the above condition is true if
\begin{equation}
\nonumber
c < b_1 + k \left(b_2-c_0\right) + (k-1)b_3
\end{equation}
This is satisfied by Equation~(\ref{eq:appkstarcost}) and
\begin{equation}
\label{G1forc0}
c_0 < b_2+b_3
\end{equation}

%G3 %
\noindent
\textbf{The newly entering node should connect to a center with the least number of leaf nodes, whenever applicable:}
This case does not arise when all centers have the same number of leaf nodes. Moreover, the way we direct the evolution of the network, the number of leaf nodes connected to any two centers differs by at most one. Without loss of generality, assume that we want the newly entering node to connect to $C_1$. Consider a center $C_p$ such that $m_p = m_1+1$. So the newly entering node should prefer connecting to $C_1$ over connecting to $C_p$.
%\vspace-5mm}
\begin{align}
\nonumber
\begin{split}
&
b_1-c + (m_1+k-1) \left(b_2 - c_0 \right) +b_3\sum_{i=2}^{k} m_i \\
&> b_1-c + (m_p+k-1) \left( b_2 - c_0 \right) + b_3\sum_{\substack{1 \leq i \leq k \\ i \neq p}} m_i
\end{split}
\end{align}
As $m_p = m_1 + 1$, we have
\begin{equation}
\label{G3}
c_0 > b_2-b_3
\end{equation}

%G2%
\noindent
\textbf{For a center with the least number of leaf nodes to accept the link from the newly entering node:}
It can be easily seen that this is ensured by Equation~(\ref{eq:appkstarcost}).\\

%G4 %
\noindent
\textbf{The newly entering node should not connect to any leaf node:}
It can be easily seen that owing to benefits degrading with distance, for the newly entering node, connecting to any leaf node which is connected to a center with the most number of leaf nodes, strictly dominates connecting to any other leaf node, whenever applicable. So it is sufficient to ensure that the newly entering node does not connect to any leaf node which is connected to a center with the most number of leaf nodes.
This can be done by ensuring that for the newly entering node, connecting to a center with the least number of leaf nodes strictly dominates connecting to any leaf node which is connected to a center with the most number of leaf nodes.
Say we want the newly entering node to prefer connecting to center $C_1$ over a leaf node that is linked to center $C_p$.
%\vspace-5mm}
\begin{align}
\nonumber
\begin{split}
&
b_1-c + (m_1+k-1)(b_2-c_0) +  b_3\sum_{i=2}^k m_i \\
&> b_1-c +b_2 - c_0 + (m_p + k-2)b_3 +  b_4\sum_{\substack{1 \leq i \leq k \\ i \neq p}} m_i 
\end{split}
\end{align}
We need to consider two cases (i) $m_p = m_1+1$ and (ii) $m_p=m_1$\\
Case (i) $m_p = m_1+1$: Substituting  this value of $m_p$ gives
%\begin{small}
\begin{equation}
\nonumber
\begin{split}
(m_1+k-1)(b_2-c_0) +  (b_3-b_4)\sum_{\substack{2 \leq i \leq k \\ i \neq p}} m_i + m_1 (b_3-b_4) \\+ b_3
> b_2 - c_0 + (m_1 + k-1)b_3 
\end{split}
\end{equation}
%\end{small}
As the minimum value of $\sum_{2 \leq i \leq k, i \neq p}m_i$ is $k-2$, the above remains true if we replace $\sum_{2 \leq i \leq k, i \neq p}m_i$ with $k-2$. Further simplification gives
\begin{equation}
\nonumber
(m_1+k-2)(b_2-b_4-c_0)>0
\end{equation}
Since $m_1+k-2>0$ is positive (as $m_1 \geq 1$ and $k \geq 3$), we must have
\begin{equation}
\label{G4forc0}
c_0 < b_2-b_4
\end{equation}
Case (ii) $m_k=m_1$: It can be similarly shown that Equation~(\ref{G4forc0}) is the sufficient condition.\\

%Pairwise Stability
\noindent
Now that the newly entering node enters in a way such that $k$-star network is formed, we have to ensure that no further modifications of links occur so that the network thus formed, is pairwise stable.\\

%G5 %
\noindent
\textbf{For centers and the corresponding leaf nodes to not delete the link between them:} It can be easily seen that $c<b_1$, a weaker condition than Equation~(\ref{eq:appkstarcost}), is a sufficient condition to ensure this.\\

%G6 %
\noindent
\textbf{No center should delete its link with any other center:} This is ensured by the inequalities in the proof of Lemma~\ref{lem:kstar0}, which are weaker than Equations~(\ref{eq:appkstargamma}) and (\ref{eq:appkstarcost}). \\

%G8 %
\noindent
\textbf{Leaf nodes of a center should not form a link with each other:} The net benefit that a leaf node would get by forming such a link should be non-positive.
\begin{equation}
\nonumber
b_1-c - b_2 \leq 0
%\vspace-5mm}
\end{equation}
\begin{equation}
\nonumber
\iff c \geq b_1-b_2
\end{equation}
which is satisfied by Equation~(\ref{eq:appkstarcost}).\\

%G9 %
\noindent
\textbf{Leaf nodes of different centers should not form a link with each other:} This is ensured by the inequality in the proof of Lemma~\ref{lem:kstar0}, which is weaker than Equations~(\ref{eq:appkstargamma}) and (\ref{eq:appkstarcost}). \\

%G7 %
\noindent
\textbf{Link between a center and a leaf node of any other center should not be created:} Let $C_i$ be the center under consideration and the leaf node under consideration be linked to $C_j$ ($j\neq i$). There are two ways to ensure this. First is to ensure that a center neither proposes nor accepts a link with a leaf node of any other center. This mathematically is
\begin{align}
\nonumber
\begin{split}
&
(k-1+m_i+1)(b_1-c)+  b_2(\sum_{\substack{1 \leq q \leq k \\ q \neq i}}m_q-1) \\ &< (k-1+m_i)(b_1-c) + b_2 \sum_{\substack{1 \leq q \leq k \\ q \neq i}}m_q 
\end{split}
\end{align}
\begin{equation}
\nonumber
\iff c > b_1-b_2
\end{equation}

An alternative to this condition is to ensure that a leaf node neither proposes nor accepts a link with a center to which it is not connected, but since this condition is already satisfied by Equation~(\ref{eq:appkstarcost}), this alternative need not be considered.

Equations~(\ref{eq:appkstargamma}), (\ref{eq:appkstarcost}), (\ref{G1forc0}), (\ref{G3}) and (\ref{G4forc0}) give the required sufficient conditions for the $k$-star network topology.
%
%\qed 
\end{proof}

\section{Proof of Theorem~\ref{thm:gedkstar}
}
\label{app:gedkstar}
%\noindent
\begin{customthm}{\ref{thm:gedkstar}}
%\label{thm:gedkstar}
There exists an $O(\mu^{k+2})$ polynomial time algorithm to compute the graph edit distance %the graph edit distance
 between a graph $g$ and a $k$-star graph with same number of nodes as $g$, where $\mu$ is the number of nodes in $g$.
\end{customthm}
\begin{proof}
Assume that the mapping of the $k$ centers of the $k$-star network to the nodes in $g$, is known. Let us call these nodes of $g$ as {\em pseudo-centers}. The graph edit distance can be computed by taking the minimum number of edge edit operations over all possible mappings. In a $k$-star graph, each node, other than centers, is allotted to exactly one center. Hence, our objective is to allot nodes, other than pseudo-centers, (call them {\em pseudo-leaves}) in $g$ to pseudo-centers such that the graph edit distance is minimized.
Let $\mu$ and $\xi$ be the number of nodes and edges in $g$, respectively. 
Let {\em vacancy} of a pseudo-center at any point of time be defined as the maximum number of nodes that can be allotted to it, given the current allotment.
Note that if $\mu$ is not a multiple of $k$, the vacancy of a pseudo-center depends not only on the number of pseudo-leaves allotted to it, but also on the number of pseudo-leaves allotted to other pseudo-centers.

It is clear that if the mapping of the $k$ centers is known, for transforming $g$ to a corresponding $k$-star, it is necessary that all missing links between any two pseudo-centers be added (let $\beta_1$ be the number of such links) and all existing links between any two pseudo-leaves be deleted (let $\beta_2$ be the number of such links).
The only other links that need to be computed for additions or deletions, in order to minimize graph edit distance, are those interlinking pseudo-leaves with pseudo-centers.
%RELOOK
The number of links that already interlink pseudo-leaves with pseudo-centers in $g$ is $\beta_3 = (\xi-\beta_2-($\begin{footnotesize}$\dbinom{k}{2}$\end{footnotesize}$-\beta_1))$.
Say the number of these edges that are retained during the transformation to $k$-star, is $f$, that is, exactly $f$ pseudo-leaves are allotted a pseudo-center and $(\mu-k-f)$ are not. So the number of edges interlinking pseudo-leaves with pseudo-centers, that are deleted during the transformation, is $(\beta_3-f)$. Also, the number of edges to be added in order to allot the pseudo-leaves, that are not allotted to any pseudo-center, to some pseudo-center having a positive vacancy, is $(\mu-k-f)$.
So the number of edge edit operations is $(\beta_1+\beta_2+\beta_3+\mu-k-2f) = (\mu+\xi+2\beta_1-$\begin{normalsize}$\frac{k}{2}$\end{normalsize}$(k+1)-2f)$.
Given a mapping of the $k$ centers, the only variable in this expression is $f$. 
So in order to minimize its value, we need to maximize the number of edges interlinking pseudo-leaves and pseudo-centers, that remain intact after the transformation to $k$-star.
We now address this problem of maximizing $f$.

Let the number of nodes in $g$ be $\mu=pk+q$ where $p$ and $q$ are integers such that $p \geq 0$ and $1\leq q< k$. 
In a $k$-star graph with $\mu$ nodes, $q$ centers are linked to $p$ leaf nodes and the remaining $k-q$ are linked to $p-1$ leaf nodes. So for transforming $g$ to a corresponding $k$-star graph, $q$ pseudo-centers should be allotted $p$ pseudo-leaves and the remaining $k-q$ should be allotted $p-1$.
So, at most $q$ pseudo-centers should be allotted $p$ nodes, that is, the vacancy of at most $q$ pseudo-centers should be $p$, while that of the remaining $k-q$ should be $p-1$.
In other words, to start with, the sum of vacancies of any $q+1$ pseudo-centers should be at most $(q+1)p-1$. 

\begin{figure}[t!]
\centering
%\hspace{-.2cm}
\includegraphics[scale=0.57]{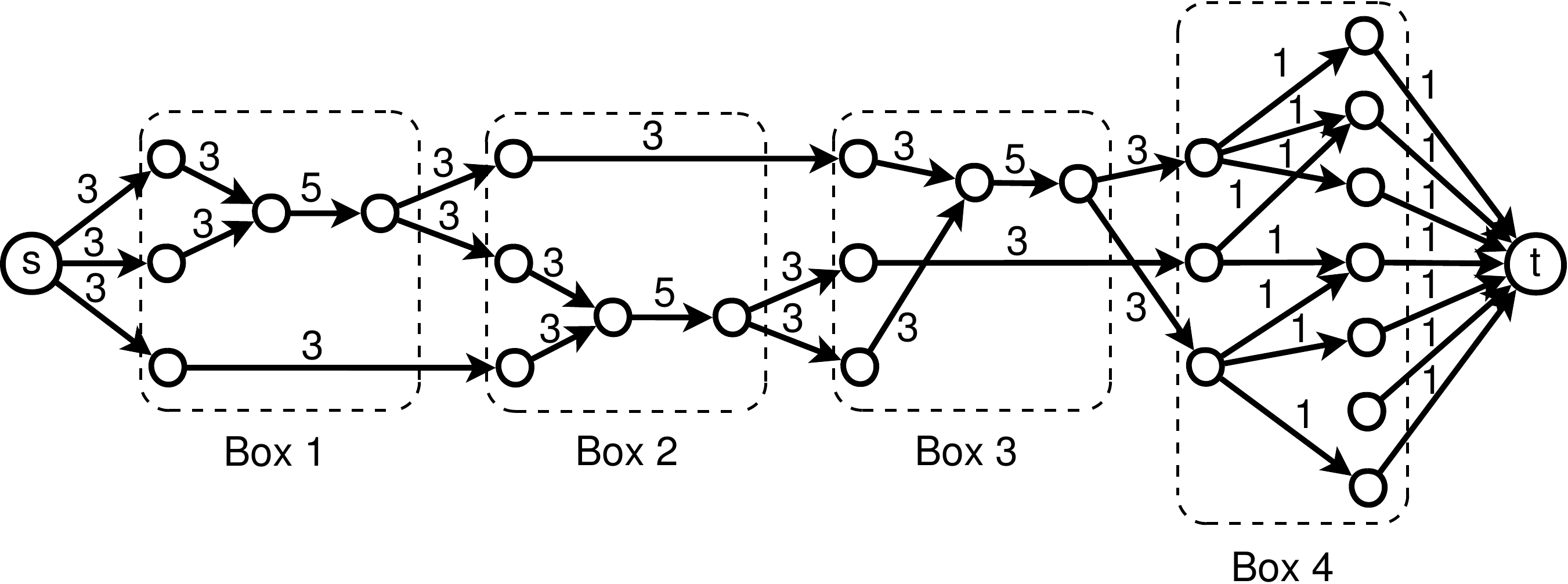}
\caption{Formulation of graph edit distance between graph $g$ ($\mu=10$) and a $3$-star graph with same number of nodes as $g$, as a max-flow problem}
\label{fig:max_flow_eg}
\end{figure}

The above problem can be formulated as an application of max-flow in a directed network.
Figure~\ref{fig:max_flow_eg} shows the formulation for a graph $g$ with 10 nodes and a 3-star graph, where $p=3$ and $q=1$. The edges directing from the source node $s$ to the left $k$ nodes in Box 1 (here $k=3$) and those in Boxes 1, 2 and 3, formulate the vacancy of each of these pseudo-centers to be $p$. Boxes 1, 2 and 3 formulate the constraint that the sum of vacancies of any $q+1$ pseudo-centers should be at most $(q+1)p-1$.
The rightmost Box 4 is obtained by considering edges only interlinking any pseudo-centers (left nodes) and pseudo-leaves (right nodes). 

As all the edges have integer capacities, the Ford-Fulkerson algorithm constructs an integer maximum flow.
%~\cite{cormen2001introduction}.
The number of constraints concerning the sum of vacancies of pseudo-centers is \begin{footnotesize}$\dbinom{k}{q+1}$\end{footnotesize} and number of edges added per such constraint is $2(q+1)+2$.

So the maximum number of edges, say $\chi$, in the max-flow formulation, is $k$ (from source node to left $k$ nodes in Box 1) $+$ $\left( 2(q+1)+2 \right)$\begin{footnotesize}$\dbinom{k}{q+1}$\end{footnotesize} (from the above calculation) $+$ $k(\mu-k)$ (upper limit on the number of edges in Box 4, interlinking pseudo-centers and pseudo-leaves) $+$ $(\mu-k)$ (number of edges directing towards target node). 
Since $1 \leq q < k$, we have $\chi=O(k^{\frac{k}{2}+1}+\mu k)$.
As the value of the maximum flow is upper bounded by $\mu-k$, the Ford-Fulkerson algorithm runs in $O(\chi \mu) = O(\mu k^{\frac{k}{2}+1}+\mu^2 k)$ time. 
Furthermore, as $k$ is a constant, the asymptotic worst-case time complexity is $O(\mu^2)$.

So given a mapping of the $k$ centers, the number of edge edit operations, $(\mu+\xi+2\beta_1-$\begin{normalsize}$\frac{k}{2}$\end{normalsize}$(k+1)-2f)$, is minimized since $f$ is maximized.
The time complexity of the above algorithm is dominated by the max-flow algorithm. The above analysis was assuming that the mapping of the $k$ centers of the $k$-star network to the nodes in $G$, is known. The graph edit distance can, hence, be computed by taking the minimum edit distance over all \begin{footnotesize}$\dbinom{\mu}{k}$\end{footnotesize} $= O(\mu^k)$ possible mappings. So the asymptotic worst-case time complexity of the algorithm is $O(\mu^{k+2}) = O(\mu^{O(1)})$, since $k$ is a constant.
%
%\qed
\end{proof}

\end{document}